\documentclass[12pt]{article}

\usepackage{amsfonts}
\usepackage{amssymb}
\usepackage{graphicx}
\usepackage{amsmath}
\usepackage{amsthm}
\usepackage[bottom]{footmisc}
\usepackage{xcolor}
\usepackage{natbib}
\usepackage{booktabs}
\usepackage{threeparttable}
\usepackage{siunitx}
\usepackage{multirow}
\usepackage{rotating}
\usepackage{url}
\usepackage{afterpage}
\usepackage{setspace}
\usepackage{cancel}
\usepackage{mathtools}

\usepackage[margin=1in]{geometry}

\newcommand{\bc}{\begin{center}}
	\newcommand{\ec}{\end{center}}
\newcommand{\be}{\begin{equation}}
\newcommand{\ee}{\end{equation}}
\newcommand{\bea}{\begin{eqnarray}}
\newcommand{\eea}{\end{eqnarray}}
\newcommand{\bean}{\begin{eqnarray*}}
	\newcommand{\eean}{\end{eqnarray*}}
\newcommand{\bt}{\begin{tabular}}
	\newcommand{\et}{\end{tabular}}

\everymath{\displaystyle}

\newcounter{pkt}

\newtheorem{theorem}{Theorem}

\newtheorem{assumption}{Assumption}

\newtheorem{corollary}{Corollary}

\newtheorem{definition}{Definition}

\newtheorem{lemma}{Lemma}

\newtheorem{remark}{Remark}

\numberwithin{theorem}{section}
\numberwithin{lemma}{section}
\numberwithin{assumption}{section}
\numberwithin{corollary}{section}
\numberwithin{equation}{section}

\newcommand{\argmin}{\operatorname*{argmin}}

\newcommand{\tr}{\operatorname*{tr}} 
\newcommand*{\bs}[1]{ \boldsymbol{#1} }

\everymath{\displaystyle}

\newcounter{saveeqn}

\usepackage[hidelinks,colorlinks=true,citecolor=blue, linkcolor=blue,urlcolor=blue]{hyperref}

\begin{document}
	
\title{\vspace*{-2cm} \bf Optimal Estimation of Two-Way Effects under Limited Mobility\thanks{We thank participants at various conferences and seminars for helpful comments. Cheng acknowledges financial support from the Jacobs Levy Center. Schorfheide acknowledges financial support from the National Science Foundation under Grant SES 1851634.}}
	
\author{Xu Cheng\thanks{Department of Economics, University of Pennsylvania. E-mail:   \href{mailto:xucheng@upenn.edu}{\nolinkurl{xucheng@upenn.edu}}.} 
	\and Sheng Chao Ho\thanks{School of Economics, Singapore Management University. Email: \href{scho@smu.edu.sg} {\nolinkurl{scho@smu.edu.sg}}.}  
\and Frank Schorfheide\thanks{Department of Economics, University of Pennsylvania. E-mail: \href{mailto:schorf@upenn.edu}{\nolinkurl{schorf@upenn.edu}}.}
}

\date{This Version: \today}
\maketitle

\begin{abstract}
	We propose an empirical Bayes estimator for two-way effects in linked data sets based on a novel prior that leverages patterns of assortative matching observed in the data. To capture limited mobility we model the bipartite graph associated with the matched data in an asymptotic framework where its Laplacian matrix has small eigenvalues that converge to zero. The prior hyperparameters that control the shrinkage are determined by minimizing an unbiased risk estimate. We show the proposed empirical Bayes estimator is asymptotically optimal in compound loss, despite the weak connectivity of the bipartite graph and the potential misspecification of the prior. We estimate teacher values-added from a linked North Carolina Education Research Data Center student-teacher data set.
\end{abstract}

\noindent JEL CLASSIFICATION: C11, C13, C23, C55, I21
	
\noindent KEY\ WORDS: Bipartite Graphs, Compound Loss, Empirical Bayes Methods, Limited Mobility, Matched Data, Shrinkage Estimation, Teacher Value-Added, Two-Way Fixed Effects, Unbiased Risk Estimation, Weak Identification.
	
\thispagestyle{empty}
\setcounter{page}{0}
\newpage

\section{ Introduction}
\label{sec:intro}

With the availability of large-scale linked datasets,  researchers increasingly estimate high-dimensional two-way effects in interaction-based models following \cite{Abowd1999}, see e.g.,  \cite{Card2013} for worker and firm effects, \cite{Kramarz2008} for student and school effects, and \cite{Finkelstein2016} for patient and area effects, among many others. 
Identification of unit-specific parameters in a linear two-way model crucially depends on forming a connected set through movements \citep{AbowdEtal2002}, e.g., workers who move from one firm to another. Without any workers that connect two firms directly or indirectly, the difference in their firm effects cannot be identified. Similarly, teachers who move from one school to another establish the links to common students that are needed to identify their teacher value-added. In empirical data sets mobility, however, is often limited, which causes econometric challenges.  Estimation becomes fragile if removal of a small number of movers could lead to disconnected sets and therefore the loss of identification. Based on graph theory, \cite{Jochmans2019} show how weak connectivity of the bipartite graph induced by the matched data adversely affects the commonly used least squares (LS) estimate. They also document weak connectivity among teachers across schools, suggesting the difficulty with teacher value-added estimation in the two-way model.

This paper takes on the challenge of limited mobility and proposes an optimal estimation of high-dimensional effects in a decision-theoretic framework. Specifically, we develop an empirical Bayes (EB) estimator based on a Gaussian prior distribution for the two sets of unit-specific parameters, which are indexed by hyperparameters that control mean, variance, and covariance. We treat the resulting posterior mean vector as a class of shrinkage estimators. Shrinkage estimators introduce bias in exchange for a reduction in variance, which is particularly appealing when the variance is large due to limited mobility. The shrinkage estimators are evaluated under a compound loss function that averages squared estimation errors across the large number of units. An unbiased estimate of the expected compound loss (risk) is minimized with respect to the hyperparameters to obtain an empirical Bayes estimator of the unit-specific parameters.

Our first contribution is to provide a novel prior on the joint distribution of the two sets of heterogeneous parameters by incorporating the observed matching patterns in the data. Assortative matching is an important feature in this interaction-based model, i.e., productive firms tend to employ high-skilled workers; see \cite{KLINE2024} for a review.  Empirical Bayes methods leverage  the common prior to improve estimation of unit-specific parameters. As such, it is important to have a prior that allows for stronger correlation of two unit-specific parameters associated with pairs that match more often in the dataset.\footnote{See \cite{Bonhomme2020} for a discussion of the conditionally independent random-effect model and the difficulty to specify conditional random-effects distributions consistent with link formation. \cite{Bonhomme2024} survey the random effect approach to network regression and its applications. } Our prior achieves this in a parsimonious fashion by introducing a hyperparameter that controls the scale of correlation and coupling it with  the normalized matched frequency in the data. The flexibility to choose this correlation hyperparameter significantly enlarges the class of shrinkage estimator we consider. It yields remarkably improved estimates by pulling information across two sides of the interaction, even when the researcher is only interested in one of the two unit-specific parameters, e.g., firm effects or teacher-value added. 

Our second important contribution is to establish frequentist asymptotic optimality in an asymptotic framework that allows for weak connectivity. The benchmark is an  oracle estimator that determines the hyperparameters based on knowledge of the true parameters.
We show that, conditional on the graph induced by the matched data and the true parameters, the compound loss differential between the proposed empirical Bayes estimator based on the unbiased risk estimate (URE) and the oracle estimator goes to zero asymptotically. 
Although the optimality is defined only for the unit specific parameters and their linear functions, the estimator provides a stepping stone for their nonlinear functions and subsequent regression analysis where these unit specific parameters are used as dependent variables.
The frequentist optimality is robust to misspecification of the prior described above, i.e., when the prior does not align with the distribution of parameters across units conditional on the observed matches. The prior simply determines the class of shrinkage estimators considered. The optimality also does not depend on the Gaussian distribution of the regression errors. 

Crucially, the asymptotic analysis is conducted in a framework that allows for weak connectivity, 
a scenario where the connectivity measure is small in finite-sample and is modeled to converge to zero asymptotically akin to weak instruments in instrumental variable regressions.
We consider a sequence of graphs where the number of matches is bounded for each unit as the sample sizes of both sides go to infinity proportionally. As in \cite{Jochmans2019}, we consider non-random graphs and use the small non-zero eigenvalues of the Laplacian matrix to measure its (global) connectivity. 
For asymptotic optimality of the URE based EB estimator, we derive conditions under which the URE criterion uniformly converges to the loss function over a large parameter space of the hyperparameters. In a setup where we only allow a finite-number of near-zero eigenvalues of the Laplacian matrix, this uniform convergence holds as long as each eigenvalue converges to zero more slowly than the inverse of the square root of the sample size.\footnote{This rate is associated with the parameter space of the hyperparameter and therefore the scope of the shrinkage estimator class where optimality is defined. If we confine the variance components of the hyperparameter to some bounded set, we can improve this rate from square root of the sample size to the sample size. }
The asymptotic optimality is robust when all eigenvalues are bounded away from zero for well connected graphs.

Third, we show that our class of shrinkage estimators includes many commonly used estimators as special cases, and document in a Monte Carlo study that the proposed EB estimator dominates these competitors. By setting hyperparameters to specific values, we obtain not only the LS estimate of two-way effects, but also some one-way EB estimators widely used in the study of value-added in labor economics, see e.g., \cite{Kane2008}, \cite{Chetty2014}, and the review in \cite{WALTERS2024}. 
With assortative matching and unobserved heterogeneity on both sides, the proposed estimator shows substantial improvement over the one-way shrinkage estimators.
The class also includes the conventional EB estimator
that determines the hyperparameters based on marginal likelihood of the data.
Compared to the conventional EB estimator, the proposed URE-based EB estimator is less sensitive to the specification of the prior and demonstrate sizable improvement when the prior is misspecified to a large degree.

Fourth, we apply the proposed estimator in an empirical application based on a matched student-teacher data set from the North Carolina Education Research Data Center (NCERDC). Our estimators suggest that there is significant student-teacher assortative matching, and further that the rankings of teachers are sensitive to the estimator employed. 
Besides this specification, our method also apply to studies with assortative matching between students and schools and between teachers and schools, e.g., \cite{Kramarz2008} and \cite{Mansfield2015}.

Shrinkage estimation has a long tradition in the statistics literature, dating back to \cite{Stein1956} and \cite{James1961} in the context of a vector of unknown means. The empirical Bayes approach can be traced back to \cite{Robbins1956}. The recent statistics literature, e.g., \cite{Brown2009} and \cite{Jiang2009}, and the recent panel data literature, e.g., \cite{Gu2017a,Gu2017b}, \cite{Liu2020} have emphasized a nonparametric treatment of the prior distribution. Setting up a nonparametric prior for the joint distribution of two types of unit-specific parameters in the presence of matching pattern is an open question. The parametric prior we propose is a step in this direction. 

Hyperparameter choice  based on URE dates back to \cite{Stein1981} and has been used more recently, for instance, in \cite{Xie2012} and \cite{Kwon2021} for hyperparameter determination in models with fixed and time-varying one-dimensional heterogeneity, respectively, and in \cite{Brown2018} for a model with two-dimensional heterogeneity. The setup in \cite{Brown2018} is similar to ours but they are interested in optimal estimation of the cell means, which is the sum of two unit-specific parameters associated with the two sides, rather than the unit-specific parameters themselves. Our results do not follow from theirs directly.  Furthermore, although they allow for missing cells, the amount of missing cells allowed is not suitable for the sparsely matched data we consider.\footnote{To measure the imbalance of the design matrix due to missing cells, \cite{Brown2018} use a diverging eigenvalue comparable to the inverse of the smallest non-zero eigenvalues of our Laplacian matrix. They require this divergence is slower than  $n^{1/8}/(log(n))^2$, where $n$ denotes the sample size. This rate is suitable for occasional missing data rather than sparsely matched data.}

Our paper builds on some recent developments on the estimation of two-way effects with sparsely matched data. 
We adopt the bipartite graph representation and the measure of graph connectivity proposed by \cite{Jochmans2019}. Our results focus on the compound decision that involves high-dimensional parameters rather than
the estimation of an individual parameter, and we extend the study from the LS estimator to a large class of shrinkage estimator.  
\cite{Verdier2020} views the two-way effects as nuisance parameters and provides inference for the homogeneous slope coefficients  in the context of sparsely matched data. We suggest using his estimator to concentrate out observed covariates before our shrinkage estimator is constructed. 

In a contemporaneous paper, \cite{HeRobin2025} independently propose to use a ridge estimator for high-dimension two-way fixed effects estimation. Ridge estimators belong to the class of shrinkage estimator we consider.\footnote{We can obtain the ridge estimator by setting the correlation parameter and the location parameter in the hyperparameters to zero.} Although both papers find shrinkage estimation is highly desirable, our analysis focuses on different objects and studies them in a distinct asymptotic setups. 
\cite{HeRobin2025} focus on a regularized estimator of the 
Laplacian matrix, its inverse, and the bias and variance of the fixed effects estimator. In contrast, we focus on estimation of the fixed effects themselves (or its subvector). Their asymptotic analysis is based on a degree-corrected stochastic block model and closes the gap between the  aforementioned estimators and their counterparts in a deterministic equivalent graph. 
Our asymptotic analysis conditions on the graph, hence it is deterministic, and closes the gap between a feasible decision rule and the optimal but infeasible decision rule to choose the regularization parameter.
Through these complementary analyses, these two papers offer different perspectives to justify the adoption of  shrinkage methods in two-way effect models.

Our empirical investigation complements the estimation of teacher value-added in the education economics literature see, for instance, \cite{Kane2008}, \cite{Chetty2014}, and \cite{Kwon2021} with various parametric EB methods and
\cite{Gilraine2020} with a nonparametric EB methods. 
Our two-way EB estimator is robust to unobserved student heterogeneity not fully controlled by past test scores and other observable and assortative matching between students and teachers.\footnote{See \cite{Graham2008} for discussions of matching and sorting in the study of student achievement in different classrooms.}
After obtaining the EB estimator, we also use it to compute various nonlinear functions of the high-dimensional effects. Some alternative methods are available in the literature, e.g., see \cite{Kline2020} for unbiased estimation of the variance and covariance, \cite{Armstrong2022} for EB confidence intervals, and \cite{Gu2023} for ranking and selection with EB methods

The remainder of the paper is organized as follows. The model and the LS estimator are presented in Section~\ref{sec:model}. Section~\ref{sec:empbayes} introduces our novel prior distribution, derives the posterior mean estimator of the unit-level effects, and discusses the URE based hyperparameter determination. The optimality theory for the proposed EB estimator is developed in Section~\ref{sec:theory}. Results from our Monte Carlo experiments are summarized in Section~\ref{sec:simul} and the empirical analysis is presented in Section~\ref{sec:appl}. Finally, Section~\ref{sec:concl} concludes. Proofs of the main results appear in the Appendix to this paper. Additional derivations, computational details, simulation results, and further information about the empirical analysis are relegated to an Supplementary Online Appendix.

\section{Two-Way Effects Regression Model}
\label{sec:model}

We consider a two-way fixed effects model for matched data
\begin{align}
	y_{it}^* & = \alpha_i + \beta_{j(i,t)} + \bs{x}_{it}^\prime\bs{\gamma} + u_{it},
	\label{eq:underlyingmodel}
\end{align}
where $i=1,\dots r$, $t = 1,\dots,T$, $j(i,t) \in \left\{1,\dots,c\right\}$ is the unit matched to $i$ in period $t$, $\alpha_i$ and $\beta_{j(i,t)} $ are fixed effects, $\bs{x}_{it}\in\mathbb{R}^k$ is a set of observed regressors that could vary with $i$, $t$, and $u_{it}$ is an exogenous error with mean zero and known variance $\sigma^{2}$, and is i.i.d. across $i$ and $t$. When $\sigma^{2}$ is unknown, we can plug in an unbiased and consistent estimator. We discuss an extension to the heteroskedastic case when introducing the unbiased risk estimate. For the asymptotic optimality of the URE-based estimator, we focus on the  homoskedastic case only.

In an application with student-teacher matched data, $y_{it}$ could be a student test score, which is a function of student ability $\alpha_i$ and teacher value-added $\beta_{j(i,t)}$, as well as other student and teacher characteristics. Here the function $j(i,t)$ characterizes the teacher $j$ assigned to student $i$ in period $t$, with the understanding that each student has only one teacher in any given time period. We will explain more carefully below that our econometric analysis conditions on the observed $j(i,t)$. Alternatively, in an application with employer-employee matched data $y_{it}$ could be the wage of employee $i$, $\alpha_i$ is the worker's skill or effort, $\beta_j$ is the firm-specific productivity, and $j(i,t)$ is the firm for which employee $i$ works in period $t$. Importantly, $\alpha_i$ and $\beta_{j(i,t)}$ are allowed to be correlated to account for assortative matching. We consider settings where the sample size of the two sides $r$ and $c$, symbolizing rows and columns of a matrix, both grow asymptotically and the time period $T$ is finite. 

To facilitate the subsequent analysis, we stack the time series for each student in the order from $i=1$ to $r$ and write the model in (\ref{eq:underlyingmodel}) in matrix notation as
\begin{equation}
	\label{eq:underlyingmodelcompact}
	\begin{aligned}
		\bs{Y}^* &= B_1 \bs{\alpha} + B_2 \bs{\beta} + X\bs{\gamma} + \bs{U}  \\
		&= B\bs{\theta} + X\bs{\gamma} 	+ \bs{U},
	\end{aligned}
\end{equation}
where 
$\bs{Y}^*\in \mathbb{R}^{rT}$, $X\in \mathbb{R}^{rT\times k}$, $\bs{U}\in \mathbb{R}^{rT}$,
$\bs{\theta} := (\bs{\alpha}^\prime,\bs{\beta}^\prime)^\prime$ with $\bs{\alpha}:=(\alpha_1,\dots,\alpha_r)^\prime$ and $\bs{\beta}:=(\beta_1,\dots,\beta_c)^\prime$, and $B=[B_1,B_2]$ selects the unit $\alpha_i$ and $\beta_j$ out of the corresponding vectors for the matched outcome.\footnote{Specifically, in the student-teacher example, the entries of the matrix $B_1$ are indicators denoting the student associated with the corresponding outcome $y_{it}^*$ in $\bs{Y}^*$, and entries of the matrix $B_2$ are indicators denoting the teacher matched to student $i$ at time $t$.} The econometrician observes $\bs{Y}^*$, $B$, and $X$. Because $\alpha$ and $\beta$ appear additively in the model, a normalization is required for their separate identification. To this end, we impose the additional restriction $\bs{1}_c^\prime\bs{\beta}=0$, which is conducive to the subsequent analysis on the subvector $\bs{\beta}$. 
 
Our goal is the optimal estimation of the fixed effects $\bs{\theta}$ or its sub-vector instead of the regression coefficients $\bs{\gamma}$.\footnote{General results in the Appendix cover linear functions of $\bs{\theta}$.} To focus on the optimal shrinkage estimator and how it is related to the structure in $B$, we first study the case where $\bs{\gamma}$ is known and consider the simplified model:
\begin{equation}
	\bs{Y} = B\bs{\theta} + \bs{U}, \text{ where } \bs{Y}:=\bs{Y}^* - X\bs{\gamma}.
	\label{eq:underlyingmodelcompact.nogamma}
\end{equation}
We subsequently show in Corollary~\ref{cor:reg} that the optimality results can be extended to the case in which $\bs{\gamma}$ is consistently estimated.

The identification of $\bs{\theta}$ is determined by the link information captured in the matrix $B$, or equivalently the function $j(\cdot,\cdot)$. 
To study this relation in graph theoretic terms, the links are viewed as edges in a bipartite graph between two sets of nodes, $\mathcal{S}:=\left\{1,\dots,r\right\}$ (students or workers) and $\mathcal{T}:=\left\{1,\dots,c\right\}$ (teachers or firms). 
The set of edges between $i$ and $j$ is the set $\mathcal{E}_{ij}:=\left\{t\leq T: \ j(i,t)=j\right\}$. The graph implied by the matrix $B$ is then fully defined as $\mathcal{G}:=\left(\mathcal{S},\mathcal{T},\left\{\mathcal{E}_{ij}\right\}\right)$.
We assume that the graph $\mathcal{G}$ is connected to ensure the identification of $\theta$.\footnote{A pair of nodes $i$ and $j$ is said to be connected if there exists a path between them -- that is, a sequence of unique edges that begins with one of $(i,j)$ and ends with the other. A graph is said to be connected if there exists a path for every pair of distinct vertices.}.   Subsequently, we will be concerned with
weak identification of $\bs{\theta}$ in finite samples when the connectivity is weak.

We consider the compound loss of a high-dimensional parameter estimator $\hat{\bs{\theta}}$. Specifically, the estimates of the unit-specific parameters are evaluated under the conventional quadratic loss function 
\begin{equation}
	l_w(\hat{\bs{\theta}},\bs{\theta})
	:=
	\left[\hat{\bs{\theta}}-\bs{\theta}\right]' W 
	\left[\hat{\bs{\theta}}-\bs{\theta}\right],
	\label{eq:compound.loss}
\end{equation}
where $W$ is a positive semi-definite matrix that serve both the role of selection and normalization. 
 In particular, we consider two leading cases for $W$:
\begin{equation}
	W_{a+b} := \frac{1}{r+c} I_{r+c}
	\quad \text{and} \quad
	W_b := \frac{1}{c}\begin{bmatrix}
		0_{r\times r} & 0_{r\times c} \\
		0_{c \times r} & I_c
	\end{bmatrix}.
	\label{eq:weights}
\end{equation}
In the first case where $W=W_{a+b}$, we are interested in the full vector $\bs{\theta}$. In the second case where $W=W_{b}$, we are interested in the subvector $\bs{\beta}$ only. This weighting matrix appears subsequently in the loss-based evaluation of estimators.

The conventional two-way fixed effects estimator estimates $\bs{\theta}$ by least squares (LS) which requires an inverse of the Laplacian matrix
\begin{equation}
 L:=B'B. 
\end{equation}
The matrix $L$ is singular and the number of zero eigenvalues corresponds to the number of connected components in the graph. Therefore, a generalized inverse of $L$ is employed to obtain the LS estimator. We show in Lemma \ref{lm: LS defn} in the Supplementary Online Appendix that, under the normalization $\bs{1}_c^\prime\bs{\beta}=0$ and the assumption of a single connected component, the unique (restricted) LS estimator is given by
\begin{equation}
		\hat{\bs{\theta}}^{\text{ls}} =  L^- B^\prime \bs{Y},
		\label{eq:ls.estimator}
\end{equation}
where $L^{-}$ is a generalized inverse of $L$ that takes the form
\begin{equation}
		 L^- :=
		\mathcal{R} L^\dagger \mathcal{R}^\prime, \quad
		\mathcal{R} := 
		\begin{bmatrix}
			I_r & \frac{1}{c}1_{r\times c} \\
			0_{c\times r} & I_c-\frac{1}{c}1_{c\times c},
		\end{bmatrix} 
	= \begin{bmatrix} {\cal R}_a \\ {\cal R}_b
	  \end{bmatrix},
	  \label{eq:Lminus.calR}
\end{equation}
and $ L^\dagger$ is the Moore-Penrose inverse of $L=B'B$. Using the partitions of ${\cal R}$, we can express the LS estimators of $\bs{\alpha}$ and $\bs{\beta}$ as
$
\hat{\bs{\alpha}}^{\text{ls}}=\mathcal{R}_a \hat{\bs{\theta}}^{\text{\text{ls}}}$ and  $\hat{\bs{\beta}}^{\text{ls}}=\mathcal{R}_b \hat{\bs{\theta}}^{\text{\text{ls}}}.
$
The rotation matrix $\mathcal{R}$ ensures that the LS estimator satisfies the normalization condition: by construction,
$\bs{1}_c^\prime\hat{\bs{\beta}}^{\text{ls}}=0$.

While $\hat{\bs{\theta}}^{\text{ls}}$ is an unbiased estimator of $\bs{\theta}$, it is generally not optimal under the quadratic compound loss function in (\ref{eq:compound.loss}). In particular, it is associated with a large compound loss when the graph $\mathcal{G}$ is weakly connected. 
Because $\hat{\bs{\theta}}^{\text{ls}}$ relies on the presence of movers, such as teachers/workers moving between different sets of students/firms, for estimation, it is imprecise when there are only a few movers between some isolated clusters,  implying that the removal of a small number of edges is sufficient to separate $\mathcal{G}$ into disconnected components. 
In this case, some non-zero eigenvalues of the Laplacian matrix $L$ are close to zero, which implies that the generalized inverse $L^{\dagger}$ that appears in $\hat{\bs{\theta}}$ has large eigenvalues. Below we study an asymptotically optimal  estimator that minimize the compound loss in an asymptotic framework where a subset of the eigenvalues of $L$ converge to $0$ asymptotically, modeling estimation with a weakly connected graph.

\section{Empirical Bayes with Unbiased Risk Estimate}
\label{sec:empbayes}

To obtain an optimal estimator that minimizes the quadratic compound loss function, we pursue an empirical Bayes strategy. Our starting point is a hierarchical Bayes model with a family of prior distributions indexed by a vector of hyperparameters. Based on the class of priors we derive a posterior mean estimator and then subsequently determine the hyperparameters in a data-driven manner, which leads to the empirical Bayes estimator. 

A generative economic model would start from the vector of unit-specific coefficients $\bs{\theta}$ characterizing agents ``technologies and preferences.'' A matching mechanism that groups students into classes and assigns teachers to classes, or a search model that connects workers and firms, determines the links encoded in the matrix $B$ and generates a conditional distribution $p(B|\bs{\theta})$. Finally, outcomes are determined based on $p(\bs{Y}|B,\bs{\theta})$ which is characterized by (\ref{eq:underlyingmodelcompact.nogamma}). In a Bayesian setting, one specifies a prior distribution for $\bs{\theta}$, which we denote by $p(\bs{\theta})$. This leads to a joint distribution of data and parameters:
\be
p(\bs{Y},B,\bs{\theta}) = p(\bs{Y}|B,\bs{\theta}) p(B|\bs{\theta}) p(\bs{\theta}).
\label{eq:joint.data.para}
\ee
According to Bayes Theorem the posterior distribution of $\bs{\theta}$ given the observables $(\bs{Y},B)$ is
\begin{align}
   p(\bs{\theta}|\bs{Y},B) &\propto p(\bs{Y}|B,\bs{\theta}) p(B|\bs{\theta}) p(\bs{\theta}) \\
   &\propto p(\bs{Y}|B,\bs{\theta}) p(\bs{\theta}|B), \nonumber 
\end{align}
where $\propto$ denotes proportionality. To obtain the second line, we replaced $p(B|\bs{\theta}) p(\bs{\theta})$ by $p(\bs{\theta}|B) p(B)$ and absorbed $p(B)$ into the constant of proportionality because it does not depend on $\theta$. Rather than deriving $p(\bs{\theta}|B)$ from a structural model of link formation $p(B|\bs{\theta})$, we directly specify a prior $p(\bs{\theta}|B,\bs{\lambda})$, where $\bs{\lambda}:=(\mu,\lambda_a,\lambda_b,\phi)$ is a vector of hyperparameters that incorporates the assortative matching patterns reasonably expected in practice. We emphasize that the prior only serves to define the estimator class, and is not assumed to be correctly specified for our subsequent optimality results to hold. 

The prior distribution $p(\bs{\theta}|B,\bs{\lambda})$ is described in Section~\ref{subsec:empbayes.prior} and the mean of the posterior distribution $p(\bs{\theta}|Y,B,\bs{\lambda})$ is derived in Section~\ref{subsec:empbayes.posterior}. Hyperparameter determination is discussed in Section~\ref{subsec:empbayes.hyperparameters} and Section~\ref{subsec:empbayes.altpriormean} shows how to covariates can be used to center the prior for $\bs{\theta}$, which may be useful in applications.

\subsection{A Novel Prior with Assortative Matching}
\label{subsec:empbayes.prior}

We proceed by introducing our proposed prior $p(\bs{\theta}|B,\bs{\lambda})$. To build the observed link information in the prior for $\bs{\theta}$, we first consider the Laplacian matrix $L=B'B$ and its interpretation. Because $B_1$ and $B_2$ are selection matrices by construction, we know that $B_1'B_1$ is a $r \times r$ diagonal matrix whose diagonal elements counts the number of times student $i$ appears in the data, $B_2'B_2$ is a $c \times c$ diagonal matrix whose diagonal elements count the number of times teacher $j$ appears in the data, $B_1'B_2$ is a $r \times c$ matrix whose $(i,j)$ element equals the number of times student $i$ and teacher $j$ are matched over the $T$ time periods. Define
\begin{equation}
D:=\text{diag}(L), \quad A:=L-D, \quad \mathcal{A} :=D^{-1/2}AD^{-1/2},
\end{equation}
which are the degree,  adjacency, and normalized adjacency matrices of $\mathcal{G}$, respectively. 
Because $B_1'B_1$ and $B_2'B_2$ are both diagonal, we have
\begin{eqnarray}
\mathcal{A}:=\begin{bmatrix}
	0_{r \times r} & (B_1'B_1)^{-1/2}B_1'B_2(B_2'B_2)^{-1/2} \\
	(B_2'B_2)^{-1/2}B_2'B_1(B_1'B_1)^{-1/2} & 0_{c \times c}
\end{bmatrix},
\end{eqnarray}
where the $(i,j)$ element of the upper right submatrix $\mathcal{A}_{12}=(B_1'B_1)^{-1/2}B_1'B_2(B_2'B_2)^{-1/2}$ measures the normalized matched frequency between $i$ and $j$. In the proposed prior, we model assortative matching based on the observed matched frequency in $\mathcal{A}$.

Conditional on $B$, we introduce a normal prior for $\bs{\theta}$ that depends on a vector of 
hyperparameters $\bs{\lambda}:=(\mu,\lambda_a,\lambda_b,\phi)$, where $\mu$ models the mean value of $\alpha_i$, $\lambda_a$ and $\lambda_b$ model the precision (inverse of variance) of $\alpha_i$ and $\beta_j$, respectively, and $\phi$ models the correlation between $\alpha_i$ and $\beta_j$ given $\mathcal{A}$. The mean of $\beta_j$ is $0$ in the prior to be consistent with the normalization. Specifically, the 
the proposed prior takes the form
\begin{align}
	\bs{\theta}\mid B,\bs{\lambda}
	&\sim\mathcal{N}
	\left(\bs{v},\left[\Lambda^*\right]^{-1}\sigma^2\right),  \text{\ \ where }
	\bs{v} := \begin{bmatrix} 
		\mu\bs{1}_r \\ \bs{0}_c
	\end{bmatrix}, \quad \Lambda := \begin{bmatrix}
	\lambda_aI_r  & 0_{r\times c} \\
	0_{c\times r}  & \lambda_bI_c
\end{bmatrix},
\label{eq:sorting.prior}
\\
	\Lambda^*&:=\Lambda^{1/2}[-\phi\mathcal{A}+I]\Lambda^{1/2} =\begin{bmatrix}
		\lambda_aI_r  &  -\phi \mathcal{A}_{12} \cdot \lambda_a^{1/2}\lambda_b^{1/2}\\
		-\phi \mathcal{A}_{21}\cdot \lambda_a^{1/2}\lambda_b^{1/2} & \lambda_bI_c
	\end{bmatrix}. 	\nonumber 
\end{align}

In the special case where $\phi=0$, the prior in \eqref{eq:sorting.prior} reduces to
\begin{equation}
	\label{eq:indep.prior}
	\alpha_i|\bs{\lambda} \sim_{iid} \mathcal{N} \left(\mu,\frac{\sigma^2}{\lambda_a} \right), \quad \beta_j | \bs{\lambda} \sim_{iid} \mathcal{N} \left(0,\frac{\sigma^2}{\lambda_b}\right).
\end{equation}
Furthermore, the prior under $\phi=0$ makes the simplifying assumption that $\alpha_i$ and $\beta_j$ are all independent and the prior does not depend on $B$:
\be
p(\bs{\theta}|B,\bs{\lambda}) = p(\bs{\theta}|\bs{\lambda}) = \left(\prod_{i=1}^r p(\alpha_i|\bs{\lambda}) \right) \left(\prod_{j=1}^c p(\beta_j|\bs{\lambda}) \right).
\label{eq:p.theta}
\ee
On the other hand, when $\phi\neq0$, Lemma~\ref{lm:proposedII} shows that the prior incorporates the assortative matching patterns present in the bipartite graph in a parsimonious way  through $\phi$.

\begin{lemma}[Conditional distribution of prior]
	\label{lm:proposedII}
	Given the prior in \eqref{eq:sorting.prior}, the conditional prior distributions are
	\begin{align*}
		\bs{\alpha} \mid B,\bs{\beta}
		&\sim
		\mathcal{N}
		\left(
		\bs{1}_r\mu + \sqrt{\frac{\lambda_b}{\lambda_a}} 
		\cdot \phi\mathcal{A}_{12} \cdot \bs{\beta} 
		\ , \
		\lambda_a^{-1}\sigma^2I_r
		\right),\\
		\bs{\beta} \mid B,\bs{\alpha}
		&\sim
		\mathcal{N}
		\left(
		\bs{0}_c + \sqrt{\frac{\lambda_a}{\lambda_b}}
		\cdot \phi \mathcal{A}_{21} \cdot \left(\bs{\alpha} - \bs{1}_r\mu\right)
		\ , \
		\lambda_b^{-1}\sigma^2I_c
		\right).
	\end{align*}
	Furthermore,
	\begin{equation*}
		\text{corr}(\alpha_i,\beta_j\mid B,\bs{\alpha}_{-i},\bs{\beta}_{-j}) 
		=
		\phi \cdot
		\mathcal{A}_{12,ij},
		\label{eq:corr}
	\end{equation*}
where $\mathcal{A}_{12,ij}$ denotes the $i^{th}$ row and $j^{th}$ column of the matrix $\mathcal{A}_{12}$.
\end{lemma}

Lemma \ref{lm:proposedII} shows that the correlation between $\alpha_i$ and $\beta_j$ is entirely determined by the parameter $\phi$ and their observed matched frequency in the normalized adjacency matrix $\mathcal{A}$. A positive value of $\phi$ implies positive assortative matching. For example, a higher value of $\beta_j$, i.e., a teacher with higher value-added or a firm with higher productivity, increases the conditional mean of $\alpha_i$ connected to it, implying on average better students taught by teacher $j$ or more productive workers in firm $j$.

For $\phi \not=0$ our prior \eqref{eq:sorting.prior} avoids the joint independence assumption of \eqref{eq:p.theta} and adopts only a conditional independence assumption:
\begin{equation}
	p(\bs{\alpha}|B,\bs{\lambda},\bs{\beta}) = \prod_{i=1}^r p(\alpha_i|B,\bs{\lambda},\bs{\beta}) 
	\label{eq:prior.extend}
\end{equation}
with a correspondence for $p(\bs{\beta}|B,\bs{\lambda},\bs{\alpha})$. This bears similarity to the one-sided correlated random effects specification employed by \cite{Bonhomme2019}.

\subsection{Posterior Mean of Bayes Estimator}
\label{subsec:empbayes.posterior}

Under the assumption that $u_{it} | (B,\bs{\theta}) \sim_{iid} \mathcal{N}(0,\sigma^2)$, the posterior mean under the prior in (\ref{eq:sorting.prior}) is given by 
\begin{equation}
		\hat{\bs{\theta}}(\bs{\lambda}) :=
		S_1\hat{\bs{\theta}}^{\text{ls}} + 
		S \bs{v},
		\label{eq:Bayes.estimator}
\end{equation}
where 
\begin{equation}
S := \mathcal{R} [L + \Lambda^*]^{-1} \Lambda^* 
\text{\ \ and\ \ } S_1:=\mathcal{R}-S. 
\end{equation}
Due to the normalization $\bs{1}_c^\prime\bs{\beta}=0$, this posterior mean formula involves the rotation matrix $\mathcal{R}$, previously defined in (\ref{eq:Lminus.calR}).
A derivation is provided in Lemma~\ref{lm:pstm} in the  Supplementary Online Appendix. 

The hyperparameter $\bs{\lambda}$ is chosen over the set 
	$
	\bs{\lambda} =(\mu, \lambda_a, \lambda_b, \phi)\in
	\bar{\mathcal{J}} := 
	[-\bar{\mu},\bar{\mu}] \times 
	[0,\infty] \times [0,\infty] \times
	[\-\bar{\phi},\bar{\phi}]
	$
	for some finite $\bar{\mu}$ and $0<\bar{\phi}<1$.\footnote{In many applications bounds for $\mu$ can be obtained from bounds on the outcome variable, e.g., test scores or wages. Shrinking to a location outside these bounds does not further reduce the loss. In contrast, we do not impose upper and lower bounds for $\lambda_a$ and $\lambda_b$ in order to consider optimal estimation in a wider class. As we show below, setting $\lambda_a$ or $\lambda_b$ to $0$ or $\infty$ leads to some commonly used one-way fixed effect estimators.}
The hyperparameter $\bs{\lambda}$ affects the posterior mean through $S$ and $\bs{v}$, which determine the shrinkage weight and the shrinkage location, respectively. 
Specifically, $S$ depends on $(\lambda_a,\lambda_b,\phi)$ and $\bs{v}$ depends on $\mu$. When $\lambda_a=0$ and $\lambda_b=0$, the posterior mean is the LS estimator. Moreover, when $\lambda_a \rightarrow \infty$ and $\lambda_b \rightarrow \infty$, the posterior mean becomes the prior mean. We will regard 
$\{\hat{\bs{\theta}}(\bs{\lambda}) \, | \, \lambda \in {\bar{\mathcal{J}} } \}$ 
as a class of shrinkage estimators and study the problem of optimally choosing an estimator within this class.

Given a data-driven choice of the hyperparameter $\hat{\bs{\lambda}}$, the resulting estimator $\hat{\bs{\theta}}(\hat{\bs{\lambda}})$ is an empirical Bayes (EB) estimator; 
see \cite{Robbins1956} or, for instance, \cite{Robert1994} for a textbook treatment. 
Although this class of shrinkage estimator is motivated by the Gaussian prior and the Gaussian regression error, the the EB estimator proposed below is asymptotically optimal within this large class of estimators even when these Gaussian distributions are misspecified. 


\subsection{Hyperparameter Determination}
\label{subsec:empbayes.hyperparameters}

\noindent {\bf Unbiased Risk Estimation.} We consider an EB based estimator that chooses the hyperparameter $\bs{\lambda}$ by minimizing an unbiased estimate of the risk. To define the estimation risk, we consider a frequentist approach that conditions on the unit-specific effects $\bs{\theta}$. Because we deliberately abstracted from a theory of how matches are formed, we also condition on the observed matrix $B$. In applications in which the matching mechanism is deterministic and $p(B|\bs{\theta})$ is a point mass, there would be no other choice than to condition on $B$. For a given hyperparameter $\bs{\lambda}$,  we define the risk as
\be
	R_w(\hat{\bs{\theta}}(\bs{\lambda}),\bs{\theta}):=\mathbb{E}_{\bs{\theta},B}[l_w(\hat{\bs{\theta}}(\bs{\lambda}),\bs{\theta})].
\ee 
The expectation is taken with respect to $\bs{Y}$ and the subscript indicates that we condition on $(\bs{\theta},B)$.

We now construct an unbiased estimate of $R_w(\hat{\bs{\theta}}(\bs{\lambda}),\bs{\theta})$, denoted by $\text{URE}(\bs{\lambda})$, that only depends on the observations $(\bs{Y},B)$ and not on the unknown parameter vector $\bs{\theta}$. Using the formula for $\hat{\bs{\theta}}(\bs{\lambda})$ in \eqref{eq:Bayes.estimator}, the risk function can be written as
\begin{equation}
	R_w(\hat{\bs{\theta}}(\bs{\lambda}),\bs{\theta})
	= \tr \big[ S'WS (\bs{\theta} - \bs{v})(\bs{\theta} - \bs{v})' \big] + \sigma^2\tr[ S_1'WS_1 L^- ].\label{eq:ure.1}
\end{equation}
Under the quadratic loss, the second term on the right-hand side of \eqref{eq:ure.1} comes from the variance of $S_1\hat{\bs{\theta}}^{\text{ls}}$, which is known, and the first term is a quadratic function of the bias induced by the prior. A natural unbiased estimator of the first term is 
\begin{equation}\label{eq:quadmean_unbiased}
\tr \big[ S'WS (\hat{\bs{\theta}}^{\text{ls}} - \bs{v})(\hat{\bs{\theta}}^{\text{ls}} - \bs{v})' \big] 
		-
		\sigma^2\tr[ S'WS L^- ],
\end{equation}
where $\sigma^2\tr[ S'WS L^- ]$ is used for bias correction because 
	\begin{equation}
		\mathbb{E}_{\bs{\theta},B}\left[
		\tr \big[ S'WS (\hat{\bs{\theta}}^{\text{ls}} - \bs{v})(\hat{\bs{\theta}}^{\text{ls}} - \bs{v})' \big] \right]
		= \tr \big[ S'WS (\bs{\theta} - \bs{v})(\bs{\theta} - \bs{v})' \big] +
		\sigma^2\tr[ S'WS L^- ] .	\label{eq:ure.2}
	\end{equation}
Combining \eqref{eq:ure.1} and \eqref{eq:ure.2} , we obtain 
\begin{eqnarray}
	\text{URE}(\bs{\lambda}) := \tr \big[ S'WS (\hat{\bs{\theta}}^{\text{ls}} - \bs{v})(\hat{\bs{\theta}}^{\text{ls}} - \bs{v})' \big]- \sigma^2\tr[S^\prime W S  L^- ] + \sigma^2\tr[S_1^\prime W S_1  L^- ]   .\label{eq:ure} 
\end{eqnarray}
The unbiasedness of $\text{URE}(\bs{\lambda})$ is summarized in the following Lemma. 


\begin{lemma}[Unbiased Risk Estimate]
	\label{lem:ure}
	The risk estimate in (\ref{eq:ure}) satisfies 
	\begin{equation*}
		\mathbb{E}_{\bs{\theta},B}[\text{URE}(\bs{\lambda})]
		=
		R_w(\hat{\bs{\theta}}(\bs{\lambda}),\bs{\theta}).
	\end{equation*}
\end{lemma}


We denote the URE hyperparameter estimate by
\begin{equation}
	\bs{\lambda}^{\text{ure}} 
	:=
	\argmin_{\bs{\lambda}\in \mathcal{J} } \;
	\text{URE}(\bs{\lambda}),
\end{equation}
where $
\bs{\lambda} =(\mu, \lambda_a, \lambda_b, \phi)\in
\mathcal{J}:= 
[-\bar{\mu},\bar{\mu}] \times 
(0,\infty) \times (0,\infty) \times
[\-\bar{\phi},\bar{\phi}]
$.\footnote{Here we rule out $0$ and $\infty$ from the parameter space of $\lambda_a$ and $\lambda_b$ to facilitate showing uniform convergence of some functions over $\mathcal{J}$, which is needed for the asymptotic analysis below. Importantly, we do allow $\lambda_a$ and $\lambda_b$ to be arbitrarily close to the boundaries.}
The proposed estimator is $\hat{\bs{\theta}}(\bs{\lambda}^{\text{ure}})$.
Justification for $\hat{\bs{\theta}}(\bs{\lambda}^{\text{ure}})$ goes beyond the unbiased property in Lemma \ref{lem:ure}. In the next section, we show that 
$\hat{\bs{\theta}}(\bs{\lambda}^{\text{ure}})$ is the asymptotic optimal estimator in a decision theoretic framework that allows for weakly connected graphs.

\noindent {\bf Treatment of $\sigma^2$.} To implement the hyperparameter determination one can replace  $\sigma^2$ with its LS squares estimate 
\begin{equation}\label{eq:sigma_est}
	\hat{\sigma}^2=\frac{1}{rT-(r+c-1)}\bs{Y}'[I_{rT}-B(B'B)^{-}B']\bs{Y}.
\end{equation}
Lemma~\ref{lm:LSvar} in the Supplementary Online Appendix shows that this estimator is unbiased, and it is also consistent under weak conditions. A practitioner might be concerned about heteroskedasticity. In this case $\sigma^2\tr[ S'WS L^- ]$  and  $\sigma^2\tr[ S_1'WS_1 L^- ]$ in  \eqref{eq:ure} can be, respectively, replaced with the leave-out estimates of 
$\tr\big[ S'WS\mathbb{V}_{\bs{\theta}, B}[\hat{\bs{\theta}}^{ \text{ls}}] \big]$ and $\tr \big[ S_1'WS_1 \mathbb{V}_{\bs{\theta}, B}[\hat{\bs{\theta}}^{\text{ls}}] \big]$ provided by \cite{Kline2020}.\footnote{We thank Patrick Kline for this insightful suggestion.} This leads to a heteroskedastic-robust $\text{URE}(\bs{\lambda})$. In addition, one could consider shrinkage estimation for $\sigma^2_i$ as well, as it is commonly done in full Bayesian implementations of one-way effect models, e.g., \cite{Liu2023} and \cite{Liu2023a}. In the remainder of this paper we will establish asymptotic optimality within the homoskedastic framework. Extensions to the heteroskedastic case are left for future research.

\noindent {\bf Marginal Likelihood.} Alternatively, one could follow the tradition in the EB literature and estimate $\bs{\lambda}$ by maximizing the marginal likelihood function
\be 
p(\bs{Y}|B,\bs{\lambda}) = \int p(\bs{Y}|\bs{\theta},B) p(\bs{\theta}|B,\bs{\lambda}) d \bs{\theta}, 
\ee
which leads to 
\be
\bs{\lambda}^{\text{mle}} 
:=
\argmin_{\bs{\lambda}\in \mathcal{J} } \;
p(Y|B,\bs{\lambda})
\ee
and $\hat{\bs{\theta}}(\bs{\lambda}^{\text{mle}})$.
We expect $\hat{\bs{\theta}}(\bs{\lambda}^{\text{mle}})$ to perform well whenever the prior \eqref{eq:prior.extend} is correctly specified. If $\bs{\theta}$ is generated from a different distribution than \eqref{eq:prior.extend}, then the marginal likelihood will be misspecified and $\hat{\bs{\theta}}(\bs{\lambda}^{\text{mle}})$ is likely to be sub-optimal. 
Comparisons in the Monte Carlo studies in Section~\ref{sec:simul} confirm that $\hat{\bs{\theta}}(\bs{\lambda}^{\text{mle}})$ is more sensitive to the prior distribution than $\hat{\bs{\theta}}(\bs{\lambda}^{\text{ure}})$.

\subsection{Covariate-Assisted Shrinkage}
\label{subsec:empbayes.altpriormean}

In practice it might be useful to shrink the fixed effects to an index constructed from observable covariates that are informative about the latent heterogeneity. 
	Define $Z_a\in\mathbb{R}^{r\times k_a}$, $Z_b\in\mathbb{R}^{c\times k_b}$, $\bs{\delta}_a\in\mathbb{R}^{k_a}$ and $\bs{\delta}_b\in\mathbb{R}^{k_b}$, where $k_a$ and $k_b$ are the dimensions of the regressors used to construct covariate-based prior means for the $\alpha_i$s and $\beta_i$s. 
	For student-teacher data sets, examples of potential $Z_a$ and $Z_b$ include student demographic information, such as gender, ethnicity, and teacher information such as years of teaching experience, respectively.
	Define
	\begin{equation}
		Z:=\begin{bmatrix}
			Z_a & 0_{r\times k_b} \\
			0_{c\times k_a}& Z_b
		\end{bmatrix}, \ \ 
		\bs{\delta}:=\begin{bmatrix}
			\bs{\delta}_a   \\
			\bs{\delta}_b 
		\end{bmatrix},  \ \ 
		Z\bs{\delta}=\begin{bmatrix}
			Z_a\bs{\delta}_a   \\
			Z_b\bs{\delta}_b 
		\end{bmatrix}. 
	\end{equation}
	 The location of shrinkage for $\bs{\theta}$ is now $Z\bs{\delta}$, where $\bs{\delta}$ is determined jointly with the other hyperparameters in minimizing the URE. This is operationalized by replacing $\bs{v}$ of \eqref{eq:sorting.prior}, \eqref{eq:Bayes.estimator} and \eqref{eq:ure} by $Z\bs{\delta}$ and then performing the minimization as before -- now minimizing the URE over $\bs{\delta}$ as well. Section~\ref{subsec:compute.ure} of the Supplementary Online Appendix shows how to concentrate out $\bs{\delta}$ when performing the URE-minimization.

\section{Asymptotic Optimality}
\label{sec:theory}

We next provide some asymptotic justifications for the empirical Bayes estimator $\hat{\bs{\theta}}(\bs{\lambda}^{\text{ure}})$. We study a sequence of growing graphs $\left\{\mathcal{G}_{r,c}\right\}$ as $r,c \rightarrow \infty $, where we use $r$ and $c$ in the subscript to indicate that the sizes of the nodes as well as other characteristics, such as the connectivity measures, could change with the sample sizes. The number of periods $T$ is fixed in our asymptotic analysis. For notational simplicity,  we omit the dependence on $r, c$ in other graph characteristics such as $B$ and the Laplacian matrix $L$. Section~\ref{subsec:theory.covergence} discusses our assumption and the convergence of the URE objective function. The optimality result is stated in Section~\ref{subsec:theory.optimality}. A comparison of the proposed two-way shrinkage estimator to other popular estimators is presented in Section~\ref{subsec:theory.alternatives} and the extension to the case of additional estimated regression coefficients is provided in Section~\ref{subsec:theory.estimated.gamma}.

\subsection{Convergence of URE Objective Function}
\label{subsec:theory.covergence}

We first state some regularity assumptions on the regression model. Throughout all the assumptions and proofs, we use $\epsilon$ and $M$ to represent generic constants for lower and upper bounds. They do not have to take the same value when they appear in different contexts.

\begin{assumption}[Regression Model]\label{ass:reg.cond}
For some finite constant $M$, as $r,c \rightarrow \infty$,
\newline
\emph{(i).}  $\tfrac{1}{r+c} \bs{\theta}^\prime\bs{\theta} \le M$.
\newline
\emph{(ii).} $\mathbb{E}[u_{it}^4] \le M $. 
\end{assumption}

Condition (i) restricts the heterogeneity of fixed effects, and is a condition widely used in the empirical Bayes literature, see e.g., \cite{Xie2012} and \cite{Brown2018}.
Condition (ii) imposes only a finite fourth moment on the regression error rather than requiring its normal distribution as often seen in these papers. The finite fourth moment is used to bound the variance of some quadratic forms of the regression errors. Next, we consider some conditions on the sequence of graphs $\{\mathcal{G}_{r,c}\}$.

\begin{assumption}[Size and Degree of Graph] \label{ass.degree}
	\hspace{1em}
	\newline 
	For some finite constant $\epsilon$ and $M$, as $r,c \rightarrow \infty$,
	\newline
	\emph{(i).} $\epsilon \le r/c \le M$,
	\newline
	\emph{(ii).} $L_{ii} \le M$ for all $i=1,\ldots,r+c$, where $L_{ii}$ denote the $i^{th}$ diagonal element of $L$.
\end{assumption}

Condition (i) requires the sizes of the two sets of nodes $\mathcal{S}$ and $\mathcal{T}$ to be proportional and grow asymptotically at the same order. Condition (ii) restricts each node to having a finite degree, characterizing its finite number of matches. In a student-teacher application both conditions could be satisfied by assuming that each student only takes one class per time period, the number of time periods $T$ is fixed, and the class-size stays asymptotically bounded from below and above. In a worker-firm application the conditions allow for a distribution of firm sizes (in terms of number of workers), but rule out asymptotics along which, for instance, one firm employs 50\% of the workers. Under Assumption~\ref{ass.degree} it is not possible to estimate $\bs{\theta}$ consistently.


To model the limited mobility issue in finite sample, we consider a sequence of graphs $\left\{\mathcal{G}_{r,c}\right\}$ whose connectivity gets weaker asymptotically. 
This connectivity measures are the eigenvalues of the Laplacian matrix $L=B'B$ when we are interested in the full vector $\bs{\theta}$. 
When we are interested in the subvector $\bs{\beta}$, the connectivity measures become the eigenvalues of the Laplacian matrix of the one-mode projected graph\footnote{In the one-mode projected graph, the vertices are $\mathcal{T}$ and the nodes from $\mathcal{S}$ moving between them act as the edges. Two nodes from $\mathcal{T}$ are thus connected whenever they are connected to a common node from $\mathcal{S}$ in the original graph $\mathcal{G}$, with the strength of their connection determined by the number of common matches, see \cite{Jochmans2019} for detailed discussions of this one-mode projection graph.} 
\begin{equation}
	L_{2,\perp}:=B_{2,\perp}'B_{2,\perp}, \text{ \ \ where }
	B_{2,\perp} := [I-B_1(B_1^\prime B_1)^{-1}B_1^\prime]B_2.
	\label{eq:def.L2perp}
\end{equation}
For clarity of the graph-theoretic interpretation, we focus on these two leading case in the main paper. Proofs of all theoretical results allow for a general subvector selection and weighting under the high-level conditions in Assumptions~\ref{ass:graph.connect.1} and \ref{ass:graph.connect.2} in the Appendix. The Assumptions below are their sufficient conditions in these two leading cases.

Let $\rho_{\ell}(A)$ denote the $\ell^{th}$  smallest eigenvalue of a symmetric matrix $A$. For a connected graph, it is known that $\rho_1(L)=0$. \cite{Jochmans2019} show $\rho_2(L)$ is crucial in the analysis of the LS estimate. Here, we allow $\rho_2(L)$ as well as some other small eigenvalues to converge to $0$ as $r,c \rightarrow \infty$, in order to study optimal estimation in an asymptotic framework for weakly connected graphs. When the parameter of interest is the subvector $\bs{\beta}$, we allow $\rho_2(L_{2,\perp})$ and some other small eigenvalues of $L_{2,\perp}$
to converge to $0$ asymptotically. 

Below we present some assumptions on these connectivity measures in order to obtain desirable properties of $\text{URE}(\bs{\lambda})$ as a selection criterion for $\bs{\lambda}$ and the proposed EB estimator $\hat{\bs{\theta}}(\bs{\lambda}^{\text{ure}})$. Beyond the unbiasedness in Lemma \ref{lem:ure}, we  show the distance between the  compound loss function $l_w(\hat{\bs{\theta}}(\bs{\lambda}),\bs{\theta})$ and the selection criterion  $\text{URE}(\bs{\lambda})$ is small under some stronger measures, which eventually leads to asymptotic optimality of $\hat{\bs{\theta}}(\bs{\lambda}^{\text{ure}})$. To this end, the conditions typically involve how weak the connectivity can be in order for $\text{URE}(\bs{\lambda})$ to be informative. As such, we impose bounds on how fast $\rho_2(L)$  or $\rho_2(L_{2,\perp})$ can go to zero and how many of these eigenvalues can go to zero. It is important to note that although the empirically relevant asymptotics are obtained when $\rho_2(L)$  or $\rho_2(L_{2,\perp})$ converge to zero, none of our assumptions require that the second smallest eigenvalue vanishes. These desirable properties of $\hat{\bs{\theta}}(\bs{\lambda}^{\text{ure}})$ also hold for well connected graphs where all non-zero eigenvalues are bounded away from zero as the sample size increases.

\begin{assumption}[Graph Connectivity: Full Graph I]
	\label{ass:graph.connect.1.full}
	As $r,c \rightarrow \infty$, $$(r+c)^{1/2}\cdot\rho_2(L) \rightarrow \infty.$$
\end{assumption}

\begin{assumption}[Graph Connectivity: Projected Graph I]
	\label{ass:graph.connect.1.b}
	As $r,c \rightarrow \infty$,
	$$c^{1/2}\cdot\rho_2(L_{2,\perp}) \rightarrow \infty.$$
\end{assumption}

Assumptions~\ref{ass:graph.connect.1.full} and~\ref{ass:graph.connect.1.b} present two versions of this connectivity assumption, one for $W=W_{a+b}$ for the full vector $\bs{\theta}$ and one for $W=W_{b}$ for the subvector $\bs{\beta}$. This set of Assumptions on graph connectivity ensures that not only the compound loss function $l_w(\hat{\bs{\theta}}(\bs{\lambda}),\bs{\theta})$ and the selection criterion  $\text{URE}(\bs{\lambda})$ share the same expectation by construction, but also the variance of their difference converge to zero asymptotically. 

Define 
\begin{equation}
	\delta_{r,c}:= 
\begin{cases} 
	(r+c)^{1/2}\cdot\rho_2(L) &  \text{if  } W=W_{a+b}\\
	c^{1/2}\cdot\rho_2(L_{2,\perp})  &\text{if  } W=W_{b}
\end{cases}
\end{equation}
and note that  Assumptions~\ref{ass:graph.connect.1.full} and~\ref{ass:graph.connect.1.b} imply that $\delta_{r,c} \rightarrow \infty$.

\begin{theorem}
	\label{thm:pw}
	Suppose that Assumptions~\ref{ass:reg.cond} and \ref{ass.degree} hold. Moreover, Assumption~\ref{ass:graph.connect.1.full} holds if $W=W_{a+b}$, and Assumption~\ref{ass:graph.connect.1.b} holds if $W=W_{b}$. Then,
	\begin{equation*}
		\sup_{\bs{\lambda}\in
		\mathcal{J}}
		\left\{
		\delta_{r,c}^2 \cdot
		\mathbb{E}_{\bs{\theta},B} [ l_w(\hat{\bs{\theta}}(\bs{\lambda}),\bs{\theta})- \text{URE}(\bs{\lambda})  ]^2 \right\}
		= O \left(1
		\right).
	\end{equation*}
\end{theorem}

Theorem \ref{thm:pw} shows pointwise convergence between $l_w(\hat{\bs{\theta}}(\bs{\lambda}),\bs{\theta})$ and $\text{URE}(\bs{\lambda}) $ in the $L_2$ norm. It also shows that the convergence depends both on the sample size and the smallest non-zero eigenvalue of the associated graph. 
As long as the smallest non-zero eigenvalue converges to $0$ slower than the square root of the sample size, we have the convergence in the $L_2$ norm. 


Below we impose stronger conditions on the connectivity measure to ensure uniform convergence between $l_w(\hat{\bs{\theta}}(\bs{\lambda}),\bs{\theta})$ and $\text{URE}(\bs{\lambda}) $ over $\bs{\lambda} \in \mathcal{J}$.

\begin{assumption}[Graph Connectivity: Full Graph II]
	\label{ass:graph.connect.2.a}
	There exists a finite integer $k$ and a small constant $\epsilon >0$, such that 
	\begin{eqnarray*}
	\rho_{\ell}(L) \ge \epsilon \quad\text{for}\quad 
 \ell = k+1, \ldots, \leq r+c.
	\end{eqnarray*}
\end{assumption}	

This assumption allows for at most $k$ eigenvalues converge to zero, requiring the full graph to have at most $k$ weakly connected components. This condition is particularly useful to obtain uniform convergence, without imposing stronger conditions on the rate in Assumptions~\ref{ass:graph.connect.1.full} and~\ref{ass:graph.connect.1.b}. In this large-scale estimation problem, $l_w(\hat{\bs{\theta}}(\bs{\lambda}),\bs{\theta})$ and $\text{URE}(\bs{\lambda}) $ both involve trace of matrices whose dimensions grow with the dimension of $\bs{\theta}$. In particular, they involve the inverse of the non-zero eigenvalues of $L$, which are explosive for weakly connected components. Having only a finite number of weakly connected components ensures that we only work with a finite number of explosive ones. Our proof in the Appendix can accommodate a slowly growing number of weakly connect components, if we restrict the weakness of the connectivity. We show this trade off through a general trace condition in  Assumption~\ref{ass:graph.connect.2} in the Appendix. In the empirical application with student-teacher linked data, this finite component assumption requires that students and teachers are generally well connected within each school, and that there are only a finite number of schools across which movements of students and teachers is limited.

Assumption \ref{ass:graph.connect.2.a} is intended for the estimation of the full vector $\bs{\theta}$, i.e., $W=W_{a+b}$. When the interest is on the subvector $\bs{\beta}$, i.e., $W=W_{b}$, we have a similar finite component assumption on the one-mode projected graph. 

\begin{assumption}[Graph Connectivity: Projected Graph II]
	\label{ass:graph.connect.2.b}	
	There exists a finite integer $k$ and a small constant $\epsilon >0$, such that 
	\begin{eqnarray*}
		\rho_{\ell}(L_{2,\perp)} \ge \epsilon \quad\text{for}\quad 
		\ell = k+1, \ldots, \leq r+c.
	\end{eqnarray*}
\end{assumption}

In addition to this requirement, 
uniform convergence between $l_w(\hat{\bs{\theta}}(\bs{\lambda}),\bs{\theta})$ and $\text{URE}(\bs{\lambda}) $ for $W=W_{b}$ also require some minimum assumption on the full graph because we conduct all the estimation simultaneously, despite our interest is on the subvector only.

\begin{assumption}[Graph Connectivity: Full Graph III]
	\label{ass:graph.connect.2.c}	
	There exists $\epsilon > 0$ such that
	\begin{equation*}
		(r+c)\cdot\rho_2^{\epsilon}(L) \rightarrow \infty.
	\end{equation*}
\end{assumption}
This is a very weak condition on the full graph by considering 
$\epsilon$ close to $0$. It is only relevant when we are interested in the subvector. When the full vector is of interested, this condition is already implied by Assumption~\ref{ass:graph.connect.1.full}.

\begin{lemma}[Uniform Convergence] \label{lm:unif.convg} 	
Suppose that Assumptions~\ref{ass:reg.cond} and 	\ref{ass.degree} hold. Furthermore, Assumptions~
	 \ref{ass:graph.connect.1.full} and \ref{ass:graph.connect.2.a} hold if $W=W_{a+b}$, and Assumptions~  \ref{ass:graph.connect.1.b},
	 \ref{ass:graph.connect.2.a}, \ref{ass:graph.connect.2.b}, and
	 \ref{ass:graph.connect.2.c} hold if $W=W_{b}$. 
	Then, as $r,c \rightarrow \infty$,
	\begin{equation*}
		\sup_{\bs{\lambda}\in\mathcal{J}}
		\left|
		l_w(\hat{\bs{\theta}}(\bs{\lambda}),\bs{\theta})
		-
		URE(\bs{\lambda})
		\right| 
		\rightarrow_p
		0.
	\end{equation*}
\end{lemma}

\subsection{Optimality of URE Hyperparameter Determination}
\label{subsec:theory.optimality}

Building on the uniform convergence between $l_w(\hat{\bs{\theta}}(\bs{\lambda}),\bs{\theta})$ and $\text{URE}(\bs{\lambda})$, now we show that $\hat{\bs{\theta}}(\bs{\lambda}^{\text{ure}})$ is asymptotically optimal in the sense that its compound loss is comparable to that of an oracle estimator that choose $\bs{\lambda}$ based on $l_w(\hat{\bs{\theta}}(\bs{\lambda}),\bs{\theta})$ directly. To this end, define
\begin{equation}
	\bs{\lambda}^{\text{ol}}
	:=
	\argmin_{\bs{\lambda} \in J}
	l_w(\hat{\bs{\theta}}(\bs{\lambda}),\bs{\theta}),
\end{equation}
and the oracle estimator oracle $\hat{\bs{\theta}}(\bs{\lambda}^{\text{ol}})$. 
The superscript ``ol'' stands for oracle loss because it is with knowledge of the parameters $\bs{\theta}$. 
Note that this oracle estimator targets the actual compound loss rather than the risk and it achieves the lowest possible compound loss within the class of shrinkage estimator considered. Under parameter heterogeneity, this oracle loss is strictly positive. We say that an estimator is asymptotically optimal if it achieves the oracle loss.

\begin{definition}[Asymptotic Optimality] \label{def:optimality} We say an estimator $\hat{\bs{\theta}}$ is asymptotically optimal if for any $\epsilon > 0$
		\begin{equation*}
		\lim_{r,c\rightarrow \infty}
		\mathbb{E}_{\bs{\theta},B}
		\left [
		l_w(\hat{\bs{\theta}}) 
		\geq
		l_w(\hat{\bs{\theta}}(\bs{\lambda}^{\text{ol}}),\bs{\theta}) 
		+
		\epsilon
		\right ]
		= 0.
	\end{equation*}
\end{definition}

The main theoretical result of the paper is that $\hat{\bs{\theta}}(\bs{\lambda}^{\text{ure}})$ is asymptotically optimal.

\begin{theorem}[URE Optimality] \label{thm:main}
	Suppose that Assumptions~\ref{ass:reg.cond} and 	\ref{ass.degree} hold. 
	Furthermore, Assumptions~
	\ref{ass:graph.connect.1.full} and \ref{ass:graph.connect.2.a} hold if $W=W_{a+b}$, and Assumptions~  \ref{ass:graph.connect.1.b},
	\ref{ass:graph.connect.2.a}, \ref{ass:graph.connect.2.b}, and
	\ref{ass:graph.connect.2.c} hold if $W=W_{b}$. Then, $\hat{\bs{\theta}}(\bs{\lambda}^{\text{ure}})$ is asymptotically optimal in the sense of Definition~\ref{def:optimality}. 
\end{theorem} 

Theorem \ref{thm:main}  follows immediately from the uniform convergence in Lemma~\ref{lm:unif.convg} and $\text{URE}(\bs{\lambda})$ is minimized by $\bs{\lambda}^{ure}$. It shows that the proposed estimator  $\hat{\bs{\theta}}(\bs{\lambda}^{\text{ure}})$ is asymptotically as good as the oracle estimator $\hat{\bs{\theta}}(\bs{\lambda}^{\text{ol}})$ in the sense that they achieve the same compound loss in probability. Crucially, this holds even under weak connectivity permitted by Assumptions~\ref{ass:graph.connect.1.full} to~\ref{ass:graph.connect.2.c}.  The asymptotic optimality established in Theorem \ref{thm:main} is a statement that conditions on the fixed effects $\bs{\theta}$. It therefore holds even if the priors of \eqref{eq:sorting.prior} are misspecified, i.e., when $\bs{\theta}$ is generated from a distribution other than \eqref{eq:sorting.prior}. This makes the estimator $\hat{\bs{\theta}}(\bs{\lambda}^{\text{ure}})$ robust to misspecification, which is not guaranteed for other empirical Bayes estimators that rely on distributional assumptions in selecting $\bs{\lambda}$, such as $\hat{\bs{\theta}}(\bs{\lambda}^{\text{mle}})$.

\subsection{Comparison to Alternative Estimators}
\label{subsec:theory.alternatives}

The class of two-way estimators defined by \eqref{eq:Bayes.estimator} trivially covers the LS estimator by setting $\lambda_a=\lambda_b=0$. In addition, it covers two widely-used estimators that only shrink in the $\beta$ dimension. We denote the first estimator by $\hat{\bs{\beta}}^{\text{1-way}}(\tilde \lambda_b)$. It is constructed by assuming that the $\alpha_i$s are homogeneous. We initially impose the normalization on $\alpha_i=\alpha=0$ and derive the posterior mean based on the prior $\beta_j \sim_{iid} {\cal N}(\underline{\beta},\sigma^2/\tilde \lambda_b)$. {\em Ex post} we demean the estimator to make it comparable to the one that we used previously. This leads to 
\be
  \hat{\bs{\beta}}^{\text{1-way}}(\tilde \lambda_b) := M_c \cdot  [B_2^\prime B_2 + \tilde \lambda_b I_c]^{-1}
  \cdot  B_2^\prime \bs{Y}, \quad M_c:=I_c - \frac{1	}{c}1_{c\times c}.
  \label{eq:betahat1way}
\ee 
Notice that the {\em ex post} demeaning eliminates the effect of the prior mean $\underline{\beta}$. The estimator  $\hat{\bs{\beta}}^{\text{1-way}}(\tilde \lambda_b)$ can be obtained as the limit of our two-way shrinkage estimator as the precision of the prior for $\alpha_i$ goes to infinity.


\begin{lemma}[Shrinkage in One-Way Model]
	\label{lm:1way}
	Suppose $\lambda_b = \tilde{\lambda}_b$, $\phi=0$, $\mu=0$, and $\lambda_a\rightarrow\infty$. Then $\hat{\bs{\beta}}(\bs{\lambda})
    \rightarrow \hat{\bs{\beta}}^{\text{1-way}}(\tilde{\lambda}_b)$.
\end{lemma}

If we let $\hat{\bs{\beta}}^{\text{1-way}}(\lambda_b^{\text{mom}})$, where  $\lambda_b^{\text{mom}} := \tfrac{\sigma^2}{\text{var}(\beta_j)}$  and $\text{var}(\beta_j)$ is an empirical estimate of the variance within $\bs{\beta}$, we obtain the
one-way EB estimator of \cite{Kane2008}. This estimator is widely used for the estimation of teacher-value added and is generalized to time-varying 
$\bs{\beta}$ in \cite{Chetty2014} and \cite{Kwon2021}. 
Compared to $\hat{\bs{\beta}}^{\text{1-way}}(\cdot)$, our two-way EB estimator uses a data-dependent choice of $\lambda_a$ rather than setting it to infinity. When $\alpha_i$ is indeed heterogeneous, due to unobserved heterogeneity not fully controlled by observed covariates, the omitted variable bias could be large in the one-way estimator under assortative matching. 

The second estimator actually maintains the heterogeneity of the $\alpha_{i}$s but conducts shrinkage in the $\bs{\beta}$ dimension only. Specifically, this estimator first projects out $\bs{\alpha}$ using LS and subsequently estimates $\bs{\beta}$ with a one-way EB estimator that respects the normalization restriction $\bs{1}_c'\bs{\beta}=0$. We use $\hat{\bs{\beta}}^{\text{2-way}}(\tilde \lambda_b)$ to denote the posterior of $\bs{\beta}$ under a Gaussian prior, after $\bs{\alpha}$ is projected out. Its formula is given by
\be 
  \hat{\bs{\beta}}^{\text{2-way}}(\tilde \lambda_b) = 
  M_c \cdot L_{2,\perp} [L_{2,\perp} + \tilde \lambda_b I_c]^{-1}
  \cdot \hat{\bs{\beta}}^{\text{ls}}.
\ee  
Here $\hat{\bs{\beta}}^{\text{ls}}$ is the $\beta$ component of $\hat{\bs \theta}^{\text{ls}}$ defined in (\ref{eq:ls.estimator}), $L_{2,\perp}$ is the Laplacian matrix of the one-mode projected graph in \eqref{eq:def.L2perp}, and $M_c$ replaces ${\cal R}$ in the definitions of $S_1$ and $S$ in the posterior mean equation \eqref{eq:Bayes.estimator}. When $\tilde{\lambda}_b= 0$, define  $\hat{\bs{\beta}}^{\text{2-way}}= M_c\cdot \hat{\bs{\beta}}^{\text{ls}}$. This estimator  $\hat{\bs{\beta}}^{\text{2-way}}(\tilde \lambda_b)$ is our two-way shrinkage estimator with the precision of the prior for $\alpha_i$ set to zero.


\begin{lemma}[One-Way Shrinkage in Two-Way Model]
	\label{lm:1way.2}
	Suppose $\lambda_b = \tilde{\lambda}_b$, $\phi=0$, $\mu=0$, and $\lambda_a=0$. Then 
$\hat{\bs{\beta}}(\bs{\lambda})
		=
		\hat{\bs{\beta}}^{\text{2-way}}(\tilde \lambda_b).$
\end{lemma}

\cite{Chetty2018} use a one-way shrinkage estimator akin to $\hat{\bs{\beta}}^{\text{2-way}}(\lambda_b)$ to estimate the causal effects of growing up in different neighborhoods on the future outcomes of resident children. In this application, the mover-based identification\footnote{More precisely, the identification scheme uses a more specific variation, namely differences in movers' timings from one neighborhood to another to identify the relative per-year causal effects of the two neighborhoods.} and LS estimation eliminates unobserved assortative matching between individuals and neighborhoods. Then, a one-way shrinkage is conducted to estimate the causal effect for different neighborhoods.\footnote{While the shrinkage estimators in these applications typically use only the variances of $\hat{\bs{\beta}}^{\text{ls}}$ instead of $L_{2,\perp}$ in the shrinkage weight, we expect the latter to be more efficient as it takes into account the entire correlational structure.} 
Compared to $\hat{\bs{\beta}}^{\text{2-way}}(\lambda_b)$, the two-way shrinkage estimator benefit from pooling information in the $\bs{\alpha}$ dimension, and importantly utilizing information in the observed assortative matching pattern with a data-dependent choice of the correlation parameter $\phi$.

In Theorem~\ref{thm:main} we established that choosing $\bs{\lambda}=\bs{\lambda}^{ure}$ is optimal. It is tempting to deduce that as soon as a researcher chooses $\bs{\lambda}\not=\bs{\lambda}^{ure}$, as she most likely would when using $\hat{\bs{\beta}}^{\text{ls}}$, $\hat{\bs{\beta}}^{\text{1-way}}(\tilde{\lambda}_b)$, or $\hat{\bs{\beta}}^{\text{2-way}}(\tilde \lambda_b)$, the resulting estimator is no longer optimal in the sense of Definition~\ref{def:optimality}. Unfortunately, this form of suboptimality does not follow directly from Theorem~\ref{thm:main}. We proceed with a formal definition of dominance that is akin to inadmissibility.

\begin{definition} \label{def:dominance} An estimator $\hat{\bs{\theta}}$ is dominated by $\hat{\bs{\theta}}(\bs{\lambda}^{ure})$ if there is a set of parameters $(B,\bs{\theta})$ for which $\hat{\bs{\theta}}$ does not satisfy Definition~\ref{def:optimality}.
\end{definition}

\noindent {\bf Example:} To find parameter values	$(B,\bs{\theta})$ under which $\hat{\bs{\theta}}^{\text{ls}}$ is dominated by $\hat{\bs{\theta}}(\bs{\lambda}^{\text{ure}})$, it is sufficient to find $(B,\bs{\theta})$  under which $\hat{\bs{\theta}}^{\text{ls}}$ has strictly larger risk than a trivial estimator $\bs{0}_{r+c}$, which can be generated as a limit of $\hat{\bs{\theta}}(\bs{\lambda}^{\text{ure}})$ by setting $\mu=0$, $\phi=0$, and letting $\lambda_a,\lambda_b \rightarrow \infty$. The risk of the estimator $\bs{0}_{r+c}$ is $\tfrac{1}{r+c}\bs{\theta}'\bs{\theta}$, and the risk of $\hat{\bs{\theta}}^{\text{ls}}$
is\footnote{The inequality in \eqref{eq:LSdom} follows from von Neumann's trace inequality, the fact that ${\cal R'R}$ has a single eigenvalue of $r/c$ and eigenvalues of 1 with multiplicity of $r+c-1$, and an assumption that $r>c$ that corresponds to the setting of our empirical application.}
\begin{equation}\label{eq:LSdom}
	\frac{\sigma^2}{r+c}\tr[L^-]= \frac{\sigma^2}{r+c}\tr[{\cal R'R}L^\dagger] 
	\geq
	\frac{\sigma^2}{r+c}\tr[L^\dagger]
	=
	\frac{\sigma^2}{r+c}\sum_{\ell = 2}^{r+c} \frac{1}{\rho_\ell(B'B)}.
\end{equation} 
Thus, for DGPs satisfying 
\begin{equation}
	(B,\bs{\theta}):
	\quad
	\frac{\sigma^2}{r+c} \sum_{\ell =2}^{r+c} \frac{1}{\rho_\ell(B'B)} 
	> 
	\frac{1}{r+c} \left( \sum_{i=1}^r \alpha_i^2 + \sum_{j=1}^c \beta_j^2 \right),
	\label{eq:ex.btheta.condition}
\end{equation}
$\hat{\bs{\theta}}^{\text{ls}}$ has larger risk than $\bs{0}_{r+c}$. Under these DGPs, $\hat{\bs{\theta}}^{\text{ls}}$ also has a larger risk than  $\hat{\bs{\theta}}(\bs{\lambda}^{\text{ure}})$ and is dominated in the sense of Definition~\ref{def:dominance}. Condition (\ref{eq:ex.btheta.condition}) is satisfied in settings in which the graph is weakly connected (some eigenvalues of $L=B'B$ are close to zero) and the deviations of the two-way effects from zero are small. $\nobreak\hfill$ $\qed$

Rather than conducting a formal analysis for the comparison of $\hat{\bs{\theta}}(\bs{\lambda}^{\text{ure}})$ with other estimators, we provide numerical illustrations in the Monte Carlo experiments in Section~\ref{sec:simul}.


\subsection{Estimated Regression Coefficients}
\label{subsec:theory.estimated.gamma}

Next, we study conditions for the asymptotic optimality in Theorem \ref{thm:main} to hold when the regressor coefficients $\bs{\gamma}$ are estimated. Let $\tilde{\bs{\gamma}}$ be an estimator of $\bs{\gamma}$. Define all the estimators exactly the same as before with $\bs{Y}:=\bs{Y}^*-X\bs{\gamma}$ replaced by $\bs{Y}:= \bs{Y}^*-X\tilde{\bs{\gamma}}$. 
In particular, let $\tilde{\bs{\theta}}(\bs{\lambda})$, $\widetilde{\text{URE}}(\bs{\lambda})$, and $\bs{\lambda}^{\widetilde{\text{ure}}}$ be the counterparts of $\hat{\bs{\theta}}(\bs{\lambda})$,  $\text{URE}(\bs{\lambda})$, and $\bs{\lambda}^{\text{ure}}$, respectively. We provide conditions under which the impact of $\tilde{\bs{\gamma}}$ is negligible uniformly on both the compound loss and the URE, e.g., $
\sup_{\bs{\lambda}\in\mathcal{J}} 
|
l_w(\hat{\bs{\theta}}(\bs{\lambda}),\bs{\theta}) - 
l_w(\tilde{\bs{\theta}}(\bs{\lambda}),\bs{\theta})
|
\rightarrow_p 0,$ and
$
\sup_{\bs{\lambda}\in\mathcal{J}} 
|
\text{URE}(\bs{\lambda}) - 
\widetilde{\text{URE}}(\bs{\lambda})
|
\rightarrow_p 0
$
in order to extend the asymptotic optimality to the EB estimator 
$\tilde{\bs{\theta}}(\bs{\lambda}^{\widetilde{\text{ure}}})$.
To this end, we impose the following assumption on the regressors and the estimator $\tilde{\bs{\gamma}}$. In this case, $\text{URE}(\bs{\lambda})$ can be interpreted as an asymptotic unbiased risk estimate.



\begin{assumption}[Estimate $\bs{\gamma}$]
	\label{ass:reg.condreg}
	As $r,c \rightarrow\infty$, 
	$\mathbb{E}[||\bs{x}_{it}||^2]\le M$ for some finite constant $M$ and
	$\tilde{\bs{\gamma}}-\bs{\gamma} = O_p (r^{-1/2})$.
\end{assumption}

We suggest using the $r^{1/2}$-consistent estimators of $\gamma$  studied by \cite{Verdier2020} in sparsely connected two-way model. 
With strictly exogenous regressors, this holds for the standard OLS estimator. 
When regressors are only sequentially exogenous, e.g., the lagged test score, \cite{Verdier2020} propose a $r^{1/2}$-consistent estimator of $\bs{\gamma}$ by extending the recursive orthogonal transformation from the one-way fixed effect model to the two-way fixed effects model. The following Corollary is the generalization of Theorem~\ref{thm:main} to estimated regressors.



\begin{corollary}\label{cor:reg}
	Under all the Assumptions in Theorem~\ref{thm:main} and Assumption~\ref{ass:reg.condreg}, the main result in Theorem ~\ref{thm:main} hold when $\bs{Y}:=\bs{Y}^*-X\bs{\gamma}$ is replaced by $\bs{Y}:= \bs{Y}^*-X\tilde{\bs{\gamma}}$ in all definitions.
\end{corollary}



\section{Monte Carlo Simulation}
\label{sec:simul}

The estimators are introduced in Section~\ref{subsec:simul.estimators}, the
simulation designs are described in Section~\ref{subsec:simul.designs} and the results are summarized in Section~\ref{subsec:simul.results}.

\subsection{Estimators}
\label{subsec:simul.estimators}

We study the properties of the proposed estimator $\hat{\bs{\theta}}(\bs{\lambda}^{\text{ure}})$, henceforth EB-URE, vis-\`{a}-vis the infeasible oracle estimator $\hat{\bs{\theta}}(\bs{\lambda}^{\text{ol}})$, henceforth OL, the least squares estimator $\hat{\bs{\theta}}^{\text{ls}}$, henceforth LS, and $\hat{\bs{\theta}}(\bs{\lambda}^{\text{mle}})$, henceforth EB-MLE, that selects $\bs{\lambda}$ to maximize the marginal likelihood of $\bs{Y}$, i.e., the likelihood that integrates out $\bs{\theta}$ according to \eqref{eq:sorting.prior}. The computational details for these estimators are detailed in Section~\ref{sec:compute} of the Supplementary Online Appendix. For simulations that compare the estimation of only $\bs{\beta}$, we also include the one-way fixed effects estimator $\hat{\bs{\beta}}^{\text{1-way}}(\lambda_b^{\text{mom}})$ defined in (\ref{eq:betahat1way}), henceforth  EB-1way.

\subsection{Simulation Designs}
\label{subsec:simul.designs}

The Monte Carlo designs are based on a prototypical application of estimating teacher value-added based on a linked student-teacher data set. We assume that there are $r=40,000$ students, $c=4,000$ teachers, and students and teachers are allocated to $s=200$ schools, numbers that are of the same order as those in the empirical application. The time dimension of the panel is $T=2$. The designs begin with the generation of the unit-specific parameters $\bs{\theta}$, then a bipartite graph ${\cal G}$ is generated conditional on $\bs{\theta}$, which in turn determines the matrix $B$. 

\begin{enumerate}
	\item {\bf Period $t=0$:} Generate unit-specific parameters by $iid$ sampling: $\alpha_i \sim_{iid} \mathcal{N}(0,\sigma_\alpha^2)$ and $ \beta_j \sim_{iid} \mathcal{N}(0,\sigma_\beta^2)$ for $i=1,\dots,r$ and $j=1,\dots,c$.
	\item {\bf Period $t=1$:}
	\begin{enumerate}
		\item Rank teachers by $\beta_j$ in ascending order and allocate them to the $s$ schools. School~1 (School~s) receives the $c/s$ teachers with the lowest (highest) $\beta_j$ values. Randomly re-assign a fraction of $1-\pi_{\text{match}}$ teachers from each school to one of the $s$ schools. 
		\item Rank students by $\alpha_i$ in ascending order and allocate them to the $s$ schools. School~1 (School~s) receives the $r/s$ students with the lowest (highest) $\alpha_i$ values. Randomly re-assign a fraction of $1-\pi_{\text{match}}$ students from each school to one of the $s$ schools. 
		\item Within each school, students and teachers are randomly matched into classes. Thus, any student-teacher assortative matching takes place between and not within schools. 

	\end{enumerate}
	\item {\bf Period $t=2$:}
	\begin{enumerate}
		\item From each school, we randomly select $\pi_{\text{mob}}$ proportion of incumbent teachers to switch school. If teacher $j$ is designated to leave her school, she is assigned to one of the other $s-1$ schools with equal probability.
		\item Matching students to teachers: same as in Step~(c) of period $t=1$. 
	\end{enumerate}
\end{enumerate}


We start from perfect assortative matching and then make random re-assignments controlled by the parameter $\pi_{\text{match}}$ to break the matching: $\pi_{\text{match}}=1$ ($\pi_{\text{match}}=0$)  corresponds to perfect positive matching (no matching). If we skip the period $t=2$ step of teachers switching schools, then the resulting graph ${\cal G}$ will be disconnected. While there would be a potentially connected subgraph ${\cal G}_l$ for each school, these subgraphs are disconnected from each other. Thus, we can use  $\pi_{\text{mob}}$ to control the connectivity of ${\cal G}$. The outcomes for periods $t=1,2$ are determined according to
\begin{equation}
	y_{it}=\alpha_i+\beta_{j(i,t)}+ \sigma \cdot u_{it}, \text{ where } u_{it}\sim_{iid}\mathcal{N}(0,1).
\end{equation}

The simulation designs are summarized in Table~\ref{tab:mc.designs}. Design 1 is the reference design, where we calibrate the DGP parameters such that the empirical moments of $\hat{\bs{\theta}}^{\text{ls}}$ and connectivity of projected teacher graph match those estimated from the empirical application in Section~\ref{sec:appl}.
The subsequent designs then each perturb Design 1 along a single dimension: Design 2 assumes a relatively high level of student-teacher assortative matching and Design 3 assumes a relatively high level of teacher mobility across schools. Finally, Design 4 inverts the heterogeneity in student abilities and teacher value-added, so that students are relatively homogeneous. Under all designs, the estimators are evaluated according to their RMSEs defined as 
\begin{equation}
	\mbox{RMSE}(\hat{\bs{\theta}},W_b) = \sqrt{ l_{w_b}(\hat{\bs{\theta}},\bs{\theta}) },
	\label{eq:simul.rmse}
\end{equation}
where $w_b$ represents the usage of $W_b$ in the loss as defined in \eqref{eq:weights}.\footnote{Section~\ref{sec:details.MonteCarlo} of the Supplementary Online  Appendix repeats the simulations for Design 1 using the RMSE $ \sqrt{ l_{w_{a+b}}(\hat{\bs{\theta}},\bs{\theta}) }$ and also considers alternative specifications of the distributions governing the unit-specific parameters $\bs{\alpha}$ and $\bs{\beta}$ beyond Gaussianity, such as with skewness or fat tails.}

\begin{table}
	\caption{Simulation Designs}
	\label{tab:mc.designs}
	\begin{center}
		\begin{tabular}{lccccc} 
			\toprule
			&  Emp. Appl.  & \multicolumn{4}{c}{ Simulation Designs } \\
			\cmidrule(lr){3-6}
			& - & 1 & 2 & 3 & 4  \\ 
			\midrule
			DGP Parameters \\
			$\pi_{\text{match}}$ & - & 0.4 & 0.7 & 0.4 & 0.4 \\
			$\pi_{\text{mob}}$ & - & 0.05 & 0.05 & 0.12 & 0.05 \\
			$\sigma_\alpha^2$ & - & 0.6 & 0.6 & 0.6 & 0.06 \\
			$\sigma_\beta^2$ & - & 0.06 & 0.06 & 0.06 & 0.6 \\
			\cmidrule(lr){1-6}
			Empirical Moments \\
			$\text{var}(\alpha_i)$  & -  & 0.6 & 0.6  & 0.6  & 0.06  \\
			$\text{var}(\beta_{j(i,t)})$ & - & 0.06 & 0.06  & 0.06  & 0.6 \\
			$\text{cor}(\alpha_i,\beta_{j(i,t)})$  & - & 0.15 & 0.46  & 0.14 & 0.15 \\
			\cmidrule(lr){1-6}
			Empirical Moments \\
			$\text{var}(\hat{\alpha}_i^{\text{ls}})$  & 0.75 & 0.77 & 0.80 & 0.68  & 0.23  \\
			$\text{var}(\hat{\beta}_{j(i,t)}^{\text{ls}})$ & 0.15 & 0.18 & 0.21  & .088  &  0.72 \\
			$\text{cor}(\hat{\alpha}_i^{\text{ls}},\hat{\beta}_{j(i,t)}^{\text{ls}})$  & -0.28  & -0.23 & -0.13  & .023  &  -0.20 \\
			\cmidrule(lr){1-6}
			Connectivity \\ 
			$1E5 \cdot \rho_{\text{min}}$ & 17 & 52 & 35 &  1980 & 53 \\
			\bottomrule
		\end{tabular}
	\end{center}
	{\footnotesize {\em Notes:} 
		The moments and connectivity values displayed for the designs are median values across the simulation rounds. 
		The sampling variance $\sigma^2$ is set as $0.12$ for all four designs, consistent with the value estimated from the empirical application. The variable $\rho_{\text{min}}$ denotes the smallest non-zero eigenvalue of the \emph{normalized} Laplacian projected teacher graph; we choose to match these instead of the unnormalized ones because the former always lie within $[0,2]$ and as a result provide a more natural scale to match.  
	}\setlength{\baselineskip}{4mm}
\end{table}

\subsection{Results}
\label{subsec:simul.results}

\begin{figure}[t!]
	\caption{RMSE of Estimators}
	\label{fig:mc.rmse}
	\begin{center}
		\begin{tabular}{cc}
			Design 1 (Calibrated DGP) & Design 2 (Strong Sorting) \\ 
			\includegraphics[width=.35\textwidth]{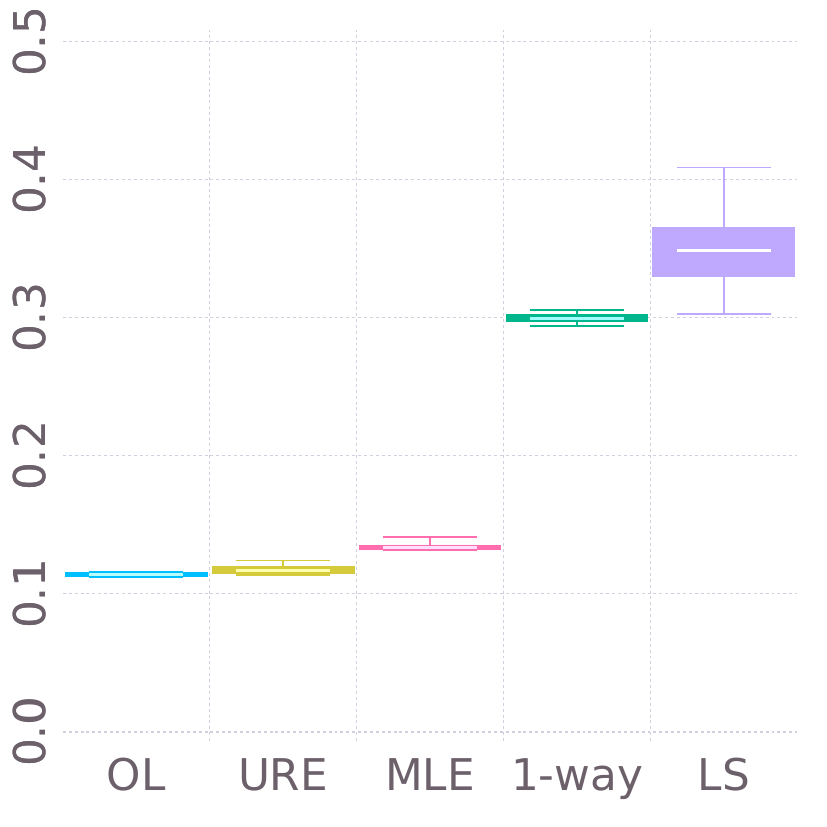} &
			\includegraphics[width=.35\textwidth]{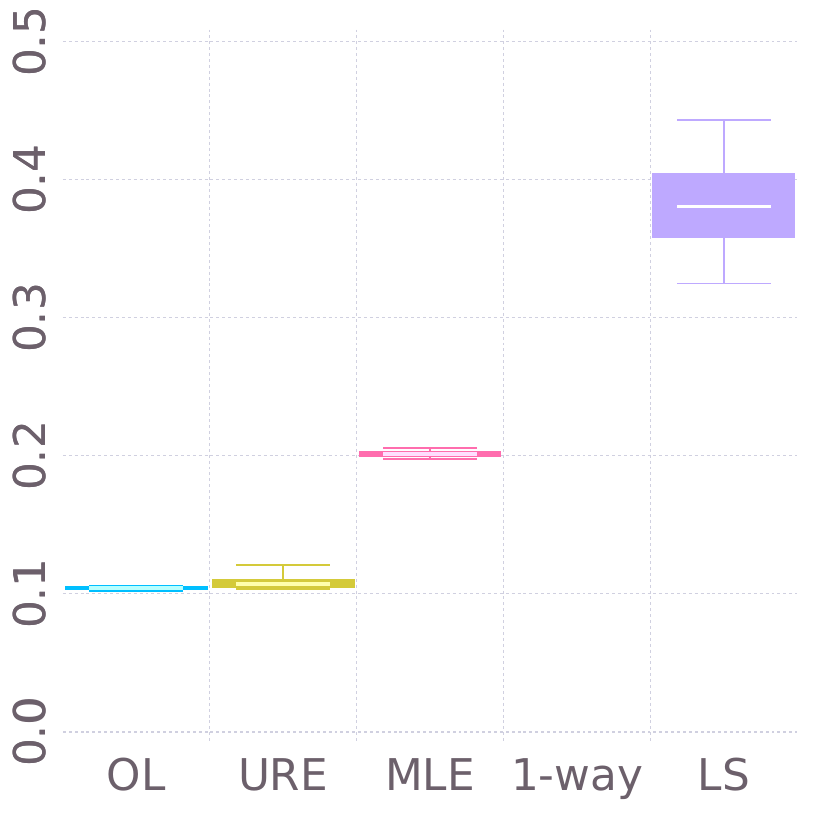} \\
			Design 3 (Strong Mobility) & Design 4 (Homogeneous Students)\\ 
			\includegraphics[width=.35\textwidth]{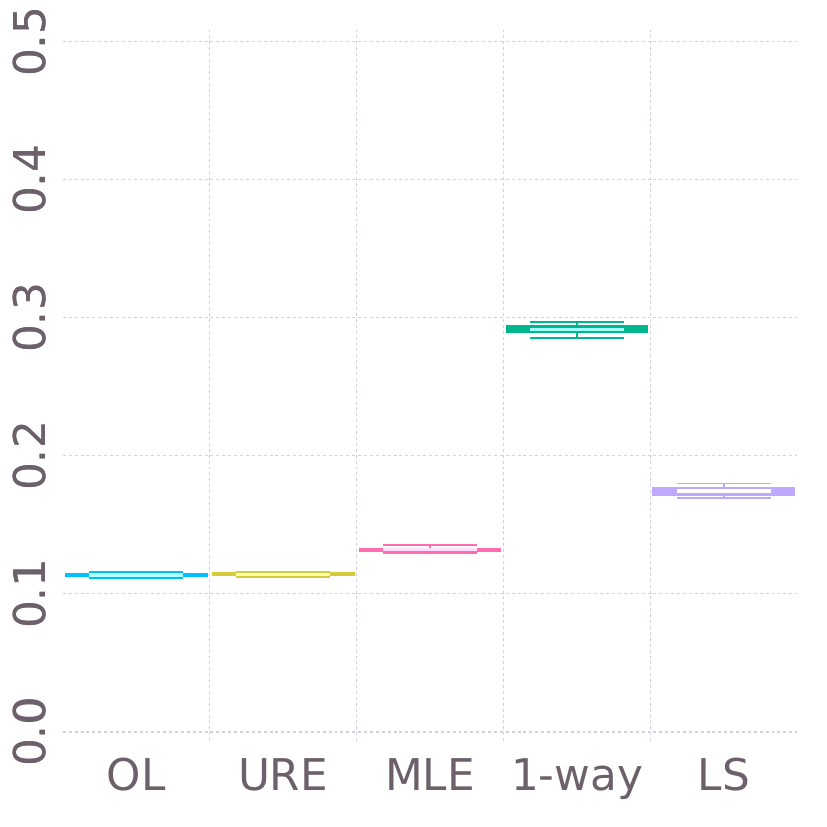} &
			\includegraphics[width=.35\textwidth]{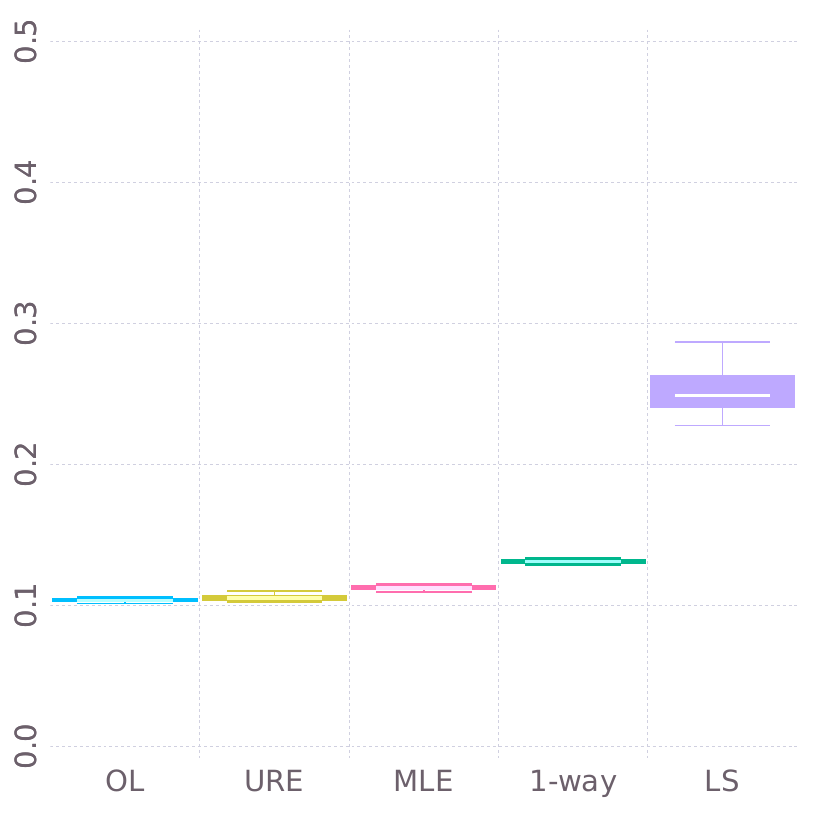} \\
		\end{tabular}   
	\end{center}
	{\footnotesize {\em Notes:} The box plots characterize the distribution of RMSEs across Monte Carlo repetitions.}\setlength{\baselineskip}{4mm}
\end{figure}

\begin{table}
	\caption{Hyperparameter Selection}
	\label{tab:mc.hyperpara}
	\begin{center}
		\begin{tabular}{lcccccccccc} 
			\toprule
			& \multicolumn{3}{c}{ OL} & 
			\multicolumn{3}{c}{ EB-URE} &
			\multicolumn{3}{c}{ EB-MLE} & 
			\multicolumn{1}{c}{ EB-1Way} \\
			\cmidrule(lr){2-4} \cmidrule(lr){5-7}
			\cmidrule(lr){8-10} \cmidrule(lr){11-11} 
			& $\lambda_a$ & $\lambda_b$ & $\phi$ 
			& $\lambda_a$ & $\lambda_b$ & $\phi$ 
			& $\lambda_a$ & $\lambda_b$ & $\phi$ 
			& $\lambda_b$ \\ 
			\midrule
			Design\\
			\;1&  0.059 & 2.27 & 0.47 & 0.065 & 2.17 & 0.42 & 0.21 & 3.3 & 0.89 & 2.8  \\ 
			\;2&  0.056 & 3.3 & 0.79 & 0.069 & 3.2 & 0.53 & 0.40 & 3.3 & 0.9 & 0.79 \\ 
			\;3&  0.056 & 2.23 & 0.50 & 0.064 & 2.3 & 0.45 & 0.2 & 3.2 & 0.88 & 2.9 \\ 
			\;4&  0.63 & 0.21 & 0.9 & 0.73 & 0.23 & 0.89 & 2.5 & 0.21 & 0.39 &  0.24 \\ 
			\bottomrule
		\end{tabular}
	\end{center}
	{\footnotesize {\em Notes:} The values displayed are the medians across simulation rounds. The value of $\mu$ selected by the various methodologies is consistently $\approx0$ and hence excluded from this table.}\setlength{\baselineskip}{4mm}
\end{table}

Figure~\ref{fig:mc.rmse} shows the distribution of RMSEs of estimators across Monte Carlo repetitions under different designs. EB-URE tracks the RMSE of the infeasible benchmark OL very closely across all designs, regardless of connectivity strength, matching intensity, or heterogeneity magnitudes. In particular, this is true for Design 1 which is calibrated towards the empirical application and provides justification for implementing the EB-URE in the application.  Table~\ref{tab:mc.hyperpara} lists the selected hyperparameters and shows that EB-URE and OL are based on similar hyperparameter values. This is a reflection of the URE objective function providing a good approximation to the infeasible loss function. The EB-MLE, on the other hand, which relies on a correct specification of the prior \eqref{eq:sorting.prior}, is suboptimal in most designs. The suboptimality is particularly pronounced under strong matching (Design 2). This suggests a higher degree of assortative matching renders \eqref{eq:sorting.prior}, which motivates our class of estimators $\bs{\beta}(\bs{\lambda})$, less representative of the true distribution of $\bs{\theta}| B$. The proposed EB-URE instead targets a purely frequentist criterion, and thus continues to perform well relative to the benchmark OL in the presence of a misspecified prior.  

We also see in Design 1 that the LS estimator performs poorly relative to OL. This is unsurprising, because the weak connectivity of the projected teacher graph results in a large variance of the LS estimators. It is therefore desirable to induce some level of shrinkage to optimize the bias-variance trade-off. However, Figure~\ref{fig:mc.rmse} also suggests that the method used in optimizing this trade-off is crucial. The feasible EB-URE attains at least 50\% RMSE reduction relative to LS across most designs. On the other hand, the performance of EB-1way is highly dependent on the underlying DGP. For instance, with even a weak positive student-teacher matching as in Design 1, the EB-1way estimator only performs marginally better than the LS. If we have a strong level of matching as in Design 2, then its performance may be even worse than the LS. In fact, its RMSE there is uniformly greater than $0.5$. By assuming $\lambda_a=\infty$, EB-1way is forced by the student-teacher matching to wrongly attribute variability in $\bs{\alpha}$ as variability in $\bs{\beta}$. It thus selects a smaller $\lambda_b$ than is optimal; see Table~\ref{tab:mc.hyperpara}. Only if the student heterogeneity is small relative to the error variance $\sigma^2$ as in Design 4, the EB-1way performs nearly optimally.


We now examine the effects of weak connectivity on the three types of  estimates at a more granular level. Figure~\ref{fig:mc.sorting} shows scatter plots of the teacher value-added $\beta_j$ ($x$-axis) and the average matched student ability for each teacher ($y$-axis), defined as 
\begin{equation}
	\mu_j:= \frac{1}{d_{b,j}}\sum_{(i,t):j(i,t)=j} \alpha_i,
	\label{eq:def.muj} 
\end{equation}
where $d_{b,j}$ is the total number of matches for teacher $j$. We provide scatter plots of the true pairs $(\beta_j,\mu_j)$, $j=1,\ldots,c$ and their estimates (LS, EB-URE, and EB-MLE) for Designs~1 and~2. The true values depicted in the top-left panels for both designs have positive slopes, reflecting the positive assortative matching patterns underlying the DGPs. The top-right panels show the same scatter plot for the LS estimates. The weak connectivity results in a bias large enough that the implied correlation between $(\beta_j,\mu_j)$ is actually highly negative for Design~1.
The bottom-left panels show the EB-URE estimates of $(\beta_j,\mu_j)$. Even though the hyperparameter selection does not target the estimation loss of $\bs{\alpha}$, the estimates indicate a strong correlation between $\beta_j$ and $\mu_j$, meaning that the EB-URE estimator captures a significant portion of matching. 

\begin{figure}[t!]
	\caption{Matching Patterns}
	\label{fig:mc.sorting}	
	\begin{center}
		\begin{tabular}{ccccc}
			& \multicolumn{2}{c}{ Design 1 (Calibrated DGP) } & \multicolumn{2}{c}{ Design 2 (Strong Matching) } \\
			\\
			& True Values $(\beta_j,\mu_j,)$ & 
			LS $(\hat{\beta}_j,\hat{\mu}_j)$ &
			True Values $(\beta_j,\mu_j)$ & 
			LS $(\hat{\beta}_j,\hat{\mu}_j)$ \\
			\rotatebox{90}{\footnotesize \hspace*{1cm} Student} &
			\includegraphics[width=.2\textwidth]{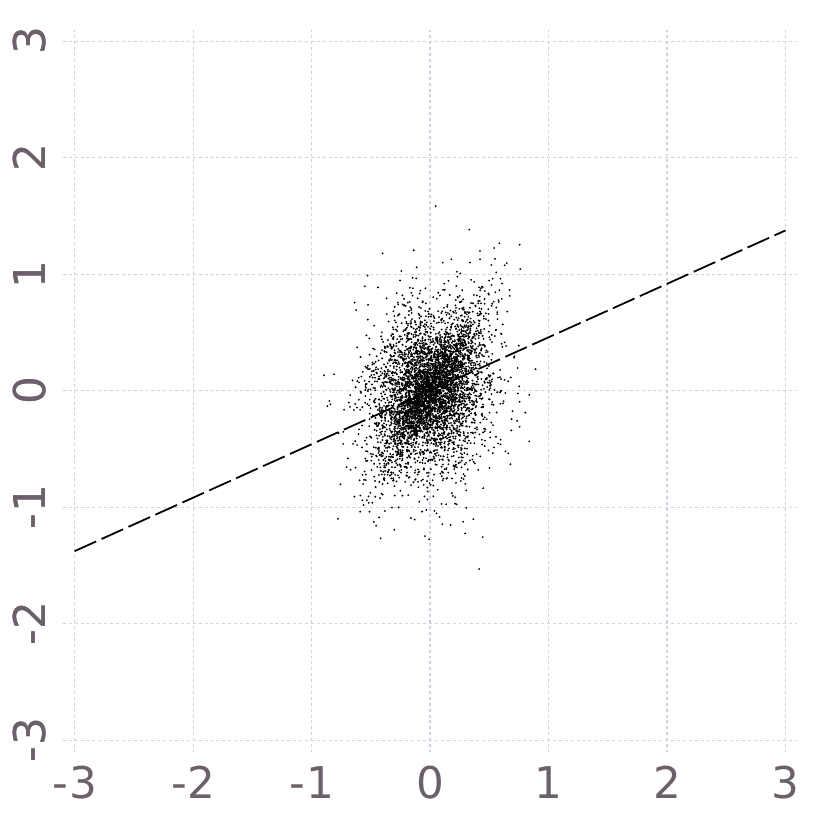} &
			\includegraphics[width=.2\textwidth]{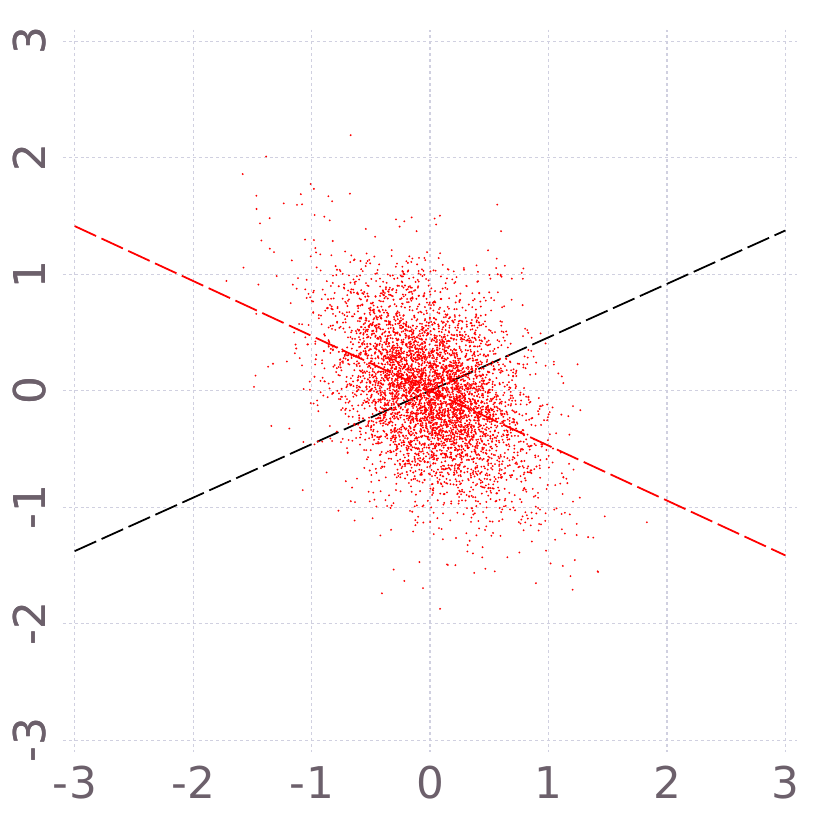} &
			\includegraphics[width=.2\textwidth]{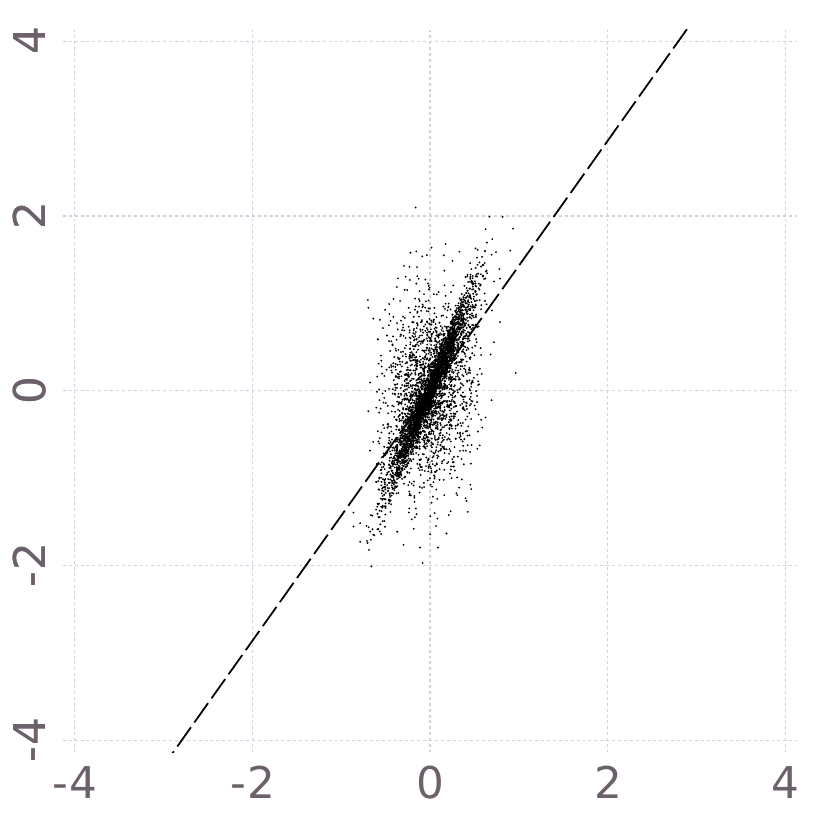} &
			\includegraphics[width=.2\textwidth]{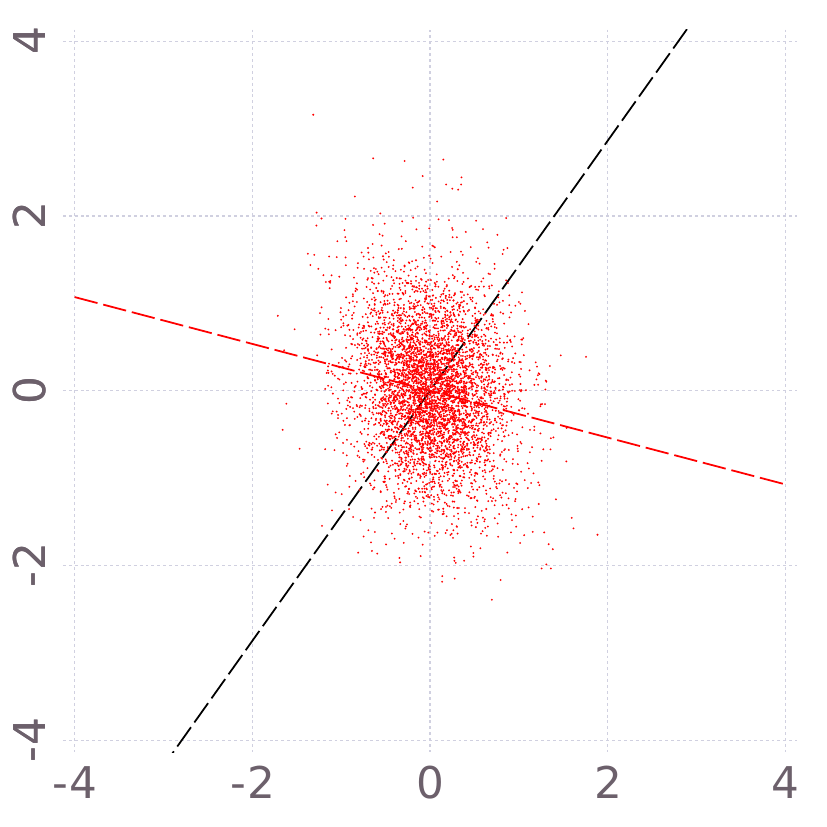}\\
			& EB-URE $(\hat{\beta}_j,\hat{\mu}_j)$ & 
			EB-MLE $(\hat{\beta}_j,\hat{\mu}_j)$ &
			EB-URE $(\hat{\beta}_j,\hat{\mu}_j)$ & 
			EB-MLE $(\hat{\beta}_j,\hat{\mu}_j)$ \\ 
			\rotatebox{90}{\footnotesize \hspace*{1cm} Student} &
			\includegraphics[width=.2\textwidth]{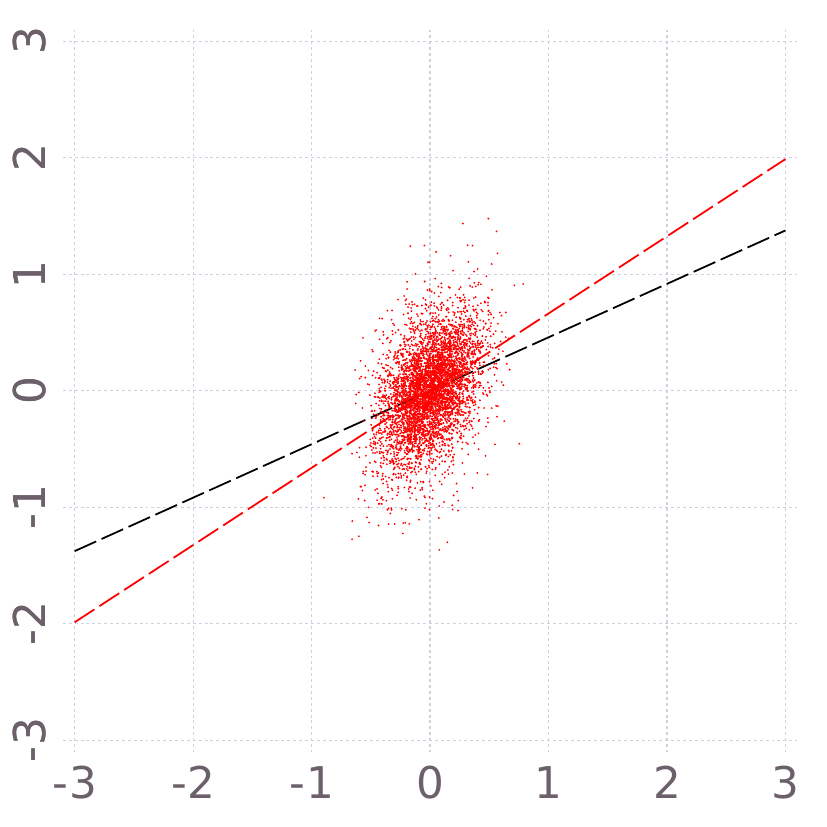} &
			\includegraphics[width=.2\textwidth]{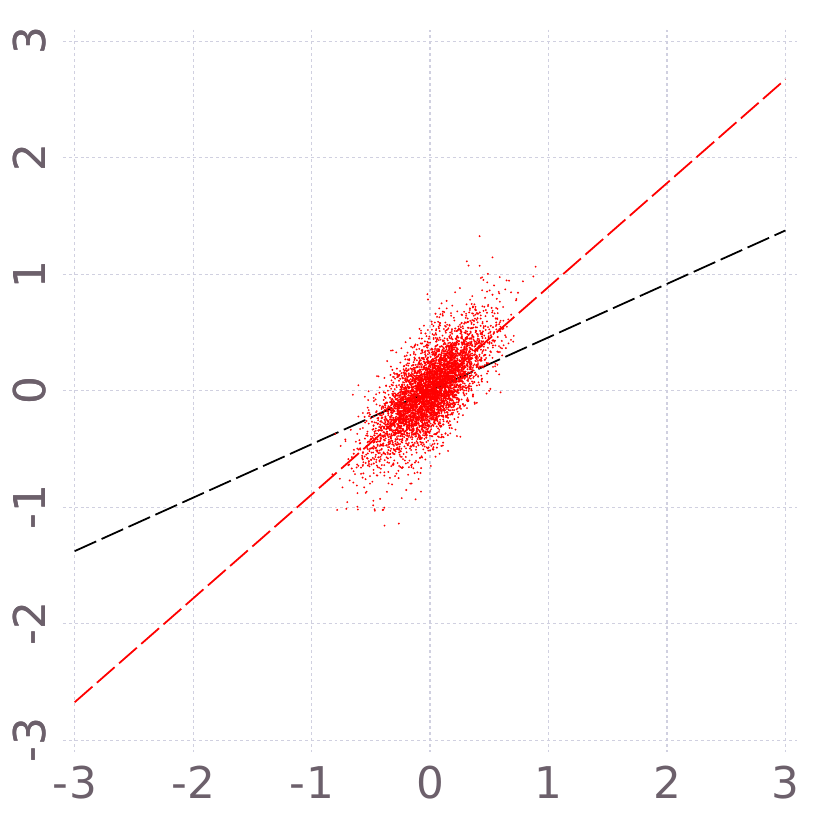} &
			\includegraphics[width=.2\textwidth]{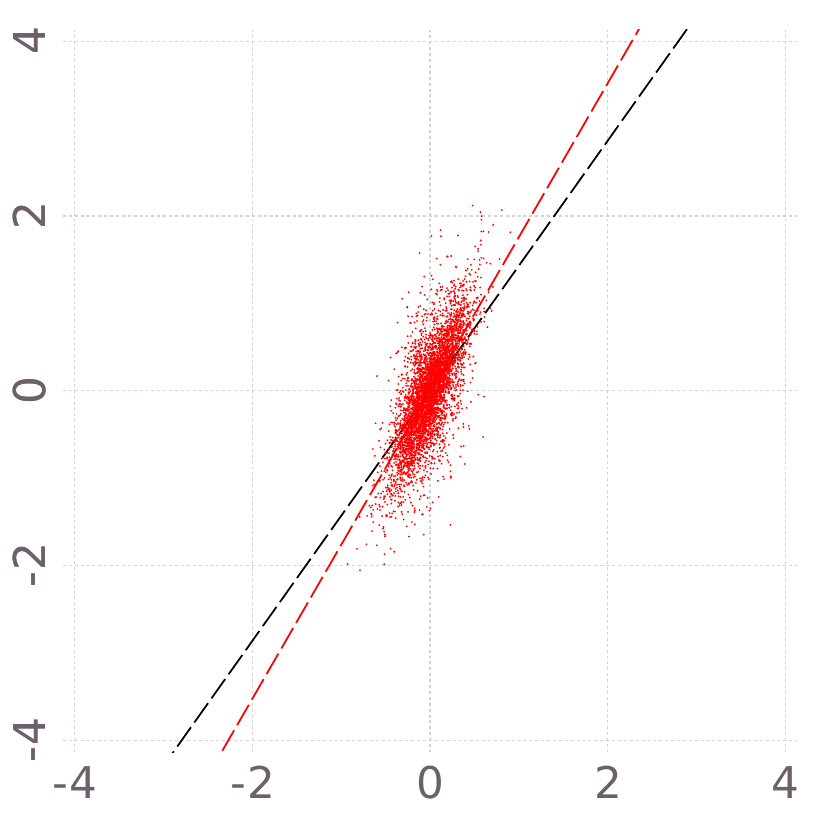} &
			\includegraphics[width=.2\textwidth]{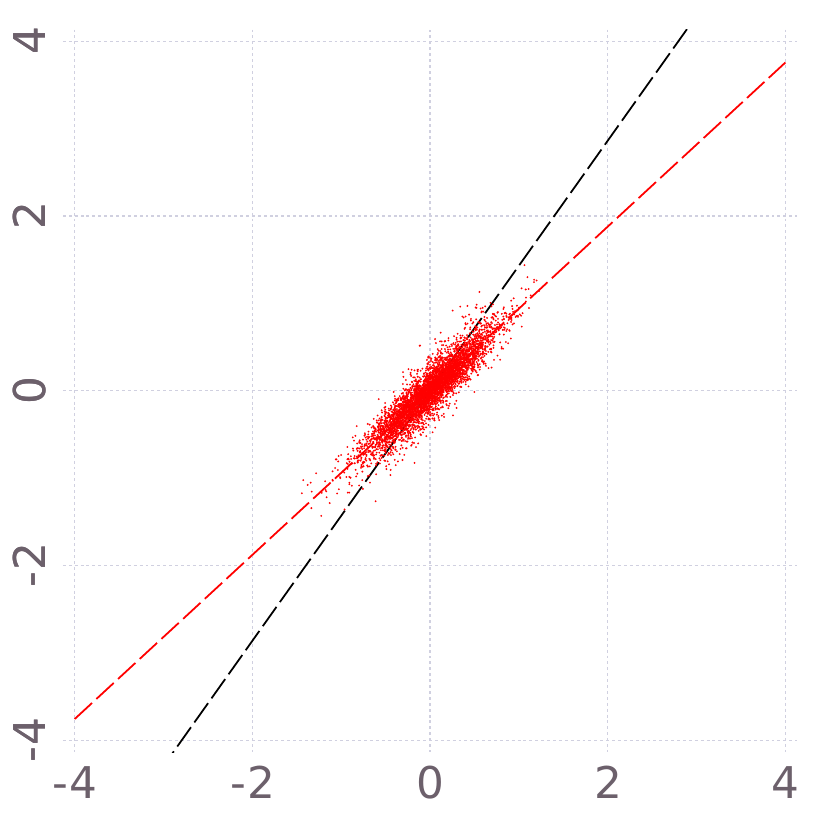} \\
			& {\footnotesize Teacher} & {\footnotesize Teacher} & {\footnotesize Teacher} & {\footnotesize Teacher}
		\end{tabular} 
	\end{center}
	{\footnotesize {\em Notes:} $x$-axis is $\beta_j$ and $y$-axis is $\mu_j$ (true values or estimates). The scatter plots are based on 5,000 random draws from all Monte Carlo repetitions. 	Black lines are the least squares regression lines from the true $(\beta_j,\mu_j)$ by pooling all simulation rounds, while red lines are least squares regression lines using estimated values (LS, EB-URE, or EB-MLE).}
\end{figure}

Table~\ref{tab:mc.mom} provides further evidence through empirical moments computed from the true effects and their  estimates. We report the variances of $\alpha_i$ and $\beta_j$, respectively, and the correlations between $\mu_j$ and $\beta_j$. For instance, under Design 1 the empirical correlation between the true effects is 0.15, whereas the correlation between the EB-URE estimates is 0.19. For Design 2 these values change to 0.46 and 0.50, respectively. The empirical variances of $\alpha_i$ and $\beta_j$ computed from the true and their estimates are also very close. The empirical moments computed from the EB-MLE estimates exhibit slightly larger discrepancies. For instance, for Designs 2 and 3 the correlations among the true effects are 0.46 and 0.14, whereas the correlation among the EB-MLE estimates are 0.56 and 0.28, respectively. We conclude from the simulations that the EB-URE estimates not just generate low RMSEs, but they also are able to reproduce key empirical moments of the underlying true effects.

\begin{table}[t!]
	\caption{Empirical Moments}
	\label{tab:mc.mom}
	\begin{center}
		\begin{tabular}{lccccccccc} 
			\toprule
			& \multicolumn{3}{c}{ True} & 
			\multicolumn{3}{c}{ EB-URE} &
			\multicolumn{3}{c}{ EB-MLE} \\
			\cmidrule(lr){2-4} \cmidrule(lr){5-7}
			\cmidrule(lr){8-10}
			& $\text{v}_\alpha$ & $\text{v}_\beta$ & $\rho_{\alpha,\beta}$ 
			& $\text{v}_\alpha$ & $\text{v}_\beta$ & $\rho_{\alpha,\beta}$   
			& $\text{v}_\alpha$ & $\text{v}_\beta$ & $\rho_{\alpha,\beta}$ 
			\\ 
			\midrule
			Design\\
			\;1&  0.6 & 0.06 & 0.15 & 0.62 & .049 & 0.19 & 0.51 & .057 & 0.29 \\ 
			\;2&  0.6 & 0.06 & 0.46 & 0.63 & .050 & 0.50 & 0.38 & 0.14 & 0.56 \\ 
			\;3&  0.6 & 0.06 & 0.14 & 0.62 & .047 & 0.19 & 0.52 & .055 & 0.28 \\ 
			\;4&  0.06 & 0.6 & 0.15 & .065 & 0.57 & 0.21 & .021 & 0.63 & 0.11 \\ 
			\bottomrule
		\end{tabular}
	\end{center}
	{\footnotesize {\em Notes:} $(\text{v}_\alpha,\text{v}_\beta,\rho_{\alpha,\beta})$ represent $\text{var}(\hat{\alpha}_i)$, $\text{var}(\hat{\beta}_{j(i,t)})$ and $\text{cor}(\hat{\alpha}_i,\hat{\beta}_{j(i,t)})$ respectively. The columns for ``True'' use the true values of $(\alpha_j,\beta_{j(i,t)})$ in the computation and are thus the transpose of the second row block of Table~\ref{tab:mc.designs}.}\setlength{\baselineskip}{4mm}
\end{table}


\section{Empirical Application}
\label{sec:appl}

We utilize a matched student-teacher dataset from the North Carolina Education Research Data Center with observations on students from grades three to five for the years 2017 and 2018 period to estimate teacher value-added. For identification purposes, we restrict the dataset to the largest connected component of the student-teacher graph. We remove students that appear only once in the data set, since they are not relevant for the identification of the teacher value-added parameters, students that repeat grades, and students with special accommodations. This leaves us with $r=41,243$ students, $c=5,332$ teachers, and $s=258$ schools. Because we have test scores for two consecutive years, this leads to 82,486 observations in total.
	
Starting point of the empirical analysis is model \eqref{eq:underlyingmodel}. We take the outcome variable $y_{it}$ to be a math test score, standardized within each (year, grade) cell in accordance with the teacher value-added literature to have a mean of zero and variance of one. We only include a single regressor $\bs{x}_{it}$, namely the student's lagged test score $y_{it-1}$.\footnote{We decided to keep the model simple. As a robustness exercise we repeat the analysis by also including cubic polynomials in age and class size; see  Section~\ref{sec:details.empirics.1way} of the Supplementary Online Appendix. The resulting teacher value-added estimates have a correlation of $0.989$ with the ones reported here.}  We apply the transformation in (\ref{eq:underlyingmodelcompact.nogamma}) by subtracting $X\tilde{\bs{\gamma}}$ from the outcome variable. Because $y_{it-1}$ is only sequentially and not strictly exogenous, we use \cite{Verdier2020}'s estimator to obtain $\tilde{\bs{\gamma}}=0.067$. The estimator intuitively extends the recursive orthogonal transformation in one-way to two-way effects models for consistent estimation; see Section~\ref{sec:details.empirics.verdier} of the  Supplementary Online Appendix for more details. The subsequent analysis also conditions on the error variance estimate $\hat{\sigma}^2 = 0.12$. In Section~\ref{subsec:appl.connect} we report measures of connectivity for our data set and discuss some features of the LS estimates. The EB-URE estimates are presented in  Section~\ref{subsec:appl.ure} and in  Section~\ref{subsec:appl.rankings} we compare teacher rankings based on value-added estimates from different estimators.

\subsection{Connectivity and LS Estimates}
\label{subsec:appl.connect}

We previously emphasized that the performance of estimators for the two-way effects model depends on the connectivity of the student-teacher graph.  Table~\ref{tab:appl.eig} provides information about the distribution of eigenvalues of the Laplacian $L_{2,\perp}$ of the projected teacher graph. The smallest eigenvalue is 0.013, the first percentile is 0.069, and the fifth percentile is 0.23. This distribution appears to be broadly in line with our Assumption~\ref{ass:graph.connect.2.b}	 Section~\ref{subsec:theory.covergence} that only a finite number of the eigenvalues can be small, but not too small. When scaled by $c^{1/2}$ as in Assumption~\ref{ass:graph.connect.1.b} the smallest non-zero eigenvalue is 0.96 and the first percentile is 5.1. We interpret these numbers as evidence of weak connectivity: many subsets of teachers share few common students, in particular teachers from different schools. 

\begin{table}
	\caption{Minimum Non-zero Eigenvalue / Empirical Quantiles}
	\label{tab:appl.eig}
	\begin{center}
		\begin{tabular}{lccccccc} 
			\toprule
			& & \multicolumn{6}{c}{Empirical Quantiles} \\
			\cmidrule(lr){3-8} 
			& min & .005 & 0.01 & 0.05
			& 0.1 & 0.2 & 0.5 \\ 
			\midrule
			\;$\rho_\ell$& .013 & .049 & .069 & 0.23 & 0.44 & 1.0 & 6.0 \\ 
			\;$\rho_\ell * c^{1/2}$& 0.96 & 3.6 & 5.1 & 17 & 32 & 76 & 440 \\
			\bottomrule
		\end{tabular}
	\end{center}
	{\footnotesize {\em Notes:} $\rho_\ell$ represents the $\ell^{th}$ smallest eigenvalue of the Laplacian of the projected teacher graph. ``$\text{min}$'' represents the smallest non-zero eigenvalue, i.e., $\rho_2$. The second row is the first row multiplied by $c^{1/2}=73$.}\setlength{\baselineskip}{4mm}
\end{table}

\begin{figure}[t!]
	\caption{Scatter Plots of LS Estimates}
	\label{fig:appl.2}	
	\begin{center}
		\begin{tabular}{cc@{\hspace*{1.5cm}}cc}
			\rotatebox{90}{\hspace*{2cm} Teacher $\widehat{\beta}_j^{\text{ls}}$ } & \includegraphics[width=.35\textwidth]{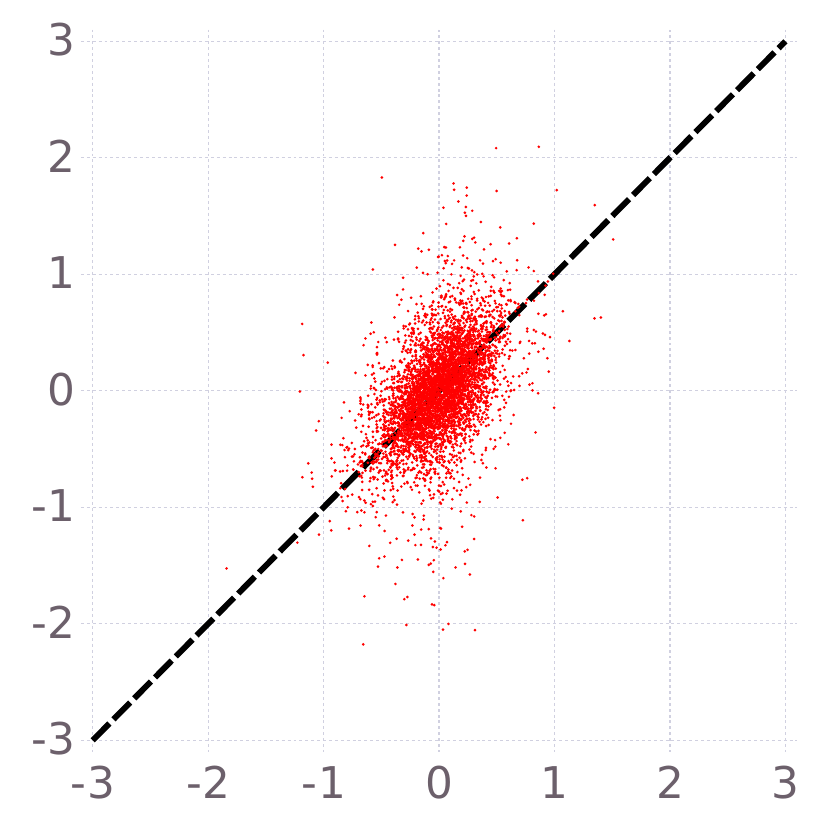} & \rotatebox{90}{\hspace*{2cm} Student $\widehat{\mu}_j^{\text{ls}}$ } &
			\includegraphics[width=.35\textwidth]{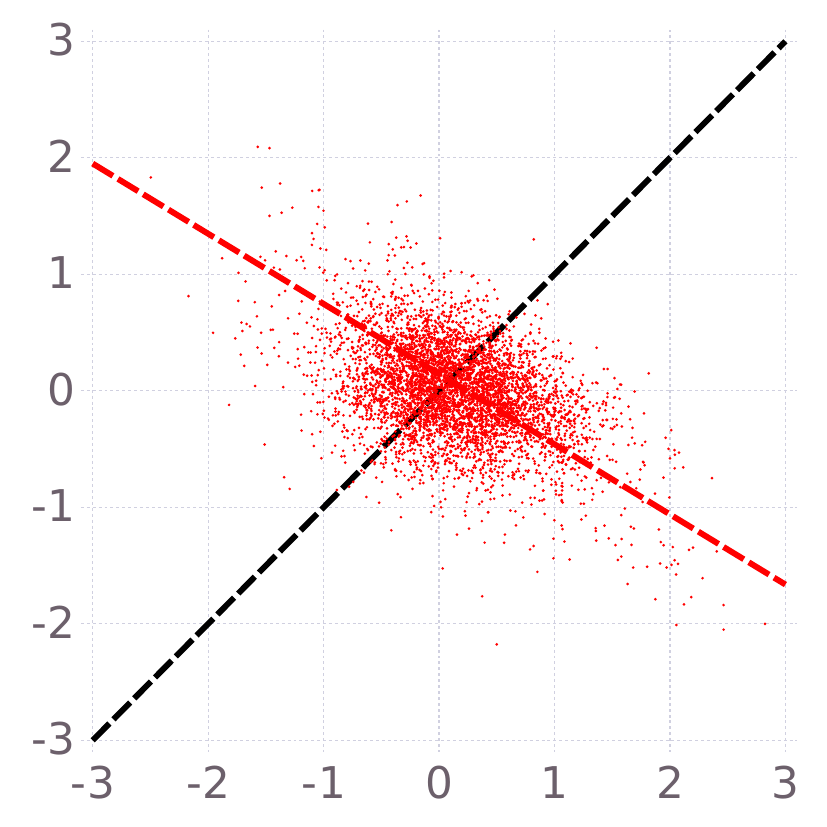} \\
			& Teacher $\widetilde{\beta}_j^{\text{ls}}$ & & Teacher $\widehat{\beta}_j^{\text{ls}}$
		\end{tabular} 
	\end{center}
	{\footnotesize {\em Notes:} 
		Red line has a regression coefficient of $-0.61$. Black lines are the 45-degree lines.}\setlength{\baselineskip}{4mm}
\end{figure}

Weak connectivity generates instability and imprecision of LS estimates. To illustrate the instability, we compare the LS estimator $\hat{\bs{\beta}}^{\text{ls}}$ to an estimator $\tilde{\bs{\beta}}^{\text{ls}}$ that is obtained by replacing $\alpha_i$ by an index function $\bs{\alpha}_x' \bs{x}_i$, where the regressors $\bs{x}_i$ are time-invariant demographic variables: gender, ethnicity, economically disadvantaged, English learner; see Section~\ref{sec:details.empirics.1way} of the  Supplementary Online Appendix for more details. The left panel of Figure~\ref{fig:appl.2} shows a scatter plot of the two estimators and documents that $\hat{\beta}_j^{\text{ls}}$ is much more variable than $\tilde{\beta}_j^{\text{ls}}$. The empirical variances of $\hat{\beta}_j^{\text{ls}}$ and of $\tilde{\beta}_j^{\text{ls}}$ are $0.15$ and $0.063$, respectively. The right panel shows a scatter plot of $\hat{\beta}_j^{\text{ls}}$ and the average ability of students taught by teacher $j$, previously denoted by  $\hat{\mu}_j^{\text{ls}}$ and defined in \eqref{eq:def.muj}. The slope of a linear regression line is -0.61. While this pattern seems to suggest negative assortative matching between students and teachers, it is likely an artifact of the graph's weak connectivity. In fact, in the Monte Carlo simulation we observed a similar pattern in the second and fourth top-row panels of Figure~\ref{fig:mc.sorting}, where it was a limited mobility bias.

\subsection{EB-URE Estimation}
\label{subsec:appl.ure}

\begin{figure}[t!]
	\caption{Scatter Plots of EB-URE Estimates}
	\label{fig:appl.3}	
	\begin{center}
		\begin{tabular}{cc@{\hspace*{1.5cm}}cc}
			\rotatebox{90}{\hspace*{2cm}Student $\widehat{\mu}_j^{\text{ure}}$} &\includegraphics[width=.35\textwidth]{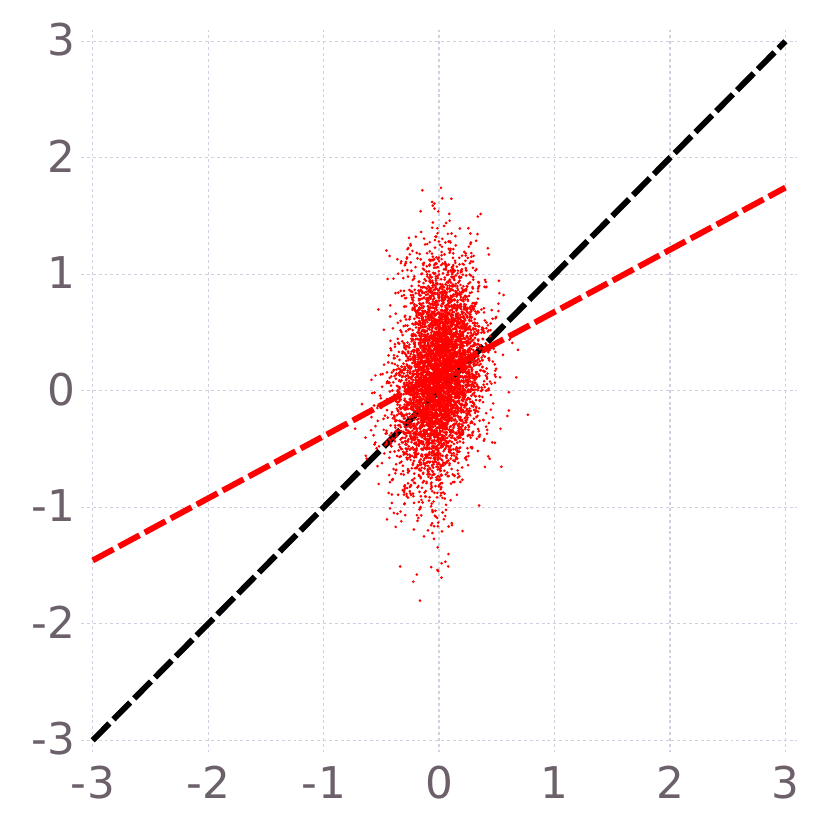} & 			\rotatebox{90}{\hspace*{0.4cm}School-average Student $\widehat{\mu}^{\text{ure}}$} &
			\includegraphics[width=.35\textwidth]{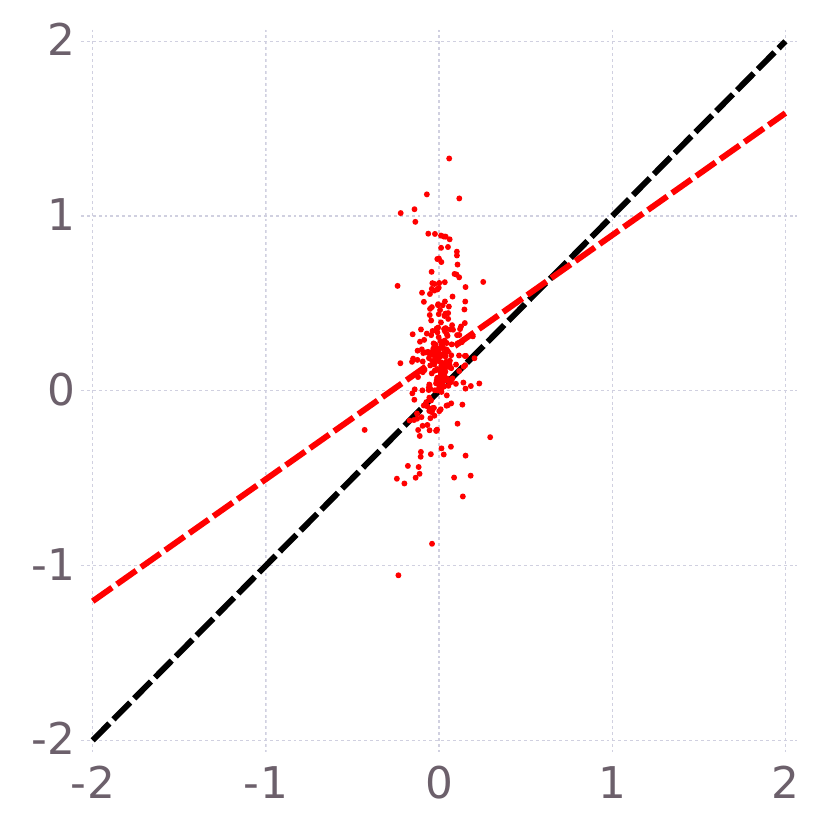} \\
			& Teacher $\widehat{\beta}_j^{\text{ure}}$ & & School-average Teacher $\widehat{\beta}^{\text{ure}}$
		\end{tabular} 
	\end{center}
	{\footnotesize {\em Notes:} 
			The red line from the left (right) panel has a regression coefficient of $0.533$ ($0.698$). Black lines are the 45-degree lines.}\setlength{\baselineskip}{4mm}
\end{figure}

To construct our proposed EB estimator, we adopt the covariate-based prior of Section~\ref{subsec:empbayes.altpriormean} where the student effects $\bs{\alpha}$ are shrunk to the location $Z_a \bs{\delta}_a$ and $Z_a$ contains the same student demographics as $\bs{x}_i$ in the definition of $\tilde{\bs{\beta}}^{\text{ls}}$ above. We shrink the teacher effects to $0$, i.e., $\bs{\delta}_b=0$ under the normalization. The parameters $(\lambda_\alpha,\lambda_\beta,\phi,\bs{\delta}_a)$ are jointly determined via URE-minimization using $W=W_b$. This produces our proposed EB-URE estimates that targets the $\bs{\beta}$ estimation loss. The estimated scale hyperparameters are $(\lambda_a^{\text{ure}},\lambda_b^{\text{ure}},\phi^{\text{ure}})=(0.075,1.67,0.7)$. The estimate of $\bs{\delta}_a$ is reported in Section~\ref{sec:details.empirics.ure} of the Supplementary Online Appendix.

The left panel of Figure~\ref{fig:appl.3} plots $\hat{\beta}_j^{\text{ure}}$ against $\hat{\mu}^{\text{ure}}_j$. The regression coefficient is now $0.533$, which implies positive student-teacher matching as opposed to the negative matching implied by the LS estimates. The magnitude of matching patterns is obscured by the difference in the scale of variances in the $\alpha$ and $\beta$ dimension. Thus, we also provide the empirical correlation of $\hat{\alpha}_i^{\text{ure}}$ and $\hat{\beta}_{j(i,t)}^{\text{ure}}$ in Table~\ref{tab:appl.mom}, which is computed to be $0.12$. In combination with the right panel of Figure~\ref{fig:appl.2} the empirical results are consistent with the Monte Carlo results in Section~\ref{sec:simul}, demonstrating that the EB-URE is able to accurately pick up on matching patterns otherwise undetected by the LS estimates due to weak connectivity. We also plot the school averages of $\hat{\beta}_j^{\text{ure}}$ against $\hat{\alpha}_j^{\text{ure}}$ in the right panel of Figure~\ref{fig:appl.3} with a regression coefficient of $0.698$ and empirical correlation of $0.188$, suggesting that part of the observed positive student-teacher matching pattern arises from assortative matching between schools.

\begin{table}
	\caption{Empirical Moments}
	\label{tab:appl.mom}
	\begin{center}
		\begin{tabular}{lccc} 
			\toprule
			& $\text{var}(\hat{\alpha}_i)$ & $\text{var}(\hat{\beta}_{j(i,t)})$ & $\text{cor}(\hat{\alpha}_i,\hat{\beta}_{j(i,t)})$ \\
			\midrule 
			\;LS & 0.76 & 0.15 & -0.28 \\ 
			\;EB-URE & 0.62 & .034 & 0.12 \\
			\bottomrule
		\end{tabular}
	\end{center}
\end{table}

\begin{figure}[t!]
	\caption{Scatter Plots of LS, EB-URE, and EB-1way Estimates}
	\label{fig:appl.4}	
	\begin{center}
		\begin{tabular}{cc@{\hspace*{1.5cm}}cc}
			\rotatebox{90}{\hspace*{2cm} Teacher $\widehat{\beta}_j^{\text{ure}}$} &
			\includegraphics[width=.35\textwidth]{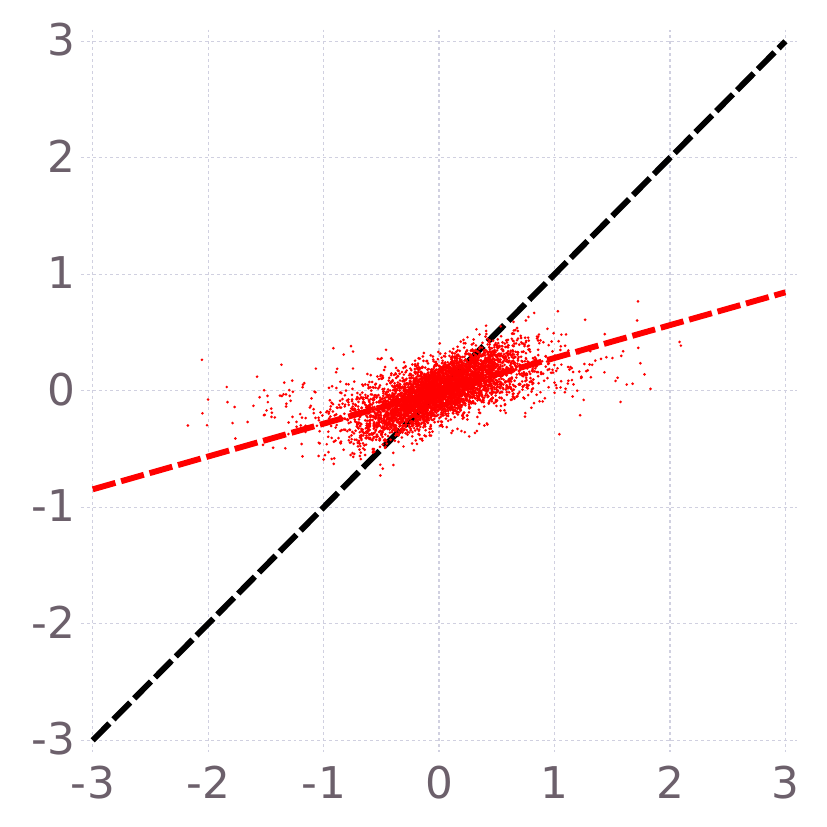} & \rotatebox{90}{\hspace*{2cm} Teacher $\widehat{\beta}_j^{\text{ure}}$} &
			\includegraphics[width=.35\textwidth]{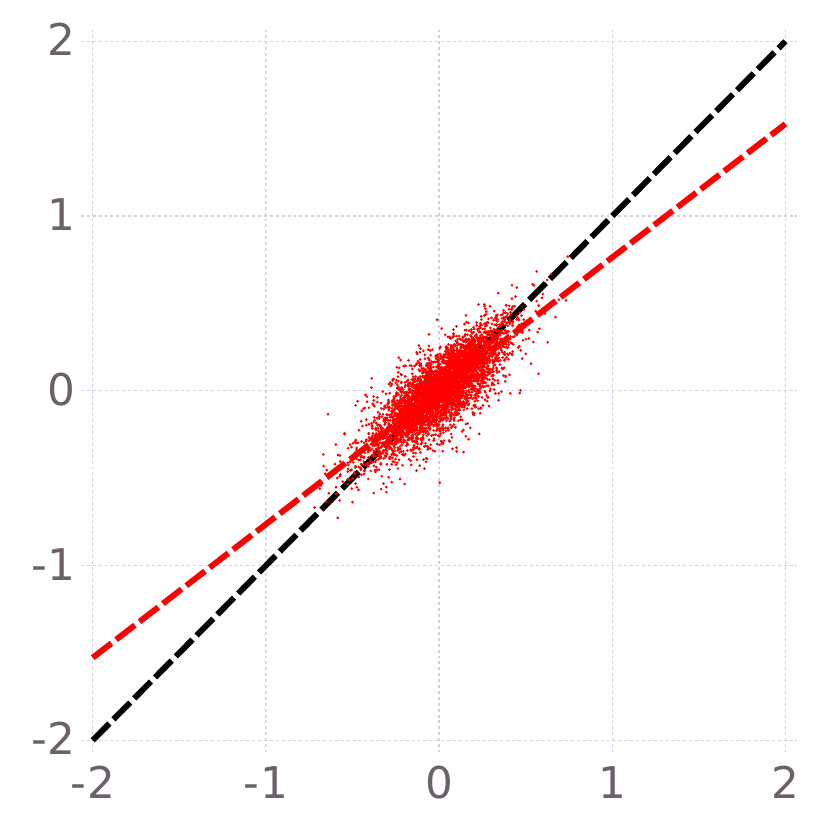} \\
			& Teacher $\widehat{\beta}_j^{\text{ls}}$ &  & Teacher $\widehat{\beta}_j^{\text{1-way}}$ 
		\end{tabular} 
	\end{center}
	{\footnotesize {\em Notes:} 
		The red line from the left (right) panel has a regression coefficient of $0.282$ ($0.764$). Black lines are the 45-degree lines.}\setlength{\baselineskip}{4mm}
\end{figure}

Figure~\ref{fig:appl.4} plots the EB-URE estimates against LS estimates and against $\hat{\bs{\beta}}^{\text{1-way}}$, the standard shrinkage estimator from a one-way model that  assumes $\alpha_i = \bs{\alpha}_x' \bs{x}_i$. The latter is commonly used in the literature and a simplified version was previously defined in (\ref{eq:betahat1way}). Section~\ref{sec:details.empirics.1way} of the Supplementary Online Appendix provides a detailed explanation of its construction for this application. We refer to it as EB-1way consistent with the naming in the simulations of Section~\ref{sec:simul}. While the EB-URE and LS estimators are positively correlated (left panel of the figure), by construction the EB-URE estimator is more stable and has a five-times smaller variance. Numerical values for the sample variances are reported in Table~\ref{tab:appl.mom}. 

The right panel of Figure~\ref{fig:appl.4} shows that the EB-1way estimates are less volatile than the LS estimates, but still more volatile than the two-way EB-URE estimates. To explore the extent to which the demographic regressors $\bs{x}_i$ explain the student heterogeneity, we project $\hat{\alpha}_i^{\text{ls}}$ onto $\bs{x}_i$, where we retrieve an $R^2$ of $0.21$. This is relatively small and indicates that most of the variation in student heterogeneity is not captured by the $\bs{x}_i$ regressors. In conjunction with the usual concern regarding assortative matching, this  provides an additional justification for adopting the two-way effects model.

\subsection{Teacher Rankings}
\label{subsec:appl.rankings}

\begin{table}[t!]
	\caption{Ranking of Teachers by Quintile}
	\label{tab:appl.bin}
	\begin{center}
		\begin{tabular}{cc}
			\begin{minipage}{0.45\textwidth}
				\centering
				\begin{tabular}{lccccc}
					\toprule
					& \multicolumn{5}{c}{EB-URE} \\
					LS & 1st & 2nd & 3rd & 4th & 5th\\ 
					\midrule
					1st & 631 & 253 & 105 & 54 & 23 \\
					2nd & 293 & 351 & 246 & 132 & 44 \\
					3rd & 89 & 262 & 331 & 270 & 115 \\
					4th & 37 & 137 & 252 & 343 & 297 \\
					5th & 16 & 63 & 133 & 267 & 588 \\
					\bottomrule
				\end{tabular}
				\vspace{1mm}
			\end{minipage}
			&
			\begin{minipage}{0.45\textwidth}
				\centering
				\begin{tabular}{lccccc}
					\toprule
					& \multicolumn{5}{c}{EB-URE} \\
					EB-1way & 1st & 2nd & 3rd & 4th & 5th\\ 
					\midrule
					1st & 693 & 273 & 76 & 21 & 3 \\
					2nd & 260 & 392 & 286 & 103 & 25 \\
					3rd & 82 & 261 & 385 & 269 & 70 \\
					4th & 27 & 114 & 233 & 429 & 263 \\
					5th & 4 & 26 & 87 & 244 & 706 \\
					\bottomrule
				\end{tabular}
				\vspace{1mm}
			\end{minipage}
		\end{tabular}
	\end{center}
\end{table}

A frequent use of value-added estimates is to rank teachers for remuneration purposes. We conclude the empirical analysis with a comparison of the rankings based on $\hat{\bs{\beta}}^{\text{ls}}$, $\hat{\bs{\beta}}^{\text{ure}}$ and $\hat{\bs{\beta}}^{\text{1-way}}$. For each set of teacher value-added estimates, teachers are split into quintiles. We then take two sets of estimates and compute the number of teachers for each of the 25 quintile pairs. The results are reported in Table~\ref{tab:appl.bin}. For instance, the number of teachers who are in the first quntile of the $\hat{\bs{\beta}}^{\text{ls}}$ distribution and also in the second quintile of $\hat{\bs{\beta}}^{\text{ure}}$ distribution is 253. Note that the number of teachers per quintile is approximately 1,066. The diagonals thus represent the number of teachers for which two estimators agree in terms of quintile rank, the principal off-diagonals represent the numbers of disagreements by one rank, and so on. 

There is substantial variation in how $\hat{\bs{\beta}}^{\text{ls}}$ and $\hat{\bs{\beta}}^{\text{ure}}$ rank teachers. In particular, there are 39 teachers (the bottom left and top right corners of the left panel) who are simultaneously ranked at opposite quintiles by $\hat{\bs{\beta}}^{\text{ure}}$ and $\hat{\bs{\beta}}^{\text{ls}}$. While the rankings within the first and fifth quintiles are relatively more similar across $\hat{\bs{\beta}}^{\text{1-way}}$ and $\hat{\bs{\beta}}^{\text{ure}}$ at the extreme quintiles, there is still substantial variation in the middle quintiles. 
In view of the theoretical optimality of $\hat{\bs{\beta}}^{\text{ure}}$ and its adaptivity that were demonstrated in the simulations calibrated to fit this application (Design 1), we believe that EB-URE is an attractive alternative to these estimators. 


\section{Conclusion}
\label{sec:concl}

We develop an empirical Bayes estimator for two-way effects models using a novel prior that incorporates assortative matching patterns based on the graph normalized Laplacian matrix and hyperparameter selection based on the minimization of an unbiased risk estimate. We prove its asymptotic optimality, allowing for weakly connected graphs due to limited mobility. We show in a Monte Carlo study that the estimator dominates, in terms of RMSE, a number of competitors and is able to capture assortative matching. An application of our estimator to the NCERDC data set suggests that there is substantial student-teacher assortative matching which is ignored by both the common one-way shrinkage estimator and the two-way LS estimator. We also find that teacher rankings based on value-added estimates are quite sensitive to the methodology employed.

\setstretch{1.1}
\bibliography{Empirical_Bayes}

@article{Gu2017b,
author = {Gu, Jiaying and Koenker, Roger},
doi = {10.1080/07350015.2015.1052457},
issn = {0735-0015},
journal = {Journal of Business \& Economic Statistics},
month = {jan},
number = {1},
pages = {1--16},
title = {{Unobserved Heterogeneity in Income Dynamics: An Empirical Bayes Perspective}},
url = {https://www.tandfonline.com/doi/full/10.1080/07350015.2015.1052457},
volume = {35},
year = {2017}
}

@article{Jiang2009,
author = {Jiang, Wenhua and Zhang, Cun-Hui},
doi = {10.1214/08-AOS638},
issn = {0090-5364},
journal = {The Annals of Statistics},
month = {aug},
number = {4},
title = {{General Maximum Likelihood Empirical Bayes Estimation of Normal Means}},
url = {https://projecteuclid.org/journals/annals-of-statistics/volume-37/issue-4/General-maximum-likelihood-empirical-Bayes-estimation-of-normal-means/10.1214/08-AOS638.full},
volume = {37},
year = {2009}
}

@article{Kwon2021,
author = {Kwon, Soonwoo},
journal = {Working Paper, Brown University},
title = {{Optimal Shrinkage Estimation of Fixed Effects in Linear Panel Data Models}},
year = {2021}
}

@incollection{Robbins1956,
author = {Robbins, Herbert},
booktitle = {Proceedings of the Third Berkeley Symposium on Mathematical Statistics and Probability},
publisher = {University of California Press, Berkeley and Los Angeles},
title = {{An Empirical Bayes Approach to Statistics}},
year = {1956}
}

@article{Gu2017a,
author = {Gu, Jiaying and Koenker, Roger},
doi = {10.1002/jae.2530},
issn = {08837252},
journal = {Journal of Applied Econometrics},
month = {apr},
number = {3},
pages = {575--599},
title = {{Empirical Bayesball Remixed: Empirical Bayes Methods for Longitudinal Data}},
url = {https://onlinelibrary.wiley.com/doi/10.1002/jae.2530},
volume = {32},
year = {2017}
}

@article{Kline2020,
author = {Kline, Patrick and Saggio, Raffaele and S{\o}lvsten, Mikkel},
doi = {10.3982/ECTA16410},
issn = {0012-9682},
journal = {Econometrica},
number = {5},
pages = {1859--1898},
title = {{Leave‐Out Estimation of Variance Components}},
url = {https://www.econometricsociety.org/doi/10.3982/ECTA16410},
volume = {88},
year = {2020}
}

@article{Armstrong2022,
	author = {Armstrong, Timothy B. and Kolesár, Michal and Plagborg-Møller, Mikkel},
	title = {Robust Empirical Bayes Confidence Intervals},
	journal = {Econometrica},
	volume = {90},
	number = {6},
	pages = {2567-2602},
	doi = {https://doi.org/10.3982/ECTA18597},
	url = {https://onlinelibrary.wiley.com/doi/abs/10.3982/ECTA18597},
	eprint = {https://onlinelibrary.wiley.com/doi/pdf/10.3982/ECTA18597},
	year = {2022}
}

@article{Chetty2014,
author = {Chetty, Raj and Friedman, John N. and Rockoff, Jonah E.},
doi = {10.1257/aer.104.9.2593},
issn = {0002-8282},
journal = {American Economic Review},
month = {sep},
number = {9},
pages = {2593--2632},
title = {{Measuring the Impacts of Teachers I: Evaluating Bias in Teacher Value-Added Estimates}},
url = {https://pubs.aeaweb.org/doi/10.1257/aer.104.9.2593},
volume = {104},
year = {2014}
}

@article{Verdier2020,
author = {Verdier, Valentin},
doi = {10.1162/rest_a_00807},
issn = {0034-6535},
journal = {The Review of Economics and Statistics},
month = {mar},
number = {1},
pages = {1--16},
title = {{Estimation and Inference for Linear Models with Two-Way Fixed Effects and Sparsely Matched Data}},
url = {https://direct.mit.edu/rest/article/102/1/1-16/58545},
volume = {102},
year = {2020}
}

@article{Xie2012,
author = {Xie, Xianchao and Kou, S. C. and Brown, Lawrence D.},
doi = {10.1080/01621459.2012.728154},
issn = {01621459},
journal = {Journal of the American Statistical Association},
number = {500},
pages = {1465--1479},
title = {{SURE Estimates for a Heteroscedastic Hierarchical Model}},
volume = {107},
year = {2012}
}

@techreport{Gilraine2020,
address = {Cambridge, MA},
author = {Gilraine, Michael and Gu, Jiaying and McMillan, Robert},
doi = {10.3386/w27094},
institution = {National Bureau of Economic Research},
month = {may},
title = {{A New Method for Estimating Teacher Value-Added}},
url = {http://www.nber.org/papers/w27094.pdf},
year = {2020}
}

@article{Stein1981,
author = {Stein, Charles},
journal = {Annals of Statistics},
number = {6},
pages = {1135--1151},
title = {{Estimation of the Mean of a Multivariate Normal Distribution}},
volume = {9},
year = {1981}
}

@article{Chetty2018,
author = {Chetty, Raj and Hendren, Nathaniel},
doi = {10.1093/qje/qjy006},
issn = {0033-5533},
journal = {The Quarterly Journal of Economics},
month = {aug},
number = {3},
pages = {1163--1228},
title = {{The Impacts of Neighborhoods on Intergenerational Mobility II: County-Level Estimates}},
url = {https://academic.oup.com/qje/article/133/3/1163/4850659},
volume = {133},
year = {2018}
}

@article{James1961,
author = {James, W and Stein, Charles},
journal = {Proc. 4th Berkeley Sympos. Math. Statist. and Prob.},
pages = {261--279},
title = {{Estimation with Quadratic Loss}},
volume = {I},
year = {1961}
}

@article{Brown2009,
author = {Brown, Lawrence D. and Greenshtein, Eitan},
doi = {10.1214/08-AOS630},
issn = {0090-5364},
journal = {The Annals of Statistics},
month = {aug},
number = {4},
title = {{Nonparametric Empirical Bayes and Compound Decision Approaches to Estimation of a High-Dimensional Vector of Normal Means}},
url = {https://projecteuclid.org/journals/annals-of-statistics/volume-37/issue-4/Nonparametric-empirical-Bayes-and-compound-decision-approaches-to-estimation-of/10.1214/08-AOS630.full},
volume = {37},
year = {2009}
}

@techreport{Kane2008,
address = {Cambridge, MA},
author = {Kane, Thomas and Staiger, Douglas},
doi = {10.3386/w14607},
institution = {National Bureau of Economic Research},
month = {dec},
title = {{Estimating Teacher Impacts on Student Achievement: An Experimental Evaluation}},
url = {http://www.nber.org/papers/w14607.pdf},
year = {2008}
}

@article{Liu2020,
author = {Liu, Laura and Moon, Hyungsik Roger and Schorfheide, Frank},
doi = {10.3982/ECTA14952},
issn = {0012-9682},
journal = {Econometrica},
number = {1},
pages = {171--201},
title = {{Forecasting With Dynamic Panel Data Models}},
url = {https://www.econometricsociety.org/doi/10.3982/ECTA14952},
volume = {88},
year = {2020}
}

@article{Brown2018,
author = {Brown, Lawrence D. and Mukherjee, Gourab and Weinstein, Asaf},
doi = {10.1214/17-AOS1599},
issn = {0090-5364},
journal = {The Annals of Statistics},
month = {aug},
number = {4},
title = {{Empirical Bayes Estimates for a Two-Way Cross-Classified Model}},
volume = {46},
year = {2018}
}

@article{Jochmans2019,
archivePrefix = {arXiv},
arxivId = {1608.01532},
author = {Jochmans, Koen and Weidner, Martin},
doi = {10.3982/ecta14605},
eprint = {1608.01532},
issn = {0012-9682},
journal = {Econometrica},
number = {5},
pages = {1543--1560},
title = {{Fixed‐Effect Regressions on Network Data}},
volume = {87},
year = {2019}
}

@article{Card2013,
author = {Card, David and Heining, J{\"{o}}rg and Kline, Patrick},
doi = {10.1093/qje/qjt006},
issn = {0033-5533},
journal = {The Quarterly Journal of Economics},
month = {aug},
number = {3},
pages = {967--1015},
title = {{Workplace Heterogeneity and the Rise of West German Wage Inequality}},
url = {https://academic.oup.com/qje/article/128/3/967/1848785},
volume = {128},
year = {2013}
}

@article{Abowd1999,
author = {Abowd, John M. and Kramarz, Francis and Margolis, David N.},
doi = {10.1111/1468-0262.00020},
issn = {0012-9682},
journal = {Econometrica},
month = {mar},
number = {2},
pages = {251--333},
title = {{High Wage Workers and High Wage Firms}},
url = {http://doi.wiley.com/10.1111/1468-0262.00020},
volume = {67},
year = {1999}
}

@article{Finkelstein2016,
author = {Finkelstein, Amy and Gentzkow, Matthew and Williams, Heidi},
doi = {10.1093/qje/qjw023},
issn = {0033-5533},
journal = {The Quarterly Journal of Economics},
month = {nov},
number = {4},
pages = {1681--1726},
title = {{Sources of Geographic Variation in Health Care: Evidence From Patient Migration}},
url = {https://academic.oup.com/qje/article/131/4/1681/2468872},
volume = {131},
year = {2016}
}

@article{Searle1966,
author = {Searle, S.},
journal = {Tech. Rep. BU-213-M, Cornell University, Biometrics Unit},
title = {{Estimable Functions and Testable Hypothesis in Linear Models}},
year = {1966}
}

@article{Kramarz2008,
author = {Kramarz, Francis and Machin, Stephen J. and Ouazad, Amine},
doi = {10.2139/ssrn.1283643},
issn = {1556-5068},
journal = {SSRN Electronic Journal},
title = {{What Makes a Test Score? The Respective Contributions of Pupils, Schools, and Peers in Achievement in English Primary Education}},
url = {http://www.ssrn.com/abstract=1283643},
year = {2008}
}

@book{Robert1994,
author = {Robert, Christian},
publisher = {Springer Verlag, New York},
title = {{The Bayesian Choice}},
year = {1994}
}

@incollection{Stein1956,
author = {Stein, Charles},
booktitle = {Proceedings of the Third Berkeley Symposium on Mathematical Statistics and Probability},
publisher = {The Regents of the University of California},
title = {{Inadmissibility of the Usual Estimator for the Mean of a Multivariate Normal Distribution}},
year = {1956}
}

@article{Mansfield2015,
author = {Mansfield, Richard K.},
doi = {10.1086/679683},
issn = {0734-306X},
journal = {Journal of Labor Economics},
month = {jul},
number = {3},
pages = {751--788},
title = {{Teacher Quality and Student Inequality}},
url = {https://www.journals.uchicago.edu/doi/10.1086/679683},
volume = {33},
year = {2015}
}

@article{Andrews1992,
	ISSN = {02664666, 14694360},
	URL = {http://www.jstor.org/stable/3532442},
	author = {Donald W. K. Andrews},
	journal = {Econometric Theory},
	number = {2},
	pages = {241--257},
	publisher = {Cambridge University Press},
	title = {Generic Uniform Convergence},
	urldate = {2025-06-18},
	volume = {8},
	year = {1992}
}

@incollection{KLINE2024,
	title = {Chapter 2 - Firm wage effects},
	editor = {Christian Dustmann and Thomas Lemieux},
	series = {Handbook of Labor Economics},
	publisher = {Elsevier},
	volume = {5},
	pages = {115-181},
	year = {2024},
	issn = {1573-4463},
	doi = {https://doi.org/10.1016/bs.heslab.2024.11.005},
	url = {https://www.sciencedirect.com/science/article/pii/S1573446324000051},
	author = {Patrick Kline}
}

@TechReport{AbowdEtal2002,
	author={John M. Abowd and Robert H. Creecy and Francis Kramarz},
	title={{Computing Person and Firm Effects Using Linked Longitudinal Employer-Employee Data}},
	year=2002,
	month=Mar,
	institution={Center for Economic Studies, U.S. Census Bureau},
	type={},
	url={https://ideas.repec.org/p/cen/tpaper/2002-06.html},
	number={2002-06},
	doi={},
}

@incollection{WALTERS2024,
	title = {Chapter 3 - Empirical Bayes methods in labor economics},
	editor = {Christian Dustmann and Thomas Lemieux},
	series = {Handbook of Labor Economics},
	publisher = {Elsevier},
	volume = {5},
	pages = {183-260},
	year = {2024},
	issn = {1573-4463},
	doi = {https://doi.org/10.1016/bs.heslab.2024.11.001},
	url = {https://www.sciencedirect.com/science/article/pii/S1573446324000014},
	author = {Christopher Walters}
}

@article{Gu2023,
	author = {Gu, Jiaying and Koenker, Roger},
	title = {Invidious Comparisons: Ranking and Selection as Compound Decisions},
	journal = {Econometrica},
	volume = {91},
	number = {1},
	pages = {1-41},
	doi = {https://doi.org/10.3982/ECTA19304},
	url = {https://onlinelibrary.wiley.com/doi/abs/10.3982/ECTA19304},
	eprint = {https://onlinelibrary.wiley.com/doi/pdf/10.3982/ECTA19304},
	year = {2023}
}

@article{Bonhomme2019,
	author = {Stéphane Bonhomme and Martin Weidner},
	title = {Posterior Average Effects},
	journal = {Journal of Business \& Economic Statistics},
	volume = {40},
	number = {4},
	pages = {1849--1862},
	year = {2022},
	publisher = {ASA Website},
	doi = {10.1080/07350015.2021.1984928},
}

@article{Graham2008,
author = {Graham, Bryan S.},
title = {Identifying Social Interactions Through Conditional Variance Restrictions},
journal = {Econometrica},
volume = {76},
number = {3},
pages = {643-660},
doi = {https://doi.org/10.1111/j.1468-0262.2008.00850.x},
url = {https://onlinelibrary.wiley.com/doi/abs/10.1111/j.1468-0262.2008.00850.x},
eprint = {https://onlinelibrary.wiley.com/doi/pdf/10.1111/j.1468-0262.2008.00850.x},
year = {2008}
}

@TechReport{HeRobin2025,
author={Junnan He and Jean-Marc Robin},
title={{Ridge Estimation of High Dimensional Two-Way Fixed Effect Regression }},
year=2025,
month=May,
institution={Sciences Po},
type={},
url={},
number={},
doi={},
}

@incollection{Bonhomme2020,
	title = {Chapter 5 - Econometric analysis of bipartite networks},
	editor = {Bryan Graham and Áureo {de Paula}},
	booktitle = {The Econometric Analysis of Network Data},
	publisher = {Academic Press},
	pages = {83-121},
	year = {2020},
	isbn = {978-0-12-811771-2},
	doi = {https://doi.org/10.1016/B978-0-12-811771-2.00011-0},
	url = {https://www.sciencedirect.com/science/article/pii/B9780128117712000110},
	author = {St\'{e}phane  Bonhomme}
}

@misc{Bonhomme2024,
	title={Estimating Heterogeneous Effects: Applications to Labor Economics}, 
	author={Stephane Bonhomme and Angela Denis},
	year={2024},
	eprint={2404.01495},
	archivePrefix={arXiv},
	primaryClass={econ.EM},
	url={https://arxiv.org/abs/2404.01495},
}

@Article{Liu2023,
  author  = {Liu, Laura and Moon, Hyungsik Roger and Schorfheide, Frank},
  journal = {Quantitative Economics},
  title   = {Forecasting with a Panel Tobit Model},
  year    = {2023},
  number  = {1},
  pages   = {117-159},
  volume  = {14},
}

@Article{Liu2023a,
  author  = {Liu, Laura},
  journal = {Journal of Business \& Economics Statistics},
  title   = {Density Forecasts in Panel Data Models: A Semiparametric Bayesian Perspective},
  year    = {2023},
  number  = {2},
  pages   = {349-363},
  volume  = {41},
}
\setstretch{1.3}


\begin{appendix}
	
\setstretch{1.25}

\section{Proofs of Main Results}

For a matrix $M\in\mathbb{R}^{n\times n}$,
let $\bar{\rho}(M)$ denote its \emph{spectral radius}, the largest of the magnitudes of its $n$ eigenvalues. Let $\lVert M \rVert$ denote its \emph{spectral norm}, its largest singular value. When $M$ is positive semi-definite, then $\rho_n(M)=\bar{\rho}(M) = ||M||$.

\subsection{Assumptions with General Weighting Matrix}
\label{subsubsec:pw.ass}

For the clarity of graphical interpretations of the Assumptions, theoretical results are given for the leading case $W=W_{a+b}$ and $W_{b}$ only in the paper. The proofs in this Appendix are for a general weighting matrix under Assumptions~\ref{ass:graph.connect.1} and \ref{ass:graph.connect.2} below. We first show that Assumptions~\ref{ass:graph.connect.1} and \ref{ass:graph.connect.2} are implied by Assumptions \ref{ass:graph.connect.1.full} -- \ref{ass:graph.connect.2.c} when $W=W_{a+b}$ and $W=W_{b} $. Subsequently, we prove Theorem~\ref{thm:pw} under Assumption~\ref{ass:graph.connect.1} and prove Lemma~\ref{lm:unif.convg}, Theorem~\ref{thm:main}, and Corollary \ref{cor:reg} under Assumptions~\ref{ass:graph.connect.1} and \ref{ass:graph.connect.2}, while keeping Assumptions~\ref{ass:reg.cond} and \ref{ass.degree} throughout.

\begin{assumption}
	\emph{(i).}\label{ass:graph.connect.1}
	$(r+c)\cdot \bar{\rho}^2(W^{1/2}L^-W^{1/2}) = o(1)$, \newline
	\emph{(ii).}
	$ (r+c)\cdot \bar{\rho}(W) \le M$ for some finite $M$, \newline
	\emph{(iii).} $ \epsilon <  (r+c)\cdot\bar{\rho}(W^{1/2}L^-W^{1/2})$ for some $\epsilon>0$.
\end{assumption}

\begin{assumption}
	\label{ass:graph.connect.2}
	\emph{(i).} 	$(r+c) \cdot \tr( [W^{1/2}L^- W^{1/2}]^2 )= O(1)$, \newline
	\emph{(ii).} $(r+c)^{-1}\cdot\tr([L^\dagger]^{\epsilon} )= O(1)$ for some $\epsilon>0$.
\end{assumption}

\begin{lemma} \emph{(i).} \label{Lemma: A- eig1}
	Assumption~\ref{ass:graph.connect.1} is implied by Assumption~\ref{ass:graph.connect.1.full} when $W=W_{a+b}$; and  by Assumption~\ref{ass.degree} and \ref{ass:graph.connect.1.b} when $W=W_{b}$. \newline \emph{(ii).} Assumption~\ref{ass:graph.connect.2} is implied by Assumptions~\ref{ass:graph.connect.1.full} and \ref{ass:graph.connect.2.a} when $W=W_{a+b}$; and
	by Assumptions~\ref{ass:graph.connect.1.b}, \ref{ass:graph.connect.2.a},  \ref{ass:graph.connect.2.b}, and \ref{ass:graph.connect.2.c} when $W=W_{b}$.
\end{lemma}

\begin{remark}
	Because $W$ builds in the normalization, such as $1/(r+c)$ for $W_{a+b}$ and $1/c$ for $W_{b}$, Assumption~\ref{ass:graph.connect.1}(ii) simply requires the standard condition that the selection matrix taking the parameter of interest out of $\bs{\theta}$ has bounded eigenvalues.  Similarly, because $W$ builds in the normalization, Assumption~\ref{ass:graph.connect.1}(iii) essentially requires a finite upper bound for the eigenvalues of $L$ (or the selected counterpart of $L$). If we drop Assumption~\ref{ass:graph.connect.1}(iii), the convergence rate in Theorem~\ref{thm:pw} changes from $\delta_{r,c}^2$ to $\delta_{r,c}$, however, the pointwise convergence still holds and the uniform convergence and asymptotic optimality in Lemma~\ref{lm:unif.convg} and Theorem \ref{thm:main} are unaffected.
	
\end{remark}


\begin{remark}
	The trace condition in Assumption~\ref{ass:graph.connect.2} allows the number of unbounded eigenvalues of $W^{1/2}L^{-}W^{1/2}$ to go to infinity as $r,c \rightarrow \infty$. Assumptions~\ref{ass:graph.connect.2.a} and \ref{ass:graph.connect.2.b} only allow for an bounded number of small eigenvalues of $L$ and $L_{2,\perp}$ such that this trace condition is implied by that on the largest eigenvalues in Assumption~\ref{ass:graph.connect.1}. We can maintain Assumption ~\ref{ass:graph.connect.2} by letting the number of unbounded eigenvalues grow slowly and tightening the rate condition in Assumption ~\ref{ass:graph.connect.1}.
\end{remark}

\begin{proof}[Proof of Lemma \ref{Lemma: A- eig1}]
	We first study the case 
	$W= W_{a+b}=\frac{1}{r+c}I_{r+c}$. In this case, 
	\begin{equation}\label{gencond1}
		(r+c)\cdot\bar{\rho}^2(W^{1/2}L^-W^{1/2})
		=\frac{1}{r+c}\bar{\rho}^2(\mathcal{R}L^{\dagger}\mathcal{R}')
		\le M\frac{1}{(r+c)\cdot\rho_2^2(L)}
	\end{equation} 
	for some finite $M$ following $||\mathcal{R}||^2\le M$, see \eqref{eq:I1}, and the smallest non-zero eigenvalue of $L$, denoted by $\rho_2(L)$, is the largest eigenvalue of the Moore-Penrose generalized inverse $L^{\dagger}$. This verifies Assumption~\ref{ass:graph.connect.1}(i).
	Assumption~\ref{ass:graph.connect.1}(ii) holds because
	$(r+c)\bar{\rho}(W) = 1$.
	Assumption~\ref{ass:graph.connect.1}(iii) holds because $(r+c)\bar{\rho}(W^{1/2}L^-W^{1/2}) = \bar{\rho}(L^-)$ which is of the same order as $1/\rho_2(L)$, and $\rho_2(L)$ is bounded 
	under Assumption~\ref{ass.degree}.
	
	Assumption 	\ref{ass:graph.connect.2}(i) is implied by Assumption~\ref{ass:graph.connect.1.full} and Assumption~\ref{ass:graph.connect.2.a}, i.e., the number of diverging eigenvalues of $L^{\dagger}$ is finite. 
	Assumption \ref{ass:graph.connect.2}(ii) holds with $\epsilon=2$ under Assumption~\ref{ass:graph.connect.1.full} and Assumption~\ref{ass:graph.connect.2.a}.
	
	Next, we consider $W= W_b$, i.e., we are interested in the estimation of $\bs{\beta}$. In this case,
	\begin{equation}
		W^{1/2}L^-W^{1/2}
		=\frac{1}{c}
		\begin{bmatrix}
			0_{r\times r} & 0_{r\times c} \\
			0_{c\times r} & [\mathcal{R}L^\dagger \mathcal{R}']_{2}
		\end{bmatrix}
		=\frac{1}{c}
		\begin{bmatrix}
			0_{r\times r} & 0_{r\times c} \\
			0_{c\times r} & M_c[L_{2,\perp}]^\dagger M_c 
		\end{bmatrix},
		\label{eq:partition}
	\end{equation}
	where $ [\mathcal{R}L^\dagger \mathcal{R}']_{2}=M_c[L_{2,\perp}]^\dagger M_c$ follows from Lemma~\ref{lm:equalities}. Because the demeaning matrix $M_c$ has bounded eigenvalues, we know 
	\begin{equation} \label{eq: verifyB2}
		(r+c)\cdot\bar{\rho}^2(W^{1/2}L^-W^{1/2})
		\le M \frac{1}{c\cdot\rho_2^2(L_{2,\perp})},
	\end{equation}
	for some finite $M$, which shows that Assumption~\ref{ass:graph.connect.1}(i) is implied by Assumption~\ref{ass:graph.connect.1.b}. 
	Assumption~\ref{ass:graph.connect.1}(ii) holds because
	$(r+c)\bar{\rho}(W) = (r+c)/r=O(1)$ by Assumption~\ref{ass.degree}.
	Assumption~\ref{ass:graph.connect.1}(iii) holds because by the same arguments as in \eqref{eq: verifyB2},  $(r+c)\bar{\rho}(W^{1/2}L^-W^{1/2})$ is of the same order as $1/\rho_2(L_{2,\perp})$, and $\rho_2(L_{2,\perp})$ is bounded under Assumption~\ref{ass.degree}. 
	
	Assumption 	\ref{ass:graph.connect.2}(i) is implied by Assumption~\ref{ass:graph.connect.1.b} and Assumption~\ref{ass:graph.connect.2.b}, i.e., the number of diverging eigenvalues of $L_{2,\perp}^{\dagger}$ is finite. 
	Assumption \ref{ass:graph.connect.2}(ii) follows from Assumptions~\ref{ass:graph.connect.2.a} and \ref{ass:graph.connect.2.c}, i.e., for the full graph, the number of diverging eigenvalues of $L^{\dagger}$ is finite and its largest eigenvalue satisfy the rate condition in \ref{ass:graph.connect.2.c}.
\end{proof}

\subsection{Proof of Pointwise Convergence}

\begin{proof}[Proof of Theorem \ref{thm:pw}] Theorem~\ref{thm:pw} follows immediately from Lemma~\ref{lm:decomp}, Lemma~\ref{lm:UC1}, and the Cauchy-Schwarz inequality. 
\end{proof}

Without loss of generality, we set $\sigma^2=1$ in the proof.
To analyze the difference between $l_w(\hat{\bs{\theta}}(\bs{\lambda}),\bs{\theta})$ and $\text{URE}(\bs{\lambda})$,  define 
\begin{align}
	\mathcal{D}_1 (\bs{\lambda}) &= [\hat{\bs{\theta}}^{\text{ls}}-\bs{\theta}]^\prime W [\hat{\bs{\theta}}^{\text{ls}}-\bs{\theta}] -
	\tr[WL^- ] , \ \ 
	\mathcal{D}_2 (\bs{\lambda}) = (S \bs{v})^\prime W [\hat{\bs{\theta}}^{\text{ls}}-\bs{\theta}],
	\nonumber \\
	\mathcal{D}_3 (\bs{\lambda})&= [S\bs{\theta}]^\prime W [\hat{\bs{\theta}}^{\text{ls}}-\bs{\theta}], \ \ 
	\mathcal{D}_4 (\bs{\lambda}) = [S_1(\hat{\bs{\theta}}^{\text{ls}}-\bs{\theta})]^\prime W [\hat{\bs{\theta}}^{\text{ls}}-\bs{\theta}]
	- \tr[S_1L ^- W ].
\end{align}

\begin{lemma}[Decomposition of L -- URE]
	\label{lm:decomp}
	$l_w(\hat{\bs{\theta}}(\bs{\lambda}),\bs{\theta})-\text{URE}(\bs{\lambda})$ can be written as 
	\begin{equation*}
		\begin{aligned}
			l_w(\hat{\bs{\theta}}(\bs{\lambda}),\bs{\theta})-\text{URE}(\bs{\lambda})		
			=  -\mathcal{D}_1(\bs{\lambda}) +2 \mathcal{D}_2(\bs{\lambda})-2\mathcal{D}_3(\bs{\lambda}) +2 \mathcal{D}_4(\bs{\lambda}).
		\end{aligned}
	\end{equation*}
\end{lemma}

Define 
$
\delta_{r,c}^{-1}= (r+c)^{1/2} \cdot\bar{\rho}(W^{1/2}L^-W^{1/2}),
$ where $\delta_{r,c}^{-1}=o(1)$ under Assumption~\ref{ass:graph.connect.1}.
By \eqref{gencond1} and \eqref{eq:partition}, $\delta_{r,c}$ defined here is equivalent to those in the main paper with the multiplication of a constant bounded from above and below, because $r$ and $c$ are of the same asymptotic order under Assumption~\ref{ass.degree}. Because $\delta_{r,c}$  only appears as an asymptotic order in all results, we can ignore such a difference in the definition.

\begin{lemma}
	\label{lm:UC1}
	Suppose Assumptions~\ref{ass:reg.cond}, \ref{ass.degree}, and \ref{ass:graph.connect.1} hold. For $k=1, 2, 3, 4$, we have
	$\sup_{\bs{\lambda}\in\mathcal{J}}
	\mathbb{E}_{\bs{\theta}, B}
	\left[	\mathcal{D}_k ^2(\bs{\lambda}) \right]
	= O\left(
	\delta_{r,c}^{-2}
	\right).
	$
\end{lemma}

The following auxiliary results in Lemma \ref{lm:pw.equalities} are used in the proof of Lemma \ref{lm:UC1}.
\begin{lemma}
	
	\label{lm:pw.equalities}
	

	\noindent \emph{(a)}. 
	$	\hat{\bs{\theta}}^{\emph{ls}}-\bs{\theta} 
	=
	\mathcal{R} L^-B^\prime \bs{U}.
	$
	\label{eq:E8}

	\noindent \emph{(b).} $\text{For generic matrix $A$, }\mathbb{V}[\bs{U}^\prime A \bs{U}]
	\leq
	(\mathbb{E}[u_{it}^4] + 2\sigma^2\sqrt{\mathbb{E}[u_{it}^4] - \sigma^4})\tr[A'A]$.
	
	\noindent	\emph{(c)}.	$S_1 = \mathcal{R}B^\dagger G B,
	\text{ where } G:= B\Lambda^{*-1}B'[B\Lambda^{*-1}B' + I]^{-1},$ and $|| G || 
	\leq 1.$

\end{lemma}

\begin{proof}[Proof of Lemma \ref{lm:UC1}]
	Because $\mathcal{D}_1$ does not depend on $\bs{\lambda}$, we have
	\begin{equation}
		\begin{aligned}[b]
			\mathbb{E}_{\bs{\theta}, B}
			\left[	\mathcal{D}_1 ^2(\bs{\lambda}) \right]&= \mathbb{E}_{\bs{\theta}, B}
			\left[	\left(
			[\hat{\bs{\theta}}^{\text{ls}}-\bs{\theta}]^\prime W
			[\hat{\bs{\theta}}^{\text{ls}}-\bs{\theta}]
			-\tr[W L^-]
			\right)^2 \right]\\
			&=_{(1)}
			\mathbb{V}_{\bs{\theta}, B}
			\left[
			[\hat{\bs{\theta}}^{\text{ls}}-\bs{\theta}]^\prime W [\hat{\bs{\theta}}^{\text{ls}}-\bs{\theta}]
			\right]
			=_{(2)}
			\mathbb{V}_{\bs{\theta}, B}
			\left[\bs{U}^\prime B L^- W  L^- B^\prime \bs{U} \right]\\
			&\leq_{(3)}
			M
			\tr[L^- W L^- W] 
			\leq_{(4)}
			M\cdot 
			(r+c) \cdot\bar{\rho}^2(W^{1/2}L^-W^{1/2})= O\left(
			\delta_{r,c}^{-2}
			\right),
		\end{aligned}
	\end{equation}
	where $=_{(1)}$ uses the fact that the second moment of a mean-zero variable is its variance; the equality $=_{(2)}$ uses Lemma \ref{lm:pw.equalities}(a);  the inequality $\leq_{(3)}$ uses Lemma \ref{lm:pw.equalities}(b), and $\leq_{(4)}$ uses
	for $V_1\in\mathbb{R}^{n\times n}$ and psd $V_2\in\mathbb{R}^{n\times n}$, 
	$\left|\tr[V_1V_2] \right|\leq \lVert V \rVert\tr[V_2]$, an inequality used repeatedly below.

	Since $S+S_1 = \mathcal{R}$, we have
	\begin{equation}
		\mathcal{D}_2(\bs{\lambda}) =(S \bs{v})^\prime W \left[\hat{\bs{\theta}}^{\text{ls}}-\bs{\theta}\right] 
		= \bs{v}^\prime W \left[\hat{\bs{\theta}}^{\text{ls}}-\bs{\theta}\right]- 
		(S_1 \bs{v})^\prime W \left[\hat{\bs{\theta}}^{\text{ls}}-\bs{\theta}\right],
	\end{equation}
	where the first term on the right hand side satisfies
		\begin{align}\label{eq:UC1.pf.1} 
			\mathbb{E}_{\bs{\theta}, B}
			\left[
			\left(\bs{v}^\prime W [\hat{\bs{\theta}}^{\text{ls}}-\bs{\theta}]
			\right)^2\right] 
			&=_{(1)}
			\mathbb{V}_{\bs{\theta}, B}
			\left[
			\bs{v}^\prime W [\hat{\bs{\theta}}^{\text{ls}}-\bs{\theta}]
			\right] 
			=
			\bs{v}' W L^- W \bs{v} \nonumber\\
			&\leq_{(2)}
			\bs{v}^\prime W \bs{v} 
			\cdot 
			\bar{\rho}(W^{1/2}L^-W^{1/2}) 
			 \leq_{(3)} r\bar{\mu}^2  \bar{\rho}(W) \cdot \bar{\rho}(W^{1/2}L^-W^{1/2}) \nonumber\\
			&\leq_{(4)}  M \cdot \delta_{r,c}^{-2} \frac{1}{(r+c)\bar{\rho}(W^{1/2}L^-W^{1/2})} 
			=_{(5)} O(\delta_{r,c}^{-2}),
		\end{align}		
	where $=_{(1)}$ uses the fact that the second moment of a mean-zero variable is its variance,  $\leq_{(2)}$ uses 
	$\left|\bs{x}^\prime V \bs{x} \right|
	\leq \bs{x}^\prime \bs{x} \cdot \lVert V \rVert$, an inequality used repeatedly in the proofs, $\leq_{(3)}$ uses that $||v||^2 \le r \bar{\mu}^2$, $=_{(4)}$ follows from $r\bar{\rho}(W)=O(1)$ by Assumption~\ref{ass:graph.connect.1}(ii) and $=_{(5)}$ follows from Assumption~\ref{ass:graph.connect.1}(iii).
	The second term satisfies
	\begin{equation}
		\begin{aligned}[b]
			&\hspace{1.5em}
			\mathbb{E}_{\bs{\theta}, B}
			\left[	\left(
			(S_1\bs{v})^\prime W \left[\hat{\bs{\theta}}^{\text{ls}}-\bs{\theta}\right]
			\right)^2\right] 
			=_{(1)}
			\mathbb{V}_{\bs{\theta}, B}
			\left[
			(S_1\bs{v})^\prime W \left[\hat{\bs{\theta}}^{\text{ls}}-\bs{\theta}\right]
			\right] 
			=
			(S_1 \bs{v})^\prime 
			W L^- W 
			(S_1 \bs{v}) \\
			&=_{(2)}
			\bs{v}'
			B' G B^{\dagger'} \mathcal{R}' WL^-W \mathcal{R}  B^\dagger G B 
			\bs{v} 
			\leq
			||G||^2 ||B||^2 
			||\bs{v}||^2 \bar{\rho}(B^{\dagger'} \mathcal{R}' WL^-W \mathcal{R}  B^\dagger) \\
			&\leq_{(3)}
			M\cdot
			r \cdot \bar{\rho}(B^{\dagger'} \mathcal{R}' WL^-W \mathcal{R}  B^\dagger)
			= M \cdot
			r \cdot \bar{\rho}^2(W^{1/2} L^- W^{1/2})=O\left(
			\delta_{r,c}^{-2}
			\right),
		\end{aligned} \label{A.31}
	\end{equation}
	where  $=_{(1)}$ uses the fact that the second moment of a mean-zero random variable is its variance, $=_{(2)}$ uses Lemma \ref{lm:pw.equalities}(c), $\leq_{(3)}$ holds for a finite constant $M$ using $||B||$ is bounded by Assumption~\ref{ass.degree}, $||G||\le 1$ by Lemma \ref{lm:pw.equalities}(d), $||v||^2 \le r \bar{\mu}^2$, and the final $	\delta_{r,c}^{-2}$ terms follows from $r/c$ is bounded from above and below by Assumption~\ref{ass.degree}.

	The proof for $\mathcal{D}_3(\bs{\lambda})$ follows through the same argument as that for 
	$\mathcal{D}_2$ with $v$ replaced by $\theta$, $||v||^2 \le r \bar{\mu}^2$ replaced with $||\theta||^2 \le (r+c) M$  for some constant $M$ by Assumption~\ref{ass:reg.cond}.
	
	To study $\mathcal{D}_4$, we have
	\begin{align}
		&\hspace{1.5em} 
		\mathbb{E}_{\bs{\theta}, B}
		\left[	\mathcal{D}_4 ^2 (\bs{\lambda})\right]=\mathbb{E}_{\bs{\theta}, B}
		\left[
		\left(
		[S_1(\hat{\bs{\theta}}^{\text{ls}}-\bs{\theta})]^\prime W [\hat{\bs{\theta}}^{\text{ls}}-\bs{\theta}]
		-\tr[S_1 L^- W]
		\right)\right]^2  \nonumber \\
		&=\mathbb{V}_{\bs{\theta}, B}
		\left[
		[S_1(\hat{\bs{\theta}}^{\text{ls}}-\bs{\theta})]^\prime W [\hat{\bs{\theta}}^{\text{ls}}-\bs{\theta}] \right] \nonumber 
		=_{(1)}
		\mathbb{V}_{\bs{\theta}, B}
		\left[
		[\mathcal{R}L^- B'\bs{U}]^\prime(S_1^\prime W)
		[\mathcal{R}L^- B'\bs{U}] \right] \nonumber\\
		&\leq_{(2)}
		M \cdot
		\tr\left[
		\left( BL^- S_1'WL^-B' \right)'\left( BL^- S_1'WL^-B' \right)
		\right] 
		=
		M \cdot
		\tr\left[
		WS_1L^-S_1'WL^- \right] \nonumber \\
		&=_{(3)}
		M \cdot
		\tr\left[
		W \mathcal{R} B^\dagger G B L^- B' G B^{\dagger'} \mathcal{R}' W L^-\right] 
		\leq
		M
		\cdot 
		|| G B L^- B' G ||
		\tr\left[
		B^{\dagger'}\mathcal{R}' WL^-W \mathcal{R} B^{\dagger} \right] \nonumber \\
		&\leq_{(4)}
		M
		\cdot 
		\tr\left[
		(W^{1/2} L^- W^{1/2})^2 \right] 
		\leq
		M
		\cdot 
		(r+c) \cdot \bar{\rho}^2(W^{1/2} L^- W^{1/2}) = O\left(
		\delta_{r,c}^{-2}
		\right),
	\end{align}
	where $=_{(1)}$ follows from Lemma 	\ref{lm:pw.equalities}(a), $\leq_{(2)}$ follows from Lemma \ref{lm:pw.equalities}(b),  $=_{(3)}$ follows from Lemma  \ref{lm:pw.equalities}(c), and $\leq_{(4)}$ follows from Lemma  \ref{lm:pw.equalities}(c), $BL^- B'$ being a projection matrix, and $\mathcal{R}B^{\dagger} B^{\dagger'}\mathcal{R}' =\mathcal{R}L^{\dagger}\mathcal{R}'=L^{-}$.
\end{proof}

\subsection{Proof of Uniform Convergence and Asymptotic Optimality}

\begin{proof}[Proof of Theorem \ref{thm:main}]
	We first show that asymptotic optimality follows immediately from the uniform convergence in Lemma~\ref{lm:unif.convg}. To see this, note that
	\begin{equation}
		\begin{aligned}
			& \hspace{1.5em} \mathbb{E}_{\bs{\theta},B}
			\left[
			l_w(\hat{\bs{\theta}}(\bs{\lambda}^{\text{ure}}),\bs{\theta}) 
			\geq
			l_w(\hat{\bs{\theta}}(\bs{\lambda}^{\text{ol}}),\bs{\theta}) 
			+
			\epsilon
			\right] \\
			& = 
			\mathbb{E}_{\bs{\theta},B}
			\left[
			l_w(\hat{\bs{\theta}}(\bs{\lambda}^{\text{ure}}),\bs{\theta}) 
			- \text{URE}(\bs{\lambda}^{\text{ure}})
			\geq
			l_w(\hat{\bs{\theta}}(\bs{\lambda}^{\text{ol}}),\bs{\theta}) 
			- \text{URE}(\bs{\lambda}^{\text{ure}})
			+
			\epsilon
			\right] \\
			& \leq _{(1)}
			\mathbb{E}_{\bs{\theta},B}
			\left[
			l_w(\hat{\bs{\theta}}(\bs{\lambda}^{\text{ure}}),\bs{\theta}) 
			- \text{URE}(\bs{\lambda}^{\text{ure}})
			\geq
			l_w(\hat{\bs{\theta}}(\bs{\lambda}^{\text{ol}}),\bs{\theta}) 
			- \text{URE}(\bs{\lambda}^{\text{ol}})
			+
			\epsilon
			\right] 
			\rightarrow _{(2)} 0.
		\end{aligned}
	\end{equation}
	where $\leq_{(1)}$ follows by $\text{URE}(\bs{\lambda}^{\text{ure}})\leq \text{URE}(\bs{\lambda}^{\text{ol}})$, because $\bs{\lambda}^{\text{ure}}$ minimize the URE by definition, and $\rightarrow_{(2)}$ follow from Lemma \ref{lm:unif.convg}.
\end{proof}

To prove uniform convergence, we first present a reparameterization  $(\tilde{\lambda}_a,\tilde{\lambda}_b)$ of $(\lambda_a,\lambda_b)$ such that 	$\tilde{\bs{\lambda}}:=(\mu,\tilde{\lambda}_a,\tilde{\lambda}_b,\phi)$ has a bounded parameter space. 
For constants $k,n>0$, define
\begin{equation}
	\begin{aligned}[b]
		\lambda_a^{-1} = [\tilde{\lambda}_a^{-n} -1]^k, \quad
		\lambda_b^{-1} = [\tilde{\lambda}_b^{-n} -1]^k.
	\end{aligned}
	\label{eq:tilde.lambda}
\end{equation}
The parameters $\tilde{\lambda}_a\in (0,1)$ and $\tilde{\lambda}_b \in (0,1)$ are strictly increasing bijective functions of $\lambda_a$ and $\lambda_b$, respectively. 
The constant $k,n>0$ are chosen such that  $[2(n+1)]^{-1}=\epsilon$, for some $\epsilon>0$ that satisfies Assumption~\ref{ass:graph.connect.2}, and $k=2(n+1)/n$. This implies that  $\epsilon=(kn)^{-1}$.

\begin{lemma}[Stochastic Equicontinuity]
	\label{lm:stoequi_bypiece}
	
	Suppose Assumptions~\ref{ass:reg.cond}, \ref{ass.degree}, \ref{ass:graph.connect.1} and \ref{ass:graph.connect.2} hold. Then, for $k=2,3,4$ and $x \in \{\mu,\tilde{\lambda}_a,\tilde{\lambda}_b,\phi\}$,	
	$
		\sup_{\tilde{\bs{\lambda}}\in\tilde{\mathcal{J}}}
		\left|
		\frac{\partial}{\partial x}\mathcal{D}_k(\bs{\tilde{\lambda}})
		\right|
		=O_p(1) .\quad 
	$
\end{lemma}

\begin{proof}[Proof of Lemma~\ref{lm:unif.convg}]
	Because $\tilde{\bs{\lambda}}$ is a bijective reparameterize of $\bs{\lambda}$, following Lemma~\ref{lm:decomp}, it is sufficient to show the uniform convergence of $\textstyle\sup_{\tilde{\bs{\lambda}}\in \tilde{\mathcal{J}}}
	\left|
	\mathcal{D}_k(\bs{\tilde{\lambda}})
	\right| 
	=
	o_p(1)$ for $k=2,3,4$, and  $\mathcal{D}_1(\bs{\tilde{\lambda}})=o_p(1).$ The latter follows from Lemma~\ref{lm:UC1} because it does not depend on $\bs{\tilde{\lambda}}$. 
	
	The uniform convergence of $\mathcal{D}_k$ for $k=2,3,4$ follows from (i) their pointwise convergence implied by Lemma~\ref{lm:UC1}, (ii) the parameter space is bounded, and (iii) stochastic equicontinuity in Lemma~\ref{lm:stoequi_bypiece}, following Theorem 1 of \cite{Andrews1992}.
\end{proof}

Below we focus on the proof of the stochastic continuity condition in Lemma~\ref{lm:stoequi_bypiece}. The following notation is needed to set up the proof. Define any power of $L$ (or $  L^\dagger$) with the eigendecomposition of $L$. Define 
\begin{eqnarray}\label{eq:s1V.check}
	V_\lambda:=L^{1/2}  [L+\Lambda^*]^{-1}  L^{1/2+\epsilon},
\end{eqnarray}
where 
$\epsilon>0$ and $\epsilon^{-1}$ is a given large positive even number. 
Eventually, a critical ingredient in the proof of stochastic equicontinuity is to show that the following holds for $V_\lambda$.

\begin{lemma}
	\label{lm:bd.deriv}
	
	Suppose Assumptions~\ref{ass:reg.cond}, \ref{ass.degree}, \ref{ass:graph.connect.1} and \ref{ass:graph.connect.2} hold. Then, for $x \in \{\tilde{\lambda}_a,\tilde{\lambda}_b,\phi\}$,
	$\textstyle\sup_{\tilde{\bs{\lambda}}\in\tilde{\mathcal{J}}} 
	\Vert \frac{\partial V_\lambda}{\partial x}\Vert = O(1)$.

\end{lemma}

\begin{lemma} \label{lm: Vdefine}
	$S_1 = \mathcal{R}  [L^\dagger]^{1/2} \cdot   V_\lambda \cdot L^{1/2-\epsilon}$.
\end{lemma}

\begin{proof}[Proof of Lemma~\ref{lm:stoequi_bypiece}]
	Define a positive semi-definite matrix $W_{\mathcal{R}}:=\mathcal{R}'W\mathcal{R}$.
	Note that $\tr [(W_\mathcal{R}^{1/2} L^\dagger
	W_\mathcal{R}^{1/2})^2]=\tr [(W^{1/2} L^{-}
	W^{1/2})^2]$, which is $O(1)$ by Assumption~\ref{ass:graph.connect.2}. 
	
	We first consider $\mathcal{D}_2(\bs{\tilde{\lambda}})$. By setting $\epsilon=1/2$ in the definition of $V_\lambda$, we have
	\begin{align}
		\mathcal{D}_2(\bs{\tilde{\lambda}}) =&	(Sv)' W (\hat{\bs{\theta}}^{\text{ls}}-\bs{\theta})
		= 
		v'W (\hat{\bs{\theta}}^{\text{ls}}-\bs{\theta})
		-v'S_1' W (\hat{\bs{\theta}}^{\text{ls}}-\bs{\theta}) \nonumber\\
		=&_{(1)}
		v'W (\hat{\bs{\theta}}^{\text{ls}}-\bs{\theta})
		-v' {V}_\lambda'   [L^\dagger]^{1/2} W_{\mathcal{R}} (\hat{\bs{\theta}}^{\text{ls}}-\bs{\theta}) ,\label{eq:Sbeta_0}
	\end{align}
	where $=_{(1)}$ follows from Lemma~\ref{lm: Vdefine} and 
	$\hat{\bs{\theta}}^{\text{ls}}-\bs{\theta} = \mathcal{R}( \hat{\bs{\theta}}^{\text{ls}}-\bs{\theta} )$ since both $ \hat{\bs{\theta}}^{\text{ls}}$ and $\bs{\theta}$ satisfy the normalization implied by $\mathcal{R}$. Note that in \eqref{eq:Sbeta_0}, $V_{\lambda}$ is the only part that depends on $x \in \{\tilde{\lambda}_a,\tilde{\lambda}_b,\phi\}$. We have
	\begin{align}		
		& 	\hspace{1.5em} \mathbb{E}_{\bs{\theta},B}\left[
		\sup_{\tilde{\bs{\lambda}}\in\tilde{\mathcal{J}}}
		\left|
		\frac{\partial}{\partial x}\mathcal{D}_2(\bs{\tilde{\lambda}})
		\right|
		\right] \nonumber 
		= 
		\mathbb{E}_{\bs{\theta},B}\left[
		\sup_{\tilde{\bs{\lambda}}\in\tilde{\mathcal{J}}}
		\left|
		v'\frac{\partial {V}_\lambda'}{\partial x}    [L^\dagger]^{1/2} W_{\mathcal{R}} (\hat{\bs{\theta}}^{\text{ls}}-\bs{\theta})
		\right|
		\right]  \nonumber \\
		& \le 
		\mathbb{E}_{\bs{\theta},B}\left[
		\sup_{\tilde{\bs{\lambda}}\in\tilde{\mathcal{J}}}
		\Vert v\Vert
		\Vert
		\frac{\partial {V}_{\lambda}}{\partial x}
		\Vert
		\Vert
		[L^\dagger]^{1/2}W_{\mathcal{R}}(\hat{\bs{\theta}}^{\text{ls}}-\bs{\theta})
		\Vert\right] \nonumber\\
		& \le 
		\left(\frac{1}{r^{1/2}}  \sup_{\tilde{\lambda}\in\tilde{\mathcal{J}}}
		\Vert v \Vert\right)
		\left(\sup_{\tilde{\bs{\lambda}}\in\tilde{\mathcal{J}}}
		\Vert\frac{\partial {V}_{\lambda}}{\partial x}\Vert\right)
		\left( r \cdot \mathbb{E}_{\bs{\theta},B}\left[\Vert 
		[L^\dagger]^{1/2}W_{\mathcal{R}}(\hat{\bs{\theta}}^{\text{ls}}-\bs{\theta})\Vert^{2}\right]\right)^{1/2}, 	\label{eq:Sbeta_1}
	\end{align}
	where we apply the Cauchy-Schwarz inequality and use the fact that
	$\hat{\bs{\theta}}^{\text{ls}}$ is the only term that is stochastic and it does
	not depend on $\tilde{\lambda}$.
	The first term in \eqref{eq:Sbeta_1} is $O(1)$ because $\mu \le \overline{\mu}$. The second term in \eqref{eq:Sbeta_1} is $O(1)$ by Lemma~\ref{lm:bd.deriv}. To see the third term in \eqref{eq:Sbeta_1} is also $O(1)$, note that 
	\begin{align}
		r\cdot 
		\mathbb{E}_{\bs{\theta},B}\left[\Vert  [L^\dagger]^{1/2}W_{\mathcal{R}}(\hat{\bs{\theta}}^{\text{ls}}-\bs{\theta})\Vert^{2}\right]\nonumber 
		= & r\cdot
		\mathbb{E}_{\bs{\theta},B}\left[\tr\left((\hat{\bs{\theta}}^{\text{ls}}-\bs{\theta})'W_{\mathcal{R}}  L^\dagger W_{\mathcal{R}}(\hat{\bs{\theta}}^{\text{ls}}-\bs{\theta})\right)\right]\nonumber \\
		=&_{(1)}  r\cdot
		\tr\left(W_{\mathcal{R}} L^\dagger W_{\mathcal{R}} L^\dagger \right)
		=r\cdot \tr( [W_{\mathcal{R}}^{1/2}L^\dagger W_{\mathcal{R}}^{1/2}]^2 ), \label{eq:Sbeta_13}
	\end{align}
	where $=_{(1)}$ uses $\mathbb{V}(\hat{\bs{\theta}}^{\text{ls}})=  L^{-}=\mathcal{R}L^\dagger\mathcal{R}'$ and \eqref{eq:Sbeta_13} is $O(1)$ by Assumptions~\ref{ass.degree} and \ref{ass:graph.connect.2}.
	Among all the terms, 
	$(Sv)' W(\hat{\bs{\theta}}^{\text{ls}}-\bs{\theta})$ is the only one that depends
	on $\mu$. 
	We have
	\begin{equation}
		\sup_{\tilde{\bs{\lambda}}\in\tilde{\mathcal{J}}}
		\left|\frac{\partial}{\partial\mu}(Sv)'W(\hat{\bs{\theta}}^{\text{ls}}-\bs{\theta})
		\right|
		=
		\sup_{\tilde{\bs{\lambda}}\in\tilde{\mathcal{J}}}
		\left|(Sd_{r})' W(\hat{\bs{\theta}}^{\text{ls}}-\bs{\theta})\right|,
	\end{equation}
	where $d_{r}=(1_{r}',0_{c}')'$ and the right hand side is $o_p(1)$ by the same arguments used to show $\sup_{\tilde{\bs{\lambda}}\in \tilde{\mathcal{J}}}
	\left|
	\mathcal{D}_2(\bs{\tilde{\lambda}})
	\right| 
	=
	o_p(1)$ with $\mu$ in $\mathcal{D}_2(\bs{\tilde{\lambda}})$ replaced by $1$.
	The result for $\mathcal{D}_3(\bs{\tilde{\lambda}})$ follows the same arguments as for $\mathcal{D}_2(\bs{\tilde{\lambda}})$ under $\tfrac{1}{r+c}\bs{\theta}'\bs{\theta} = O(1)$ by Assumption~\ref{ass:reg.cond}.

	Next, we consider $\mathcal{D}_4(\bs{\tilde{\lambda}})$. We first show the stochastic equicontinuity of 
	$(S_1(\hat{\bs{\theta}}^{\text{ls}}-\bs{\theta}))' W (\hat{\bs{\theta}}^{\text{ls}}-\bs{\theta})$ in $x\in\{\tilde{\lambda}_a,\tilde{\lambda}_b,\phi\}$.
	Using the definition of $V_{\lambda}$, we have
	\begin{align}
		& \hspace{1.5em} \mathbb{E}_{\bs{\theta},B}\left[
		\sup_{\tilde{\bs{\lambda}}\in\tilde{\mathcal{J}}}
		\left|
		\frac{\partial}{\partial x}
		\left(S_1(\hat{\bs{\theta}}^{\text{ls}}-\bs{\theta})\right)' W (\hat{\bs{\theta}}^{\text{ls}}-\bs{\theta})
		\right|\right]\nonumber \\
		& = 
		\mathbb{E} _{\bs{\theta},B}\left[
		\sup_{\tilde{\bs{\lambda}}\in\tilde{\mathcal{J}}}
		\left|
		\left( \frac{\partial V_\lambda}{\partial x} L^{1/2-\epsilon} (\hat{\bs{\theta}}^{\text{ls}}-\bs{\theta})\right)' [L^\dagger]^{1/2}W_{\mathcal{R}}(\hat{\bs{\theta}}^{\text{ls}}-\bs{\theta})
		\right|\right]\nonumber \\
		& \le  \sup_{\tilde{\bs{\lambda}}\in\tilde{\mathcal{J}}} \Vert 
		\frac{\partial {V}_\lambda}{\partial x} \Vert
		\left(
		\mathbb{E} _{\bs{\theta},B}\left[ \Vert 
		L^{1/2-\epsilon}(\hat{\bs{\theta}}^{\text{ls}}-\bs{\theta})
		\Vert^2 \right]
		\right)^{1/2}
		\left(
		\mathbb{E} _{\bs{\theta},B}\left[\Vert [L^\dagger]^{1/2}W_{\mathcal{R}} (\hat{\bs{\theta}}^{\text{ls}}-\bs{\theta}) \Vert^2 \right]
		\right)^{1/2}
		\nonumber \\
		& = 
		\sup_{\tilde{\bs{\lambda}}\in\tilde{\mathcal{J}}} \Vert 
		\frac{\partial {V}_\lambda}{\partial x} \Vert \cdot  
		\left[
		\tr([L^\dagger]^{2\epsilon}) 
		\right]^{1/2}
		\left[
		\tr( W_{\mathcal{R}}L^\dagger W_{\mathcal{R}} L^\dagger )
		\right]^{1/2}
		\nonumber \\
		&	=  
		\sup_{\tilde{\bs{\lambda}}\in\tilde{\mathcal{J}}}
		\Vert 
		\frac{\partial {V}_\lambda}{\partial x} \Vert  \cdot  
		\left[
		\frac{1}{r+c} \tr( [L^\dagger]^{2\epsilon})
		\right]^{1/2} 
		\left[
		(r+c) \cdot\tr( [W_{\mathcal{R}}^{1/2}L^\dagger W_{\mathcal{R}}^{1/2}]^2 )
		\right]^{1/2},
	\end{align}
	where the three items in the last equality are all $O(1)$ by Lemma~\ref{lm:bd.deriv} and Assumption~\ref{ass:graph.connect.2}.  
	
	Finally, we consider the equicontinuity of $\tr(S_1  L^{-}W)$
	in $x\in\{\tilde{\lambda}_a,\tilde{\lambda}_b,\phi\}$.
	Note that this item does not depend on $\mu$.
	Using $B\mathcal{R}=B$, we can write $\tr(S_{1}  L^{-}W)=\tr(S_{1} L^\dagger \mathcal{R}'W),$ and
	\begin{align}
		& \hspace{2em} \sup_{\tilde{\bs{\lambda}}\in\tilde{\mathcal{J}}}
		\left|
		\frac{\partial}{\partial x}\tr(S_{1}  L^\dagger \mathcal{R}'W)
		\right| 
		 =  _{(1)} 
		\sup_{\tilde{\bs{\lambda}}\in\tilde{\mathcal{J}}}
		\left|
		\frac{\partial}{\partial x}
		\tr\left(
		\mathcal{R} [L^\dagger]^{1/2} \cdot V_\lambda \cdot [L^\dagger]^{1/2+\epsilon}   \mathcal{R}' W
		\right) 
		\right|
		\nonumber \\
		& =  _{(2)} 
		\sup_{\tilde{\bs{\lambda}}\in\tilde{\mathcal{J}}}
		\left|
		\tr\left(
		W_{\mathcal{R}}^{1/2}   
		[L^\dagger]^{1/2} \frac{\partial V_\lambda }{\partial x} \cdot [L^\dagger]^{1/2+\epsilon}  
		W_{\mathcal{R}}^{1/2}
		\right) 
		\right|
		\nonumber \\
		& \le_{(3)} 
		\sup_{\tilde{\bs{\lambda}}\in\tilde{\mathcal{J}}}
		\left(
		\tr\left[
		\frac{\partial V_\lambda '}{\partial x}
		[L^\dagger]^{2\epsilon}
		\frac{\partial V_\lambda }{\partial x}
		\right]
		\right)^{1/2} 
		\cdot 
		\left(
		\tr\left[
		(W_\mathcal{R}^{1/2} L^\dagger
		W_\mathcal{R}^{1/2})^2
		\right]
		\right)^{1/2} \nonumber \\
		& \leq 
		\sup_{\tilde{\bs{\lambda}}\in\tilde{\mathcal{J}}} ||\frac{\partial V_\lambda }{\partial x}||\cdot 
		\left( 
		\frac{1}{r+c}
		\tr
		[L^\dagger]^{2\epsilon}
		\right)^{1/2} 
		\cdot 
		\left(
		(r+c)
		\tr [(W_\mathcal{R}^{1/2} L^\dagger
		W_\mathcal{R}^{1/2})^2]
		\right)^{1/2}, \label{eq:S_11} 
	\end{align}
	where the first term in \eqref{eq:S_11} is $O(1)$ by Lemma~\ref{lm:bd.deriv} and the last two terms are $O(1)$ by  Assumption~\ref{ass:graph.connect.2}. 
\end{proof}

\begin{lemma}[Matrix Equalities]
	\label{lm:equalities}
	\hspace{1em} 
	\newline
	\emph{(a).}		$[\mathcal{R}L^\dagger \mathcal{R}']_2
	= 
	M_c[L_{2,\perp}]^\dagger M_c
	$,\
	\emph{(b).} $B \mathcal{R}
	=
	B$, \
	\label{eq:E3} 
	\emph{(c).} $L^- L  L^-
	=
	L^-
	$, \
	\label{eq:E4} 
	\emph{(d).} $L L^- L
	=
	L$,
	\label{eq:E5} \\
	\emph{(e).} $
	\mathcal{R}
	=
	\mathcal{R} L^-L $.
	\label{eq:E6}
\end{lemma}

\begin{lemma}[Inequalities]
	\label{lm:inequalities}
	For some finite constant $M$, \\
	\emph{(a).}  $\lVert \mathcal{R} \rVert
	\leq 	
	M$,
	\label{eq:I1} \
	\emph{(b).}  $\lVert B \rVert
	\leq 	
	M$,
	\label{eq:I2} \
	\emph{(c).} $\lVert L  L^-\rVert 
	\leq 
	M.$
	\label{eq:I3}
\end{lemma}

\end{appendix}

\setstretch{1.3}


\clearpage

\renewcommand{\thepage}{A.\arabic{page}}
\setcounter{page}{1}

\begin{appendix}
	
	\markright{Supplementary Online Appendix -- This Version: \today }
	\setcounter{equation}{0}
	
	\renewcommand*\thetable{A-\arabic{table}}
	\setcounter{table}{0}
	\renewcommand*\thefigure{A-\arabic{figure}}
	\setcounter{figure}{0}
	
	
	\setcounter{section}{1}

	\renewcommand{\thesection}{\Alph{section}}
	\renewcommand{\theequation}{\Alph{section}.\arabic{equation}}
	\renewcommand{\thetheorem}{\Alph{section}\arabic{theorem}}
	\renewcommand{\thelemma}{\Alph{section}\arabic{lemma}}
	\renewcommand{\theassumption}{\Alph{section}\arabic{assumption}}
	\renewcommand{\theproposition}{\Alph{section}\arabic{proposition}}
	\renewcommand{\thecorollary}{\Alph{section}\arabic{corollary}}
	
	\thispagestyle{empty}
	
	\begin{center}
		
		{\large {\bf Supplementary Online Appendix: Optimal Estimation of Two-Way Effects under Limited Mobility}}
		
		{\bf Xu Cheng, Sheng Chao Ho, and Frank Schorfheide}
	\end{center}

This Appendix starts with section B.  Section A is the Appendix in the paper.
\noindent It consists of the following sections:
	
	\begin{itemize}
	\item[B.] Proofs of additional theoretical results in the paper
	\item[C.] Proofs of results used in Appendix A
	\item[D.] Computation of EB Estimators for Large Datasets
	\item[E.] Additional Monte Carlo Simulations
	\item[F.] Further Details on the Empirical Analysis
	\item[G.] Definitions and Notations
\end{itemize}
\clearpage

\section{Proofs of Additional Theoretical Results in Paper}
\label{sec:proofs}



\label{subsec:proof.ls}

\begin{lemma} \label{lm: LS defn}
	The LS estimator is uniquely defined as $\hat{\bs{\theta}}^{\text{ls}} =  L^- B^\prime \bs{Y}$ under the normalization.
\end{lemma}
\begin{proof}[Proof of Lemma \ref{lm: LS defn}]
	We can show that $ L^-$ is a generalized inverse of $L$, i.e., $LL^{-}L=L$ using the definition of $L^{-}=\mathcal{R}L^{\dagger}\mathcal{R}'$, $B\mathcal{R}=B$, and $LL^\dagger L = L$ for the Moore-Penrose inverse.  Thus, $\hat{\bs{\theta}}^{\text{ls}}=  L^-B^\prime \bs{Y}$ is a LS estimate. It satisfies the normalization $\bs{1}^\prime \bs{\beta}=0$ too since $\mathcal{R}$ appears on the left-hand side in $\hat{\bs{\theta}}^{\text{ls}}$. 
	
	We now turn to the uniqueness. Any least-squares estimate  $\check{\bs{\theta}}^{\text{ls}}$ can be characterized as 
	$
	\check{\bs{\theta}}^{\text{ls}} = \hat{\bs{\theta}}^{\text{ls}} +  [I-  L^-L]\bs{w} 
	$
	for some $\bs{w}\in\mathbb{R}^{r+c}$. Thus if $\check{\bs{\theta}}^{\text{ls}}$ is a LS estimate also satisfying the normalization $\bs{1}^\prime \bs{\beta}=0$, then	
	$	\check{\bs{\theta}}^{\text{ls}}
	=_{(1)}
	\mathcal{R}\check{\bs{\theta}}^{\text{ls}} 
	=
	\mathcal{R}\hat{\bs{\theta}}^{\text{ls}} 
	+
	[\mathcal{R} - \mathcal{R} L^-L]\bs{w} 
	=_{(2)}
	\hat{\bs{\theta}}^{\text{ls}} 
	+
	[\mathcal{R} - \mathcal{R}]\bs{w} 
	=
	\hat{\bs{\theta}}^{\text{ls}},$
where $=_{(1)}$ follows from the fact that $\check{\bs{\theta}}^{\text{ls}}$ satisfies the normalization $\bs{1}^\prime \bs{\beta}=0$ and $=_{(2)}$ follows from $	\mathcal{R}
=
\mathcal{R} L^-L $, see \eqref{eq:E6}.
\end{proof}

\begin{lemma} \label{lm:pstm}
The posterior mean is $	\hat{\bs{\theta}}(\bs{\lambda}) :=
S_1\hat{\bs{\theta}}^{\text{ls}} + 
S \bs{v}$.
\end{lemma}
\begin{proof}[Proof of Lemma~\ref{lm:pstm}]
The model 
$\bs{Y}\mid B,\bs{\theta} \sim N(B\bs{\theta},I)$ 
together with the prior 
$\bs{\theta} \mid B \sim N(\bs{v},\Lambda^{*-1})$ 
can be written as 
\begin{equation}
	\begin{bmatrix}
		\bs{Y} \\ \bs{\theta}
	\end{bmatrix}
	\bigg| B\sim 
	\mathcal{N}\left(
	\begin{bmatrix}
		B\bs{v} \\ \bs{v}
	\end{bmatrix}, 
	\begin{bmatrix}
		B\Lambda^{*-1}B^\prime + I & B \Lambda^{*-1} \\
		\Lambda^{*-1}B^\prime & \Lambda^{*-1}
	\end{bmatrix}
	\right).
\end{equation}
Using the formula for the conditional distribution of joint normal distributions together with the linearity of expectation give us the expression for  $\hat{\bs{\theta}}:=\mathbb{E}\left[\mathcal{R}\bs{\theta}\mid B,\bs{Y}\right]$ as
\begin{equation}
	\begin{aligned}
		&\hspace{1.5em}\mathcal{R}\mathbb{E}[\bs{\theta}\mid B,\bs{Y}]\\
		&=
		\mathcal{R}\left\{\bs{v}
		+
		\Lambda^{*-1}B^\prime[B\Lambda^* B^\prime + I]^{-1}[\bs{Y}-B\bs{v}]\right\}\\
		&=_{(1)}
		\mathcal{R}
		\left\{\bs{v}
		+
		[L + \Lambda^*]^{-1} B^\prime[\bs{Y}-B\bs{v}]\right\}
		=
		\mathcal{R}
		\left\{\bs{v} - [L + \Lambda^*]^{-1} L\bs{v} 
		+ 
		[L + \Lambda^*]^{-1} B^\prime\bs{Y}\right\} \\
		&=_{(2)}
		S\bs{v} 
		+ 
		\mathcal{R}
		\left\{[L + \Lambda^*]^{-1} L  L^- B^\prime \bs{Y}\right\} 
		=_{(3)}
		S\bs{v} 
		+ 
		S_1 \hat{\bs{\theta}}^{\text{ls}},
	\end{aligned} \label{postmean}
\end{equation}
where $=_{(1)}$ follows from
$D^{-1}C(A-BDC)^{-1} = (D-CA^{-1}B)^{-1}CA^{-1}$
for conformal matrices,  $=_{(2)}$ follows by defining $S_1:=\mathcal{R}[L + \Lambda^*]^{-1}L$, $S_1+S = \mathcal{R}$, and $LL^{-}B'=B'$  using the property that $LL^{-}$ is a projection matrix onto the column space of $L$ and $B'$ belongs to the column space of $L$, 
and $=_{(3)}$ follows from definition of $\hat{\bs{\theta}}^{\text{ls}} =  L^- B^\prime \bs{Y}$.
\end{proof}

\begin{proof}[Proof of Lemma~\ref{lm:proposedII}]
The first part of the Lemma is proved using the form of conditional distribution for a jointly normal vector together with an application of Lemma~\ref{lm:norm.cond}. The second part of the Lemma follows from Lemma~\ref{lm:norm.precision}.
\end{proof}

%

Consider the LS estimate of the error variance in 
\eqref{eq:sigma_est}. The following Lemma shows that it is unbiased and consistent.

\begin{lemma} \label{lm:LSvar}
	\emph{(a).}	$\mathbb{E}_{\bs{\theta},B}[	 \hat{\sigma}^2] =  \sigma^2$.
	
	\noindent \emph{(b).} Suppose $rT-(r+c-1) \rightarrow \infty$
	as $r,c \rightarrow \infty$,  then $\hat{\sigma}^2 \rightarrow_p \sigma^2.$
\end{lemma}

\begin{proof}[Proof of Lemma \ref{lm:LSvar}]
	Define a constant
	\begin{align}
		\kappa = \frac{rT}{rT-(r+c-1)}.
	\end{align}
	For part (a), plugging $\bs{Y}=B\bs{\theta}+\bs{U}$, we have
	$
	\hat{\sigma}^2
	=\frac{\kappa}{rT}\bs{U}'\bs{U}-\frac{\kappa}{rT}\bs{U}'B(B'B)^{-}B'\bs{U},
	$
	where 
	\begin{align}
		\frac{1}{rT}\mathbb{E}_{\bs{\theta},B}[\bs{U}'\bs{U}] &= \sigma^2, \nonumber\\
		\frac{1}{rT}\mathbb{E}_{\bs{\theta},B}\left[\bs{U}'B(B'B)^{-}B'\bs{U}\right]
		&=_{(1)}\frac{ \sigma^2}{rT}\tr[B(B'B)^{-}B'] 
		=_{(2)}\frac{\sigma^2\cdot (r+c-1)}{rT},
	\end{align}
	where $=_{(1)}$ follows from $\mathbb{E}_{\bs{\theta},B}[\bs{U}\bs{U}']=\sigma^2 I_{rT}$, and $=_{(2)}$ follows from the assumption that $r+c\le rT$ and $L=B'B$ has one $0$ eigenvalue. Therefore, we have
	$\mathbb{E}_{\bs{\theta},B}[	 \hat{\sigma}^2] =  \sigma^2$. 
	
	For part (b), by applying  Lemma \ref{lm:pw.equalities}(b), we have
	\begin{align}
		\mathbb{V}_{\bs{\theta},B} \left[\frac{\kappa}{rT}\bs{U}'\bs{U}\right]&\le\frac{M}{(rT-(r+c-1)^2}\tr[I_{rT}]=o(1), \nonumber\\
		\mathbb{V}_{\bs{\theta},B}\left[\frac{\kappa}{rT}\bs{U}'B(B'B)^{-}B'\bs{U}\right]&\le\frac{M}{(rT-(r+c-1))^2}\tr[B(B'B)^{-}B']=_{(1)}o(1),
	\end{align}
	for some finite $M$,
	where $=_{(1)}$ holds because the eigenvalues of the projection matrix $B(B'B)^{-}B'$ is no larger than $1$. 
\end{proof}

The following lemma shows that the impact of the estimated coefficient is negligible uniformly over the hyperparameter. It is the key result to deliver asymptotic optimality with the estimated coefficient in Corollary~\ref{cor:reg}.
\begin{lemma}[Estimated Regressor Coefficient]
\label{lm:UC.reg}
Suppose Assumptions~\ref{ass:reg.cond}, \ref{ass.degree}, \ref{ass:graph.connect.1}, \ref{ass:graph.connect.2}, and  \ref{ass:reg.condreg} hold. Then,
\begin{align*}
	\sup_{\bs{\lambda}\in\mathcal{J}} 
	\left|
	l_w(\hat{\bs{\theta}}(\bs{\lambda}),\bs{\theta}) - 
	l_w(\tilde{\bs{\theta}}(\bs{\lambda}),\bs{\theta})
	\right|
	&\rightarrow_p 0,\label{eq:UC.reg.1}\\
	\sup_{\bs{\lambda}\in\mathcal{J}} 
	\left|
	\text{URE}(\bs{\lambda}) - 
	\widetilde{\text{URE}}(\bs{\lambda})
	\right|
	&\rightarrow_p 0.
\end{align*}
\end{lemma}
\begin{proof}[Proof of Lemma \ref{lm:UC.reg}]
We only detail the proof of the result for the loss, as the corresponding result for the URE 
can be established using similar arguments. 
\begin{equation}
	\begin{aligned}[b]
		&l_w(\hat{\bs{\theta}}(\bs{\lambda}),\bs{\theta}) - 
		l_w(\tilde{\bs{\theta}}(\bs{\lambda}),\bs{\theta}) \\
		&=
		\left[S_1\hat{\bs{\theta}}^{\text{ls}}+S\bs{v}- \bs{\theta}\right]' W \left[S_1\hat{\bs{\theta}}^{\text{ls}}+S\bs{v} - \bs{\theta}\right]
		-
		\left[S_1\tilde{\bs{\theta}}^{\text{ls}}+S\bs{v} - \bs{\theta}\right]' W \left[S_1\tilde{\bs{\theta}}^{\text{ls}}+S\bs{v} - \bs{\theta}\right] \\
		&=
		\left[S_1\left( \hat{\bs{\theta}}^{\text{ls}}-\tilde{\bs{\theta}}^{\text{ls}}\right)\right]^\prime W
		\left[S_1\left( \hat{\bs{\theta}}^{\text{ls}}+\tilde{\bs{\theta}}^{\text{ls}}\right)
		+2S\bs{v} - 2\bs{\theta}\right] \\
		&=_{(1)}
		\left[
		S_1B^-X\left( \tilde{\bs{\gamma}}-\bs{\gamma}\right)
		\right]' W \left[
		S_1B^-X\left( \bs{\gamma}-\tilde{\bs{\gamma}}\right)
		+2\left(S_1\hat{\bs{\theta}}^{\text{ls}} + S\bs{v} - \bs{\theta}\right)
		\right]
		\\
		&=
		\left[
		S_1B^-X\left( \tilde{\bs{\gamma}}-\bs{\gamma}\right)
		\right]' W_{\mathcal{R}} \left[
		S_1B^-X\left( \bs{\gamma}-\tilde{\bs{\gamma}}\right)
		+2S_1\left(\hat{\bs{\theta}}^{\text{ls}} - \bs{\theta}\right)
		+2S\left(\bs{v}-\bs{\theta}\right)
		\right]
	\end{aligned}
\end{equation}
where $B^- = RB^\dagger$ and $=_{(1)}$ is due to
\begin{equation}
	\hat{\bs{\theta}}^{\text{ls}}-\tilde{\bs{\theta}}^{\text{ls}} = B^-[\bs{Y}^*-X\bs{\gamma}] - B^-[\bs{Y}^*-X\tilde{\bs{\gamma}}] = B^-X[\tilde{\bs{\gamma}}-\bs{\gamma}].
\end{equation}
Furthermore, by defining $\tilde{V}_\lambda:= L^{1/2+\epsilon}[L+\Lambda^*]^{-1} L^{1/2} $ in a similar logic as we did for $V_\lambda$ in \eqref{eq:s1V.check}, we have
\begin{equation}
	W_{\mathcal{R}}^{1/2} S_1B^-X(\tilde{\bs{\gamma}}-\bs{\gamma})
	=
	W_{\mathcal{R}}^{1/2} [L^\dagger]^{1/2+\epsilon} \cdot 
	\tilde{V}_\lambda \cdot
	L^{1/2}B^\dagger \cdot 
	\mathcal{R} X(\tilde{\bs{\gamma}}-\bs{\gamma}),
	\label{eq:UC.reg.3}
\end{equation}
where the first term is $o(1)$ under Assumption~\ref{ass:graph.connect.1} with $\epsilon=1/2$ for the case $W=W_{a+b}$ and $W=W_{b}$ following \eqref{gencond1} and \eqref{eq:partition}, respectively, and the (derivative of) the second term is handled in a similar fashion\footnote{See in particular Remarks~\ref{rmk:deriv.v} and \ref{rmk:deriv.v.1}; they apply in an almost identical fashion to $\tilde{V}_\lambda$ as to $V_\lambda$.} as $V_\lambda$ in Section~\ref{subsec:bd.deriv}. The third term of \eqref{eq:UC.reg.3} has norm less than one because $L^{1/2}B^\dagger B^{\dagger'} L^{1/2}=L^{1/2}L^\dagger L^{1/2} $ is a projection matrix. Finally the fourth term of \eqref{eq:UC.reg.3} is $O_p(1)$ by $\lVert \mathcal{R} \rVert
\leq 	
M$ and Assumption~\ref{ass:reg.condreg}.
Using arguments analogous to those in the proof of Lemma~\ref{lm:stoequi_bypiece}, we have $S_1(\hat{\bs{\theta}}^{\text{ls}}-\bs{\theta})$ and $S(\bs{v}-\bs{\theta})$ are uniformly $O_p(1)$.
\end{proof}

\begin{proof}[Proof of Corollary \ref{cor:reg}]

First note that 
\begin{equation}
	\begin{aligned}[b]
		&l_w(\tilde{\bs{\theta}}(\bs{\lambda}^{\widetilde{\text{ure}}}),\bs{\theta}) - l_w(\hat{\bs{\theta}}(\bs{\lambda}^{\text{ol}}),\bs{\theta})
		\\
		&=
		\underbrace{l_w(\tilde{\bs{\theta}}(\bs{\lambda}^{\widetilde{\text{ure}}}),\bs{\theta}) - l_w(\hat{\bs{\theta}}(\bs{\lambda}^{\widetilde{\text{ure}}}),\bs{\theta}) }_{\text{(I)}}
		+
		\underbrace{l_w(\hat{\bs{\theta}}(\bs{\lambda}^{\widetilde{\text{ure}}}),\bs{\theta}) - l_w(\hat{\bs{\theta}}(\bs{\lambda}^{\text{ol}}),\bs{\theta})}_{\text{(II)}}
	\end{aligned}
\end{equation}
and term (I) is $o_p(1)$ by Lemma \ref{lm:UC.reg}. Furthermore, term (II) can be bound as follows
\begin{equation}
	\begin{aligned}[b]
		0 
		&\leq 
		l_w(\hat{\bs{\theta}}(\bs{\lambda}^{\widetilde{\text{ure}}}),\bs{\theta}) - l_w(\hat{\bs{\theta}}(\bs{\lambda}^{\text{ol}}),\bs{\theta}) \\
		&=
		\underbrace{ l_w(\hat{\bs{\theta}}(\bs{\lambda}^{\widetilde{\text{ure}}}),\bs{\theta}) - \text{URE}(\bs{\lambda}^{\widetilde{\text{ure}}})}_{\text{(III)}}
		+
		\underbrace{\text{URE}(\bs{\lambda}^{\widetilde{\text{ure}}}) - \text{URE}(\bs{\lambda}^{\text{ure}})}_{\text{(IV)}}
		\\
		&+
		\underbrace{\text{URE}(\bs{\lambda}^{\text{ure}}) - l_w(\hat{\bs{\theta}}(\bs{\lambda}^{\text{ure}}),\bs{\theta})}_{\text{(V)}} 
		+
		\underbrace{l_w(\hat{\bs{\theta}}(\bs{\lambda}^{\text{ure}}),\bs{\theta}) - l_w(\hat{\bs{\theta}}(\bs{\lambda}^{\text{ol}}),\bs{\theta})}_{\text{(VI)}}
	\end{aligned}
\end{equation}
Note that terms (III) and (V) are $o_p(1)$ by Lemma \ref{lm:unif.convg}, and term (VI) is $o_p(1)$ by Theorem \ref{thm:main}. 
Furthermore, term (IV) can be bound as follows
\begin{equation}
	\begin{aligned}[b]
		0 
		&\leq 
		\text{URE}(\bs{\lambda}^{\widetilde{\text{ure}}}) - \text{URE}(\bs{\lambda}^{\text{ure}}) \\
		&=
		\text{URE}(\bs{\lambda}^{\widetilde{\text{ure}}}) - \widetilde{\text{URE}}(\bs{\lambda}^{\widetilde{\text{ure}}}) 
		+
		\widetilde{\text{URE}}(\bs{\lambda}^{\widetilde{\text{ure}}}) - \text{URE}(\bs{\lambda}^{\text{ure}}) \\
		&\leq_{(1)}
		\underbrace{\text{URE}(\bs{\lambda}^{\widetilde{\text{ure}}}) - \widetilde{\text{URE}}(\bs{\lambda}^{\widetilde{\text{ure}}})}_{\text{(VII)}} 
		+
		\underbrace{\widetilde{\text{URE}}(\bs{\lambda}^{\text{ure}}) - \text{URE}(\bs{\lambda}^{\text{ure}})}_{\text{(VIII)}}
	\end{aligned}
\end{equation}
where $\leq_{(1)}$ follows because $\widetilde{\text{URE}}(\bs{\lambda}^{\widetilde{\text{ure}}})\leq\widetilde{\text{URE}}(\bs{\lambda}^{\text{ure}})$ by definition of URE minimization. Furthermore, terms (VII) and (VIII) are $o_p(1)$ by Lemma \ref{lm:UC.reg}.

Thus, we have term $0\leq\text{(II)}\leq\text{(III)}+\text{(V)}+\text{(VI)}+\text{(VII)}+\text{(VIII)} = o_p(1)$ and, therefore,
\begin{equation}
	\begin{aligned}[b]
		l_w(\tilde{\bs{\theta}}(\bs{\lambda}^{\widetilde{\text{ure}}}),\bs{\theta}) - l_w(\hat{\bs{\theta}}(\bs{\lambda}^{\text{ol}}),\bs{\theta})
		=
		\text{(I)}
		+
		\text{(II)}
		= o_p(1) + o_p(1)=o_p(1).
	\end{aligned}
\end{equation}
\end{proof}

\begin{proof}[Proof of Lemmas~\ref{lm:1way}]
We first introduce some additional definitions for this proof. Define 
\begin{equation}
	\begin{aligned}[b]
		D_{1,\lambda_a} &:= B_1^\prime B_1 + \lambda_a, \ \
		D_{2,\lambda_b} := B_2^\prime B_2 + \lambda_b, \\
		M_{1,\lambda_a} &:= I_{rT} - B_1 D_{1,\lambda_a}^{-1} B_1^\prime, \ \
		M_{2,\lambda_b} := I_{rT} - B_2 D_{2,\lambda_b}^{-1} B_2^\prime.
	\end{aligned}
\end{equation}
Note that $M_{1,0}$, which we simply write as $M_1$, is the matrix that projects onto the space orthogonal to the column space of $B_1$, whereas $M_{1,\infty}$ is an identity matrix. 

Recalling that $\phi=0$ for this Lemma so that $\Lambda^*=\Lambda$, we apply the block matrix inversion formula to $[L + \Lambda]^{-1}$, see Lemma~\ref{lm:blk.inv.2},
such that 
\begin{equation}
	\begin{aligned}
		S_0(\bs{\lambda})& := [L + \Lambda]^{-1}\Lambda
		\\
		&=
		\begin{bmatrix}
			[\lambda_a + B_1^\prime M_{2,\lambda_b} B_1]^{-1} \lambda_a &
			-D_{1,\lambda_a}^{-1} B_1^\prime B_2
			[\lambda_b + B_2^\prime M_{1,\lambda_a} B_2]^{-1} \lambda_b
			\\
			-[\lambda_b + B_2^\prime M_{1,\lambda_a} B_2]^{-1}
			B_2^\prime B_1 D_{1,\lambda_a}^{-1} \lambda_a
			&
			[\lambda_b + B_2^\prime M_{1,\lambda_a} B_2]^{-1} \lambda_b
		\end{bmatrix}.
	\end{aligned}
\end{equation}
For given $\lambda_b$, as $\lambda_a \rightarrow \infty$, we have 
\begin{equation}
	\begin{aligned}
		S_0(\bs{\lambda})
		& \rightarrow
		\begin{bmatrix}
			I_r & 
			0_{r\times c} \\
			-[\lambda_b+B_2^\prime B_2]^{-1}B_2^\prime B_1 & 
			[\lambda_b + B_2^\prime  B_2]^{-1} \lambda_b
		\end{bmatrix}
	\end{aligned}
\end{equation}
and our estimator $\hat{\bs{\theta}}(\bs{\lambda})=\mathcal{R}[\hat{\bs{\theta}}^{\text{ls}}-S_0(\bs{\lambda})\hat{\bs{\theta}}^{\text{ls}}]$ satisfies
\begin{equation}
	\begin{aligned}[b]
		\hat{\bs{\theta}}(\bs{\lambda})
		&\rightarrow 
		\mathcal{R}
		\begin{bmatrix}
			\hat{\bs{\alpha}}^{\text{ls}}
			\\
			\hat{\bs{\beta}}^{\text{ls}}
		\end{bmatrix}
		-
		\mathcal{R}
		\begin{bmatrix}
			\hat{\bs{\alpha}}^{\text{ls}}
			\\
			-[\lambda_b+B_2^\prime B_2]^{-1}B_2^\prime B_1\hat{\bs{\alpha}}^{\text{ls}}
			+ 
			[\lambda_b+B_2^\prime B_2]^{-1}\lambda_b \hat{\bs{\beta}}^{\text{ls}}
		\end{bmatrix}
		\\
		&=
		\mathcal{R}
		\begin{bmatrix}
			\bs{0}_r
			\\
			[\lambda_b+B_2^\prime B_2]^{-1}B_2^\prime 
			[B_1\hat{\bs{\alpha}}^{\text{ls}} + B_2 \hat{\bs{\beta}}^{\text{ls}}]
		\end{bmatrix} 
		= _{(1)}
		\mathcal{R}
		\begin{bmatrix}
			\bs{0}_r
			\\
			[\lambda_b+B_2^\prime B_2]^{-1}B_2^\prime \bs{Y}
		\end{bmatrix}, 
	\end{aligned}
\end{equation}
where $=_{(1)}$ follows from
$
B^\prime (B_1\hat{\bs{\alpha}}^{\text{ls}} + B_2\hat{\bs{\beta}}^{\text{ls}})
=
B^\prime B\hat{\bs{\theta}}^{\text{ls}}
=
L  L^{-}B^\prime \bs{Y}
=B^\prime \bs{Y}
$ using $LL^{-}B'=B'$.
\end{proof}

\begin{proof}[Proof of Lemma \ref{lm:1way.2}]

Continuing from the above proof, when $\lambda_a= 0$ and $\lambda_b \ne 0$, we have 
\begin{equation}
	\begin{aligned}
		S_0(\bs{\lambda})
		& =
		\begin{bmatrix}
			0_{r\times r} & 
			- (B_1'B_1)^{-1} B_1'B_2 [\lambda_b + L_{2,\perp}]^{-1}\lambda_b \\
			0_{c\times r} & 
			[\lambda_b + L_{2,\perp}]^{-1} \lambda_b
		\end{bmatrix}
	\end{aligned}
\end{equation}
where $L_{2,\perp} := B_2' M_1 B_2$ is the Laplacian matrix of the unipartite graph after projecting out the $\bs{\alpha}$-dimension of the bipartite graph. 
In this case, our estimator $\hat{\bs{\theta}}(\bs{\lambda})=\mathcal{R}[\hat{\bs{\theta}}^{\text{ls}}-S_0(\bs{\lambda})\hat{\bs{\theta}}^{\text{ls}}]$ satisfies
\begin{equation}
	\begin{aligned}[b]
		\hat{\bs{\theta}}(\bs{\lambda})
		& =
		\mathcal{R}
		\begin{bmatrix}
			\hat{\bs{\alpha}}^{\text{ls}}
			\\
			\hat{\bs{\beta}}^{\text{ls}}
		\end{bmatrix}
		-
		\mathcal{R}
		\begin{bmatrix}
			-(B_1'B_1)^{-1} B_1'B_2 [\lambda_b + L_{2,\perp}]^{-1}\lambda_b \hat{\bs{\beta}}^{\text{ls}}
			\\ 
			[\lambda_b+L_{2,\perp}]^{-1}\lambda_b \hat{\bs{\beta}}^{\text{ls}}
		\end{bmatrix}
		\\
		& =
		\mathcal{R}
		\begin{bmatrix}
			\hat{\bs{\alpha}}^{\text{ls}} +
			(B_1'B_1)^{-1} B_1'B_2 [\lambda_b + L_{2,\perp}]^{-1}\lambda_b \hat{\bs{\beta}}^{ls}
			\\
			[\lambda_b+L_{2,\perp}]^{-1}L_{2,\perp} \hat{\bs{\beta}}^{\text{ls}}
		\end{bmatrix}. 
	\end{aligned}
\end{equation}
\end{proof}

\section{Proofs of Results in Appendix of the Paper}

\begin{proof}[Proof of Lemma \ref{lm:decomp}]
	Using the relation $S+S_1=\mathcal{R}$, one can write $l_w(\hat{\bs{\theta}},\bs{\theta}) - \text{URE}(\bs{\lambda})$ as 
	\begin{equation}
		\label{eq:decomp1}
		\begin{aligned}
			& \hspace{1.3em} l_w(\hat{\bs{\theta}},\bs{\theta}) - \text{URE}(\bs{\lambda})
			= 
			\left[\hat{\bs{\theta}}^{\text{ls}}-\bs{\theta}-S\hat{\bs{\theta}}^{\text{ls}}+S\bs{v}\right]^\prime W 
			\left[\hat{\bs{\theta}}^{\text{ls}}-\bs{\theta}-S\hat{\bs{\theta}}^{\text{ls}}+S\bs{v}\right]
			\\  & - 
			\left\{
			\tr\left[2S_1 L^- W - WL^- \right]
			+ (\hat{\bs{\theta}}^{\text{ls}} -\bs{v})'S'WS(\hat{\bs{\theta}}^{\text{ls}} -\bs{v}) \right \} \\
			&=
			[\hat{\bs{\theta}}^{\text{ls}}-\bs{\theta}]^\prime W[\hat{\bs{\theta}}^{\text{ls}}-\bs{\theta}] 
			-2(S\hat{\bs{\theta}}^{\text{ls}})^\prime W (\hat{\bs{\theta}}^{\text{ls}}-\bs{\theta})+2(S \bs{v})^\prime W[\hat{\bs{\theta}}^{\text{ls}}-\bs{\theta}]\\
			& -
			2\tr[S_1L^- W ] +\tr[WL^- ], 
		\end{aligned}
	\end{equation}
	where the second term in the last equality satisfies
	\begin{equation}
		\label{eq:decomp2}
		\begin{aligned}
			&(S\hat{\bs{\theta}}^{\text{ls}})^\prime W(\hat{\bs{\theta}}^{\text{ls}}-\bs{\theta})
			=
			(S(\hat{\bs{\theta}}^{\text{ls}}-\bs{\theta}))^\prime W (\hat{\bs{\theta}}^{\text{ls}}-\bs{\theta})
			+(S\bs{\theta})^\prime W (\hat{\bs{\theta}}^{\text{ls}}-\bs{\theta}) \\
			&=
			(\hat{\bs{\theta}}^{\text{ls}}-\bs{\theta})^\prime W (\hat{\bs{\theta}}^{\text{ls}}-\bs{\theta})
			-(S_1(\hat{\bs{\theta}}^{\text{ls}}-\bs{\theta}))^\prime W (\hat{\bs{\theta}}^{\text{ls}}-\bs{\theta})
			+(S\bs{\theta})^\prime W (\hat{\bs{\theta}}^{\text{ls}}-\bs{\theta}). 
		\end{aligned}
	\end{equation}
	Substituting \eqref{eq:decomp2} into \eqref{eq:decomp1} and we have the desired expression.
\end{proof}

\begin{proof}[Proof of  Lemma \eqref{eq:E8}]
	To show part (a), we have 
	\begin{equation}
		\begin{aligned}[b]
			\hat{\bs{\theta}}^{\text{ls}} - \bs{\theta}
			=_{(1)}
			\mathcal{R}  L^- B^\prime\bs{Y} -
			\mathcal{R}  L^- L \bs{\theta}
			=
			\mathcal{R}  L^- B^\prime[\bs{Y} - B\bs{\theta}] 
			=
			\mathcal{R}  L^- B^\prime\bs{U},
		\end{aligned}
	\end{equation}
	where $=_{(1)}$ follows from definition of $\hat{\bs{\theta}}^{\text{ls}}$ and \eqref{eq:E6}.

	To show part (b), 
	with a slight abuse of notation, we replace the double subscript of $\left\{u_{it}\right\}$ with a single subscript $\left\{u_k\right\}$. 
	\begin{equation}
		\begin{aligned}[b]
			\mathbb{V}[\bs{U}^\prime A \bs{U}]
			&=
			\mathbb{V}\left[\sum_{k\neq l}A_{kl}u_ku_l + \sum_{k=1}A_{kk}u_k^2\right],
		\end{aligned}
	\end{equation}
	\begin{equation}
		\begin{aligned}[b]
			\mathbb{V}\left[\sum_{k\neq l}A_{kl}u_ku_l\right]
			&=_{(1)}
			\sum_{k\neq l}A_{kl}^2\mathbb{V}[u_k]\mathbb{V}[u_l] 
			=
			\sigma^4\sum_{k\neq l}A_{kl}^2 
			\leq\sigma^4\tr[A'A], \\
			\mathbb{V}\left[\sum_{k=1}A_{kk}u_k^2\right] 
			&=_{(2)}
			\sum_{k=1}A_{kk}^2\mathbb{V}[u_k^2] 
			=
			\sum_{k=1}A_{k}^2(\mathbb{E}[u_k^4] - \sigma^4) 
			\leq
			(\mathbb{E}[u_k^4] - \sigma^4)\tr[A'A],
		\end{aligned}
	\end{equation}
	where $=_{(1)}$ and $=_{(2)}$ follow from independence of $\left\{u_{it}\right\}$. Thus,
	\begin{equation}
		\begin{aligned}[b]
			\mathbb{V}[\bs{U}^\prime A \bs{U}] 
			&= 
			\mathbb{V}\left[\sum_{k\neq l}A_{kl}u_ku_l + \sum_{k=1}A_{kk}u_k^2\right] 
			\\
			&=	
			\mathbb{V}\left[\sum_{k\neq l}A_{kl}u_ku_l\right]
			+
			\mathbb{V}\left[\sum_{k=1}A_{kk}u_k^2\right]
			+
			2Cov\left[\sum_{k\neq l}A_{kl}u_ku_l \ , \ \sum_{k=1}A_{kk}u_k^2\right] \\
			&\leq_{(1)}
			\left(\mathbb{E}[u_k^4] + 2\sigma^2\sqrt{\mathbb{E}[u_k^4] - \sigma^4}\right)\tr[A'A],
		\end{aligned}
	\end{equation}
	where $\leq_{(1)}$ follow by the Cauchy-Schwarz inequality. 
	
	To show part (c), note that
	\begin{equation}
		\begin{aligned}[b]
			&S_1 = R [ L + \Lambda^* ]^{-1} L
			=_{(1)}
			R L^- L \cdot [L + \Lambda^*]^{-1} L 
			=_{(2)}
			R L^\dagger L \cdot \Lambda^{*-1} B' [B\Lambda^{*-1}B' + I]^{-1} B \\ 
			&=
			R L^\dagger B' \cdot B\Lambda^{*-1} B' [B\Lambda^{*-1}B' + I]^{-1} B 
			=_{(3)} R B^\dagger G B,
		\end{aligned}
	\end{equation}
	where $=_{(1)}$ follows from $\mathcal{R}
	=
	\mathcal{R} L^-L $, see Lemma \eqref{eq:E6}, $=_{(2)}$ follows from $B\mathcal{R}=B$, see Lemma \eqref{eq:E3} and 
	$D^{-1}C(A-BDC)^{-1} =(D-CA^{-1}B)^{-1}CA^{-1}$
	for conformal matrices, and $=_{(3)}$ follows from property of the Moore-Penrose inverse.
	
	Since $G + [B\Lambda^{*-1}B' + I]^{-1} = I$, we have 
	$
			||G||
			=
			|| I - [B\Lambda^{*-1}B' + I]^{-1} || 
			\leq_{(1)} 1,
$
	where $\leq_{(1)}$ holds because the eigenvalues of the positive semidefinite $B\Lambda^{*-1}B' + I$ are no smaller than $1$, and hence the eigenvalues of its inverse are no greater than $1$.
\end{proof}	

\begin{proof}[Proof of Lemma \ref{lm: Vdefine}]
	
	We have
	$S_1 = _{(1)}  \mathcal{R}[L+\Lambda^*]^{-1}L 
	=_{(2)}  \mathcal{R}  L^{-}L[L+\Lambda^*]^{-1}L
	=_{(3)}  \newline \mathcal{R}  L^\dagger L[L+\Lambda^*]^{-1}L
	=_{(4)}  \mathcal{R}  [L^\dagger]^{1/2} \cdot   V_\lambda \cdot L^{1/2-\epsilon}$,
	where $=_{(1)}$ is the definition of $S_1$, $=_{(2)}$ uses  $R=RL^{-}L$, see Lemma
	\eqref{eq:E6}, $=_{(3)}$ use 
	$L^- = \mathcal{R}  L^\dagger \mathcal{R}'$, $\mathcal{R}^2=\mathcal{R}$, and $\mathcal{R}'L=L$ because $B\mathcal{R}=B$, $=_{(4)}$ uses $[L^\dagger]^{1/2}L^{1/2} = [L^\dagger]L$ and the definition of $V_\lambda $.
\end{proof}

%

{}
\label{subsec:bd.deriv}

\begin{proof}[Proof of Lemma 	\ref{lm:bd.deriv}]
	Our objective is to bound 
	\begin{align}
		\sup_{\tilde{\bs{\lambda}}\in\tilde{\mathcal{J}} }
		\lVert \frac{\partial V_\lambda}{\partial x} \rVert
		&=
		\sup_{\tilde{\bs{\lambda}}\in\tilde{\mathcal{J}} }
		\lVert 
		L^{1/2+\epsilon}
		\frac{\partial}{\partial x}
		[L + \Lambda^*]^{-1} L^{1/2}
		\rVert
		\label{eq:deriv.bound}
	\end{align}
	for
	$x\in\{\tilde{\lambda}_a,\tilde{\lambda}_b,\phi\}$, where $\epsilon^{-1}$ is a given large positive even number. For presentation simplicity, we first define
	\begin{align}
		\check{V}_\lambda:&=L^{\epsilon}  [L+\Lambda^*]^{-1}  L, \text{ such that } \nonumber\\
		V_\lambda:&=L^{1/2}  [L+\Lambda^*]^{-1}  L^{1/2+\epsilon}=L^{1/2-\epsilon}\check{V}_\lambda [L^\dagger]^{1/2-\epsilon},
		\label{eq:deriv.bound.v}
	\end{align} 
	We first show how to bound 
	\begin{align}
		\sup_{\tilde{\bs{\lambda}}\in\tilde{\mathcal{J}} }
		\lVert \frac{\partial \check{V}_\lambda}{\partial x} \rVert
		&=
		\sup_{\tilde{\bs{\lambda}}\in\tilde{\mathcal{J}} }
		\lVert 
		L^{\epsilon}
		\frac{\partial}{\partial x}
		[L + \Lambda^*]^{-1} L
		\rVert,
		\label{eq:deriv.bound2}
	\end{align}
	and then discuss how to bound $\textstyle\sup_{\tilde{\bs{\lambda}}\in\tilde{\mathcal{J}} }
	\lVert \frac{\partial V_\lambda}{\partial x} \rVert$.
	
	We start with the derivative of $\check{V}_\lambda$ w.r.t. $\tilde{\lambda}_a$. The steps for derivative w.r.t. $\tilde{\lambda}_b$ are obviously symmetric, whereas those for the derivative w.r.t. $\phi$ are elaborated in Remarks~\ref{rmk:deriv.phi} and \ref{rmk:deriv.phi.1}. Finally, the steps for the derivatives of $V_\lambda$ w.r.t. $(\tilde{\lambda}_a,\tilde{\lambda}_b,\phi)$ are elaborated in Remarks~\ref{rmk:deriv.v} and \ref{rmk:deriv.v.1}.

	The rest of the proof proceeds as follows:
	We first writes the derivative of \eqref{eq:deriv.bound2} explicitly, which the subsequent sections utilize in bounding \eqref{eq:deriv.bound2} in different regimes. We first takes the regime where both $\lambda_a,\lambda_b\leq1$, then takes the regimes where at least one of $\lambda_a,\lambda_b$ is greater than 1.
	
	
	\textbf{Step 1. Take Derivative.} 
	Note that
	\begin{equation}
		\frac{\partial}{\partial \tilde{\lambda}_a}
		L^\epsilon
		[L + \Lambda^*]^{-1} L
		=
		\frac{\partial}{\partial \lambda_a^{1/2}}
		L^\epsilon
		[L + \Lambda^*]^{-1} L
		\cdot
		\frac{\partial\lambda_a^{1/2}}{\partial \tilde{\lambda_a}}.
		\label{eq:deriv.main}
	\end{equation}
	Following Lemma~\ref{lm:take.deriv},
	\begin{align}
		\frac{\partial\lambda_a^{1/2}}{\partial\tilde{\lambda}_a}
		=&
		\frac{\partial\lambda_a^{1/2}}{\partial\lambda_a^{-1}} \cdot
		\frac{\partial\lambda_a^{-1}}{\partial\tilde{\lambda}_a}
		=
		\frac{kn}{2}
		\left[1+\lambda_a^{\frac{1}{k}}\right]^{\frac{n+1}{n}}
		\lambda_a^{\frac{1}{2}-\frac{1}{kn}}
		=_{(1)}
		\frac{1}{2\epsilon}
		\left[1+\lambda_a^{\frac{1}{k}}\right]^{\frac{k}{2}}
		\lambda_a^{\frac{1}{2}-\epsilon}
		,
		\label{eq:deriv.lambda.tilde}
	\end{align}
	where $=_{(1)}$ follows from our choice of $(n,k)$ such that $\epsilon=(kn)^{-1}$.
	
	Let $\Lambda_a$ be a diagonal matrix with $1$ on the first $r$ diagonals and zero everywhere else. Then we have 
	\begin{equation}
		\begin{aligned}[b]
			&\frac{\partial}{\partial \lambda_a^{1/2}}
			L^\epsilon
			[L + \Lambda^*]^{-1} L
			=
			L^\epsilon
			[L + \Lambda^*]^{-1}
			\left\{
			\Lambda_a
			[-\phi\mathcal{A}+I]\Lambda^{1/2}
			+
			\Lambda^{1/2}[-\phi\mathcal{A}+I]\Lambda_a
			\right\}
			[L + \Lambda^*]^{-1}
			L
			\\
			=&
			L^\epsilon
			[L + \Lambda^*]^{-1}
			\cdot
			\begin{bmatrix}
				2\lambda_a^{1/2} & -\phi\lambda_b^{1/2}\mathcal{A}_{12} \\
				-\phi\lambda_b^{1/2}\mathcal{A}_{21} & 0
			\end{bmatrix}
			\cdot
			[L + \Lambda^*]^{-1}
			L.
		\end{aligned}
		\label{eq:deriv}
	\end{equation}
	This implies that bounding \eqref{eq:deriv.bound2} is equivalent to bounding\footnote{Any power of $[L+\Lambda^*]^{-1}$ may be defined using its eigendecomposition.} 
	\begin{align}
		L^\epsilon
		[L + \Lambda^*]^{-\epsilon}
		\cdot 
		[L + \Lambda^*]^{-(1-\epsilon)}
		\begin{bmatrix}
			\lambda_a^{1/2} & 0 \\
			-\phi\lambda_b^{1/2}\mathcal{A}_{21} & 0
		\end{bmatrix}
		\frac{\partial\lambda_a^{1/2}}{\partial\tilde{\lambda}_a} 
		\cdot
		[L + \Lambda^*]^{-1}
		L.
		\label{eq:deriv.1}
	\end{align}
	We can show that the term that replace $-\phi\lambda_b^{1/2}\mathcal{A}_{21}$ with $-\phi\lambda_b^{1/2}\mathcal{A}_{12}$ in the symmetric diagonal position can be bounded by the same arguments.
	
	\textbf{Step 2. Study the case with small $\lambda_a$ and small $\lambda_b$.}
	In this section, we assume that we are in the regime where $\lambda_a,\lambda_b\leq 1$. Recall that our goal is to uniformly bound \eqref{eq:deriv.1}, and that from \eqref{eq:deriv.lambda.tilde}
	\begin{equation}
		\frac{\partial\lambda_a^{1/2}}{\partial\tilde{\lambda}_a}
		=
		\frac{1}{2\epsilon}
		\left[1+\lambda_a^{\frac{1}{k}}\right]^{\frac{k}{2}}
		\cdot
		\lambda_a^{\frac{1}{2}-\epsilon} = O(\lambda_a^{\frac{1}{2}-\epsilon} ),
	\end{equation}
	because $\lambda_a\leq1$. As such, our task reduces to bounding
	\begin{equation}
		L^\epsilon
		[L + \Lambda^*]^{-\epsilon}
		\cdot 
		[L + \Lambda^*]^{-(1-\epsilon)}
		\begin{bmatrix}
			\lambda_a^{1-\epsilon} & 0 \\
			-\phi\lambda_a^{1/2-\epsilon}\lambda_b^{1/2}\mathcal{A}_{21} & 0
		\end{bmatrix}
		\cdot
		[L + \Lambda^*]^{-1}
		L.
		\label{eq:deriv.small.a.1}
	\end{equation}
	We will bound each of the three terms in  \eqref{eq:deriv.small.a.1}.
	
	The first term in \eqref{eq:deriv.small.a.1} is bounded following Lemma~\ref{lm:deriv.bound.lower} because $\epsilon^{-1}$ is an even number by construction.
	
	To bound the norm of the third term in \eqref{eq:deriv.small.a.1} is equivalent to bounding the norm of $[L + \Lambda^*]^{-1}\Lambda^*$. Because $||\phi\mathcal{A}_{21}||=O(1)$, this follows from
	\begin{equation}
		\begin{aligned}[b]
			||[L + \Lambda^*]^{-1}|| \cdot ||\Lambda^*||
			\leq& 
			|| [L+\Lambda^*]^{-1} || 
			\cdot 
			|| \lambda_a + \lambda_a^{1/2}\lambda_b^{1/2} ||
		\end{aligned}
		\label{eq:deriv.small.a.2.1}
	\end{equation}
	which is indeed bounded, because for instance
	\begin{equation}
		\begin{aligned}[b]
			[L+\Lambda^*]^{-1} \lambda_a 
			&=[L+\Lambda^*]^{-1} \lambda_a 
			\left[\left\{\lambda_a>\lambda_b\right\} + \left\{\lambda_b>\lambda_a\right\} \right]
			\\
			&=_{(1)} 
			O(\lambda_a^{-1}) \lambda_a \left\{\lambda_a>\lambda_b\right\}
			+
			O(\lambda_b^{-1}) \lambda_b \frac{\lambda_a}{\lambda_b} \left\{\lambda_b>\lambda_a\right\}
			= O(1) + O(1),
		\end{aligned}
		\label{eq:deriv.small.a.2.2}
	\end{equation}
	where $=_{(1)}$ follows from Lemma~\ref{lm:deriv.bound}.
	
	Now we bound the norm of the second term in \eqref{eq:deriv.small.a.1}, which is achieved by bounding
	\begin{align}
		&\left([L + \Lambda^*]^{-1}\right)^{1-\epsilon}
		(\lambda_a^{1-\epsilon}+ \lambda_a^{1/2-\epsilon}\lambda_b^{1/2})
		\nonumber\\
		=&
		\left([L + \Lambda^*]^{-1} \cdot\lambda_a\right)^{1-\epsilon}
		+
		\left([L + \Lambda^*]^{-1} \cdot\lambda_a\right)^{1/2-\epsilon}
		\left([L + \Lambda^*]^{-1} \cdot\lambda_b\right)^{1/2}
		\label{eq:deriv.small.a.3}
	\end{align}
	which is indeed bounded by steps identical to \eqref{eq:deriv.small.a.2.2}.
	
	\begin{remark}[Bounding the derivative of $\phi$]
		\label{rmk:deriv.phi}
		The steps for bounding 
		\begin{equation*}
			\sup_{\tilde{\bs{\lambda}}\in\tilde{\mathcal{J}}}
			\lVert 
			L^\epsilon
			\frac{\partial}{\partial \phi}
			[L + \Lambda^*]^{-1} L
			\rVert
		\end{equation*}
		is similar to the arguments above: our task is reduced to bounding,  similar to \eqref{eq:deriv.1} and \eqref{eq:deriv.small.a.1}, 
		\begin{equation}
			L^\epsilon
			[L + \Lambda^*]^{-\epsilon}
			\cdot 
			[L + \Lambda^*]^{-(1-\epsilon)}
			\cdot
			\begin{bmatrix}
				0 & 0 \\
				-\lambda_a^{1/2}\lambda_b^{1/2}\mathcal{A}_{21} & 0
			\end{bmatrix}
			\cdot
			[L + \Lambda^*]^{-1}
			L,
		\end{equation}
		which is achieved by bounding 
		$
		[L + \Lambda^*]^{-1}
		\lambda_a^{1/2}\lambda_b^{1/2}
		\cdot
		[L + \Lambda^*]^{-1}
		L.
		$
		This is further reduced to bounding the following items,
		$
		[L + \Lambda^*]^{-1}\lambda_a,
		\quad 
		[L + \Lambda^*]^{-1}\lambda_b,
		\quad 
		[L + \Lambda^*]^{-1} L,
		$
		all of which are exactly the terms studied previously. 
		\qed
	\end{remark}

	\begin{remark}[Bounding of $V_\lambda$]
		\label{rmk:deriv.v}
		Recall that from \eqref{eq:deriv.bound.v}
		\begin{equation}
			\frac{\partial}{\partial \tilde{\lambda}_a} V_\lambda
			=
			L^{1/2-\epsilon} \cdot
			\frac{\partial}{\partial \tilde{\lambda}_a} \check{V}_\lambda \cdot 
			[L^\dagger]^{1/2-\epsilon},
			\label{eq:deriv.v.0}
		\end{equation} 
		so that bounding $\tfrac{\partial}{\partial \tilde{\lambda}_a} V_\lambda$ in this regime only requires a different factorization of the terms from the workings above. To elaborate: recall that $\tfrac{\partial}{\partial \tilde{\lambda}_a} \check{V}_\lambda$ is bounded by \eqref{eq:deriv.small.a.1}, and thus substituting out $\tfrac{\partial}{\partial \tilde{\lambda}_a} \check{V}_\lambda$ from \eqref{eq:deriv.v.0} using \eqref{eq:deriv.small.a.1} reveals that bounding $\tfrac{\partial}{\partial \tilde{\lambda}_a} V_\lambda$ is equivalent to bounding
		\begin{equation}
			\begin{aligned}[b]
				L^{\frac{1}{2}}
				[L + \Lambda^*]^{-\frac{1}{2}}
				\cdot 
				[L + \Lambda^*]^{-\frac{1}{2}}
				\begin{bmatrix}
					\lambda_a^{1-\epsilon} & 0 \\
					-\phi\lambda_a^{1/2-\epsilon} \lambda_b^{1/2}\mathcal{A}_{21} & 0
				\end{bmatrix}
				[L + \Lambda^*]^{-(\frac{1}{2}-\epsilon)} 
				\cdot
				[L + \Lambda^*]^{-(\frac{1}{2}+\epsilon)}
				L^{\frac{1}{2}+\epsilon},
			\end{aligned}
			\label{eq:deriv.v.1}
		\end{equation}
		where the first and third terms are bounded by Lemma~\ref{lm:deriv.bound.lower} -- this step for the first term is straightforward whereas the third term deserves more elaboration:  
		\begin{equation}
			\begin{aligned}[b]
				&|| L^{\frac{1}{2}+\epsilon}
				[L + \Lambda^*]^{-(\frac{1}{2}+\epsilon)} ||^2
				=
				\bar{\rho}( L^{\frac{1}{2}+\epsilon}[L + \Lambda^*]^{-(1+2\epsilon)}L^{\frac{1}{2}+\epsilon}) 
				=
				\bar{\rho}( L[L + \Lambda^*]^{-1} \cdot [L + \Lambda^*]^{-2\epsilon}L^{2\epsilon}) \\
				\leq&
				|| L[L + \Lambda^*]^{-1} \cdot [L + \Lambda^*]^{-2\epsilon}L^{2\epsilon} || 
				\leq
				|| L[L + \Lambda^*]^{-1} || \cdot || [L + \Lambda^*]^{-2\epsilon}L^{2\epsilon} || \\
				=&_{(1)}
				|| I - \Lambda^*[L + \Lambda^*]^{-1} || \cdot || [L + \Lambda^*]^{-2\epsilon}L^{2\epsilon} || ,
			\end{aligned}
		\end{equation}
		where the first term of $=_{(1)}$ is bounded as in \eqref{eq:deriv.small.a.2.2}, and the second term of $=_{(1)}$ is bounded by Lemma~\ref{lm:deriv.bound.lower} because $\tfrac{1}{2}\epsilon^{-1}$ is an even number by construction. 
		
		Finally the middle term of \eqref{eq:deriv.v.1} is bound by 
		$
		[L + \Lambda^*]^{-\frac{1}{2}}
		\cdot
		(\lambda_a^{1-\epsilon}+ \lambda_a^{1/2-\epsilon}\lambda_b^{1/2})
		\cdot
		[L + \Lambda^*]^{-(\frac{1}{2}-\epsilon)},
		$
		which is bounded once we bound the following terms
		$
		[L + \Lambda^*]^{-(\frac{1}{2}-\epsilon)}\lambda_a^{\frac{1}{2}-\epsilon}
		,\quad
		[L + \Lambda^*]^{-\frac{1}{2}}\lambda_a^{\frac{1}{2}}
		,\quad
		[L + \Lambda^*]^{-\frac{1}{2}}\lambda_b^{\frac{1}{2}}
		$.
		These steps are identical to \eqref{eq:deriv.small.a.3}.
		\qed
	\end{remark}
	
	
	\textbf{Step 3. Study the case where at least one of $\lambda_a,\lambda_b$ is large.}
	We will now define several matrices for notational compactness. Let $V:=L + \Lambda^*$. Also let $\bar{\Lambda}:= \Lambda^{1/2} ( -|\phi|I + I ) \Lambda^{1/2}$ which is a diagonal matrix, and lastly $\bar{V}:= L + \bar{\Lambda}$. Note that
	\begin{equation}
		\begin{aligned}[b]
			V_1 =& B_1^\prime B_1 + \lambda_a, \qquad
			& \bar{V}_1 =& B_1^\prime B_1 + \lambda_a(1-|\phi|),
			\\
			V_{12} =& B_1^\prime B_2 - \phi\lambda_a^{1/2}\lambda_b^{1/2}\mathcal{A}_{12}, \qquad
			& \bar{V}_{12} =& B_1^\prime B_2,
			\\
			V_2 =& B_2^\prime B_2 + \lambda_b, \qquad
			& \bar{V}_2 =& B_2^\prime B_2 + \lambda_b(1-|\phi|).
		\end{aligned}
	\end{equation}
	
	For a generic matrix $M$, we use $[M^{-1}]_{1}$, $[M^{-1}]_{12}$, and $[M^{-1}]_{2}$ to denote the top left, top right, and bottom right blocks of the matrix inverse $M^{-1}$ respectively. On the other hand, we use $M_{1}^{-1}$, $M_{2}^{-1}$ to denote the inverses of the top left and bottom right blocks of the matrix $M$ respectively.
	
	We now work in the regime where at least one of $\lambda_a,\lambda_b$ is greater than 1.
	Recall our goal is to uniformly bound \eqref{eq:deriv.bound2}, which reduces to bounding \eqref{eq:deriv.1}.
	Because the norm of $L$ is bounded, and recalling \eqref{eq:deriv.lambda.tilde}, our goal thus further reduces to bounding the norm of
	\begin{align}
		&[L + \Lambda^*]^{-1}
		\cdot
		\begin{bmatrix}
			\lambda_a^{1/2} & 0 \\
			-\phi\lambda_b^{1/2}\mathcal{A}_{21} & 0
		\end{bmatrix}
		\cdot 
		[1+\lambda_a^{\frac{1}{k}}]^{\frac{k}{2}}
		\lambda_a^{\frac{1}{2}-\epsilon}
		\cdot
		[L + \Lambda^*]^{-1}
		\nonumber
		\\
		=&
		\begin{bmatrix}
			[V^{-1}]_{1} & [V^{-1}]_{12} \\
			[V^{-1}]_{21} & [V^{-1}]_{2}
		\end{bmatrix}
		\cdot
		\begin{bmatrix}
			\lambda_a^{1-\epsilon} & 0 \\
			-\phi\lambda_a^{1/2-\epsilon}\lambda_b^{1/2}\mathcal{A}_{21} & 0
		\end{bmatrix}
		\cdot 
		\begin{bmatrix}
			[1+\lambda_a^{\frac{1}{k}}]^{\frac{k}{2}}& 0 \\
			0 & 0
		\end{bmatrix}
		\cdot
		\begin{bmatrix}
			[V^{-1}]_{1} & [V^{-1}]_{12} \\
			[V^{-1}]_{21} & [V^{-1}]_{2}
		\end{bmatrix} 
		\nonumber
		\\
		=&
		\begin{bmatrix}
			[V^{-1}]_{1}\lambda_a^{1-\epsilon} - [V^{-1}]_{12}\mathcal{A}_{21}\phi\lambda_a^{1/2-\epsilon}\lambda_b^{1/2}  
			& 
			0
			\\
			[V^{-1}]_{21}\lambda_a^{1-\epsilon} - [V^{-1}]_{2}\mathcal{A}_{21}\phi\lambda_a^{1/2-\epsilon}\lambda_b^{1/2}  
			&
			0
		\end{bmatrix}
		\cdot 
		\begin{bmatrix}
			O(1+\lambda_a^{1/2}) & 0 \\
			0 & 0
		\end{bmatrix}
		\begin{bmatrix}
			[V^{-1}]_{1} & [V^{-1}]_{12} \\
			[V^{-1}]_{21} & [V^{-1}]_{2}
		\end{bmatrix}.
		\label{eq:remainingcases.1} 
	\end{align}
	First, we take the top-left block of the left matrix in \eqref{eq:remainingcases.1}. First note that from Lemma~\ref{lm:deriv.bound.non.small}, $[V^{-1}]_{1}\lambda_a^{1-\epsilon} = O(1)$.
	Second note that
	\begin{align}
		[V^{-1}]_{12} \lambda_a^{1/2-\epsilon} \lambda_b^{1/2}
		=_{(1)}& 
		O(1\wedge\lambda_a^{-1}) \lambda_a^{1-\epsilon}
		\cdot
		O(\lambda_a^{1/2}\lambda_b^{1/2}) \frac{1}{\lambda_a^{1/2}\lambda_b^{1/2}}
		\cdot
		O(1\wedge\lambda_b^{-1}) \lambda_b 
		=
		O(1),
	\end{align}
	where $=_{(1)}$ follows from the third equation of Lemma~\ref{lm:deriv.bound.non.small}.
	
	Next, we take the bottom-left block of the left matrix in \eqref{eq:remainingcases.1}. Note that 
	\begin{equation}
		\begin{aligned}
			&[V^{-1}]_{21} \lambda_a^{1-\epsilon} - [V^{-1}]_{2} \mathcal{A}_{21} \phi\lambda_a^{1/2-\epsilon}\lambda_b^{1/2}
			=
			\left([V^{-1}]_{21} \lambda_a^{1/2} - [V^{-1}]_{2} \mathcal{A}_{21} \phi\lambda_b^{1/2}\right)\cdot \lambda_a^{1/2-\epsilon}
			\\
			=&_{(1)}
			-[V^{-1}]_{2}
			\left[B_2'B_1 \lambda_a^{1/2} + \phi\lambda_b^{1/2}\mathcal{A}_{21}B_1'B_1\right] V_1^{-1} 
			\cdot \lambda_a^{1/2-\epsilon} \\
			=&
			-[V^{-1}]_{2} 
			\left[B_2'B_1 \cdot V_1^{-1}\lambda_a^{1-\epsilon} + \lambda_b^{1/2} \cdot \phi\mathcal{A}_{21}B_1'B_1 \cdot V_1^{-1}\lambda_a^{1/2-\epsilon} \right]
			\\
			=&_{(2)}
			O(1\wedge\lambda_b^{-1})\left[ O(1) + \lambda_b^{1/2}\cdot O(1)\right]
			=
			O(1),
		\end{aligned}
	\end{equation}
	where $=_{(1)}$ follows from the second equation of Lemma~\ref{lm:factor}, and $=_{(2)}$ follows from Lemma~\ref{lm:deriv.bound.non.small}.
	
	Finally, we take the right matrix (product) of \eqref{eq:remainingcases.1}. Note that $[V^{-1}]_{1} O(1+\lambda_a^{1/2})=O(1)$ by Lemma~\ref{lm:deriv.bound.non.small}, and
	\begin{align}
		[V^{-1}]_{12} O(1+\lambda_a^{1/2})
		=_{(1)} 
		O(1\wedge \lambda_a^{-1}) O(\lambda_a^{1/2}+\lambda_a) O(\lambda_b^{1/2}) O(1\wedge \lambda_b^{-1})
		=O(1),
	\end{align} 
	where again $=_{(1)}$ follows from Lemma~\ref{lm:deriv.bound.non.small}.

	\begin{remark}[Bounding derivative wrt $\phi$]
		\label{rmk:deriv.phi.1}Replace the derivative with respect to (wrt) $\tilde{\lambda}_a$ with the derivative wrt $\phi$,
		our goal is now bounding 
		\begin{align*}
			&[L + \Lambda^*]^{-1}
			\cdot
			\begin{bmatrix}
				0 & 0 \\
				\lambda_a^{1/2}\lambda_b^{1/2} & 0
			\end{bmatrix}
			\cdot
			[L + \Lambda^*]^{-1}
			=
			[L + \Lambda^*]^{-1}
			\begin{bmatrix}
				0 & 0 \\
				\lambda_b^{1/2} & 0
			\end{bmatrix}
			\cdot
			\begin{bmatrix}
				\lambda_a^{1/2} & 0 \\
				0 & 0
			\end{bmatrix}
			[L + \Lambda^*]^{-1},
		\end{align*}
		which is achieved by bounding 
		$
		[V^{-1}]_{12}\lambda_b^{1/2},
		\quad
		[V^{-1}]_{2}\lambda_b^{1/2},
		\quad
		[V^{-1}]_{1}\lambda_a^{1/2},
		\quad
		[V^{-1}]_{12}\lambda_a^{1/2}
		$
		using Lemma~\ref{lm:deriv.bound.non.small}.
		\qed
	\end{remark}

	\begin{remark}[Bounding of $V_\lambda$]
		\label{rmk:deriv.v.1}
		Since any positive power of $L$ is bounded in norm, there is no change in the workings for bounding $V_\lambda$ in this regime.
		\qed
	\end{remark}
	This completes the proof of Lemma~\ref{lm:bd.deriv}.
\end{proof}



Below we present some results used in the proof of Lemma~\ref{lm:bd.deriv}.

\begin{lemma}[Bounding in small $\lambda$ regime]
	\label{lm:deriv.bound}
	Suppose that $\lambda_a,\lambda_b\leq1$. \newline Then $[L+\Lambda^*]^{-1}=O(\lambda_a^{-1} \wedge \lambda_b^{-1})$.
\end{lemma}
\begin{proof}[Proof of Lemma \ref{lm:deriv.bound}]
	We first write 
	\begin{equation}
		\begin{aligned}[b]
			[L + \Lambda^*]^{-1} 
			=
			&\begin{bmatrix}
				[V^{-1}]_1
				&
				[V^{-1}]_{12}
				\\
				[V^{-1}]_{21}
				&
				[V^{-1}]_2
			\end{bmatrix}.
		\end{aligned}
	\end{equation}
	We now show that $[V^{-1}]_1=O(\lambda_a^{-1}\wedge\lambda_b^{-1})$. First 
	$|| [V^{-1}]_1 || \leq || [\bar{V}^{-1}]_1 || $ 
	by Lemma~\ref{lm:submat}, and $[\bar{V}^{-1}]_1 = O(\lambda_a^{-1})$ 
	holds because $[\bar{V}^{-1}]_1\lambda_a \preceq I$; see Lemma~\ref{lm:blk.inv.1}. Similarly, $[V^{-1}]_1 = O(\lambda_b^{-1})$ because $||[V^{-1}]_1|| \leq || [\bar{V}^{-1}]_1 || $, and
	$
	[\bar{V}^{-1}]_1 
	=_{(1)}
	\bar{V}_1^{-1}
	+
	\bar{V}_1^{-1}
	\bar{V}_{12}
	[\bar{V}^{-1}]_2 \bar{V}_1^{-1}
	$
	where $=_{(1)}$ follows from \eqref{eq:blk.inv.1.2} of Lemma~\ref{lm:blk.inv.1}, and note that $[\bar{V}^{-1}]_2\lambda_b \preceq I$, whereas $\bar{V}_1^{-1}$ and $\bar{V}_{12}$ are both $O(1)$.
	An analogous argument holds for $[V^{-1}]_2$. It also holds for $[V^{-1}]_{12}$ because $[V^{-1}]_{12} = -V_1^{-1} V_{12} [V^{-1}]_2$, by applying  Lemma~\ref{lm:blk.inv.2} and using both $V_1^{-1}$ and $V_{12}$ are $O(1)$. 
\end{proof}

\begin{lemma}[Bounding in non-small $\lambda$ regime]
	\label{lm:deriv.bound.non.small}
	Suppose that at least one of $\lambda_a,\lambda_b$ is greater than one. Then we have
	\begin{align*}
		[V^{-1}]_1 = O(1\wedge \lambda_a^{-1}), 
		[V^{-1}]_2 = O(1\wedge \lambda_b^{-1}), \ \ 
		[V^{-1}]_{12} = O(1\wedge\lambda_a^{-1}) \ \ O(\lambda_a^{1/2}\lambda_b^{1/2}) O(1\wedge\lambda_b^{-1}).
	\end{align*}
\end{lemma}

\begin{proof}[Proof of Lemma \ref{lm:deriv.bound.non.small}]
	To show the first equation, note that $|| [V^{-1}]_1 || \leq ||[\bar{V}^{-1}]_1||$ by Lemma~\ref{lm:submat}. When $\lambda_a$ is greater than 1,  $[\bar{V}^{-1}]_1\lambda_a\preceq I$. When $\lambda_a$ is small, we have that $[\bar{V}^{-1}]_1$ is itself $O(1)$ because
	$
	[\bar{V}^{-1}]_1 
	=_{(1)}
	\bar{V}_1^{-1}
	+
	\bar{V}_1^{-1}
	\bar{V}_{12}
	[\bar{V}^{-1}]_2 \bar{V}_{21} \bar{V}_1^{-1},
	$
	where $=_{(1)}$ follows from \eqref{eq:blk.inv.1.2} of Lemma~\ref{lm:blk.inv.1}, and $[\bar{V}^{-1}]_2=O(1)$ because, in this regime, $\lambda_a\leq1$ implies $\lambda_b\geq1$. Note that $\bar{V}_1^{-1}$, $\bar{V}_{12}=B_1^\prime B_2$, and $\bar{V}_{21}$ are all $O(1)$ too.
	The second equation follows by the same argument. 
	For the third equation of the lemma, we have
	\begin{align}
		[V^{-1}]_{12} 
		=_{(1)}& -[V^{-1}]_1 V_{12} V_2^{-1}  
		=_{(2)}
		O(1\wedge\lambda_a^{-1}) O(\lambda_a^{1/2}\lambda_b^{1/2}) O(1\wedge\lambda_b^{-1}),
	\end{align} 
	where $=_{(1)}$ follows from Lemma~\ref{lm:blk.inv.2}, and $=_{(2)}$ follows from the first equation of this lemma.
\end{proof}

\begin{lemma}[Bounding of lower powers]
	\label{lm:deriv.bound.lower}
	Let $n$ be a positive even number.\newline  Then
$
		|| L^{1/n} [L+\Lambda^*]^{-1/n} ||\leq 1.
$
\end{lemma}

\begin{proof}[Proof of Lemma \ref{lm:deriv.bound.lower}] We have
	\begin{align}
		|| L^{1/n} [L+\Lambda^*]^{-1/n}  ||^2
		=&
		\bar{\rho}(L^{1/n} [L+\Lambda^*]^{-2/n} L^{1/n}) 
		\nonumber\\
		=&
		\bar{\rho}(L^{2/n} [L+\Lambda^*]^{-2/n}) 
		\leq_{(1)}
		|| L^{2/n} [L+\Lambda^*]^{-2/n} ||, 
	\end{align}
	where $\leq_{(1)}$ follows from Lemma~\ref{lm:norm.eig}.
	By iterating the argument above enough times, we get that $|| L^{1/n} [L+\Lambda^*]^{-1/n}  ||$ is bounded by some positive finite power of 
	\begin{align}
		\bar{\rho}(L^{1/2} [L+\Lambda^*]^{-1} L^{1/2}) 
		=&
		\bar{\rho}( [L+\Lambda^*]^{-1/2}L [L+\Lambda^*]^{-1/2})
		\leq_{(1)}1,
	\end{align}
	where $\leq_{(1)}$ holds because clearly $L \preceq L+\Lambda^*$ (see Lemma~\ref{lm:pd.var}), and
	\begin{align}
		L \preceq L+\Lambda^* 
		\iff 
		[L+\Lambda^* ]^{-1/2}L[L+\Lambda^* ]^{-1/2} \preceq I.
	\end{align}
	Note that $L+\Lambda^*$ is a positive definite matrix, because $\Lambda^*$ is positive definite -- see Lemma~\ref{lm:pd.var}.
\end{proof}

\begin{lemma}[Factorization]
	\label{lm:factor}
	\begin{equation*}
		[V^{-1}]_{12} \lambda_a^{1/2}
		-
		\mathcal{A}_{12} [V^{-1}]_2  \phi \lambda_b^{1/2}
		=\,
		-V_1^{-1} 
		\left[
		B_1^\prime B_2 \lambda_a^{1/2} 
		+B_1^\prime B_1\phi\lambda_b^{1/2}\mathcal{A}_{12}
		\right] [V^{-1}]_2.
	\end{equation*}
\end{lemma}
\begin{proof}[Proof of Lemma \ref{lm:factor}]
	We have
	\begin{equation}
		\begin{aligned}[b]
			&
			[V^{-1}]_{12} \lambda_a^{1/2}
			-
			\mathcal{A}_{12} [V^{-1}]_2  \phi \lambda_b^{1/2}
			=_{(1)}
			\left[
			-V_1^{-1} V_{12} \lambda_a^{1/2} - \mathcal{A}_{12}\phi\lambda_b^{1/2}
			\right] [V^{-1}]_2
			\\
			=&_{(2)}
			\left[
			-V_1^{-1}  B_1^\prime B_2 \lambda_a^{1/2} 
			+ (\lambda_aV_1^{-1}-I)\phi\lambda_b^{1/2}\mathcal{A}_{12}
			\right] [V^{-1}]_2
			\\
			=&_{(3)}
			\left[
			-V_1^{-1} B_1^\prime B_2 \lambda_a^{1/2} 
			- V_1^{-1}B_1^\prime B_1\phi\lambda_b^{1/2}\mathcal{A}_{12}
			\right] [V^{-1}]_2
			\\
			=&
			-V_1^{-1} 
			\left[
			B_1^\prime B_2 \lambda_a^{1/2} 
			+ B_1^\prime B_1\phi\lambda_b^{1/2}\mathcal{A}_{12}
			\right] [V^{-1}]_2,
		\end{aligned}
	\end{equation}
	where $=_{(1)}$ follows from $[V^{-1}]_{12}=[V^{-1}]_{21}^\prime = -V_1^{-1}V_{12}[V^{-1}]_2$ -- see Lemma~\ref{lm:blk.inv.2}, and note that $V^{-1}$  is symmetric, $=_{(2)}$ follows from definition of $V_{12}$, and $=_{(3)}$ follows from
	$
	V_1^{-1}\lambda_a + V_1^{-1}B_1^\prime B_1= I.
	$
\end{proof}

\begin{lemma}[Submatrix norms]
	\label{lm:submat}
	We have $|| [V^{-1}]_1 || \leq || [\bar{V}^{-1}]_1 ||$ and $|| [V^{-1}]_2 || \leq || [\bar{V}^{-1}]_2 ||$. 
\end{lemma}
\begin{proof}[Proof of Lemma \ref{lm:submat}]
	Note that $V = \bar{V} + \Lambda^{1/2}[|\phi|I-\phi\mathcal{A}]\Lambda^{1/2}$, where $\Lambda^{1/2}[|\phi|I-\phi\mathcal{A}]\Lambda^{1/2}$ is a positive semidefinite matrix, because $|| \mathcal{A} ||\leq1$. This implies that $\bar{V} \preceq V$ and thus $V^{-1}\preceq \bar{V}^{-1}$. Therefore, for every vector $x$, we have $x^\prime V^{-1} x \leq x^\prime \bar{V}^{-1} x$, including the vectors $x$ that have zeros in their lower halves. That is to say, we also have $x_1^\prime [V^{-1}]_1 x_1 \leq x_1^\prime [\bar{V}^{-1}]_1 x_1$ for every vector $x_1$. This implies that $|| [V^{-1}]_1 || \leq || [\bar{V}^{-1}]_1 ||$. The proof for $|| [V^{-1}]_2 || \leq || [\bar{V}^{-1}]_2 ||$ is analogous.
\end{proof}
\begin{lemma}[Block matrix inversion I]
	\label{lm:blk.inv.1}
	The inverse of $\bar{V}=L + \bar{\Lambda}$ takes the form 
	\begin{equation*}
		\begin{aligned}[b]
			&\begin{bmatrix}
				[\bar{V}^{-1}]_{1} & [\bar{V}^{-1}]_{12} \\
				[\bar{V}^{-1}]_{21} & [\bar{V}^{-1}]_{2}
			\end{bmatrix}
			\\
			=&
			\begin{bmatrix}
				\left[\lambda_a[1-|\phi|] + B_1^\prime \bar{M}_2 B_1\right]^{-1}  
				& 
				-\left[\lambda_a[1-|\phi|] + B_1^\prime \bar{M}_2 B_1\right]^{-1}  
				B_1B_2^\prime
				\bar{V}_2^{-1}
				\\
				-\left[\lambda_b[1-|\phi|] + B_2^\prime \bar{M}_1 B_2\right]^{-1}   
				B_2^\prime B_1 \bar{V}_1^{-1}
				&
				\left[\lambda_b[1-|\phi|] + B_2^\prime \bar{M}_1 B_2\right]^{-1}  
			\end{bmatrix}
			\label{eq:blk.inv.1.1}
		\end{aligned}
	\end{equation*}
	where
	\begin{align*}
		\bar{M}_1
		&:=
		I-B_1 (B_1^\prime B_1 + \lambda_a[1-|\phi|])^{-1} B_1^\prime, \\
		\bar{M}_2
		&:= I-B_2 (B_2^\prime B_2 + \lambda_b[1-|\phi|])^{-1} B_2^\prime
	\end{align*}
	are positive semidefinite matrices. 
	Furthermore, 
	\begin{align}
		[\bar{V}^{-1}]_{1}
		&=
		\bar{V}_1^{-1} + 
		\bar{V}_1^{-1} 
		B_1^\prime B_2 [\bar{V}^{-1}]_2 B_2^\prime B_1  \bar{V}_1^{-1},  \nonumber
	\\
		[\bar{V}^{-1}]_{2}
	&	=
		\bar{V}_2^{-1} + 
		\bar{V}_2^{-1} 
		B_2^\prime B_1 [\bar{V}^{-1}]_1 B_1^\prime B_2 \bar{V}_2^{-1}. 
		\label{eq:blk.inv.1.2}
	\end{align}
\end{lemma}
\begin{proof}[Proof of Lemma~\ref{lm:blk.inv.1}]
	First, we show the equivalence in \eqref{eq:blk.inv.1.1} for the top left block, and the remaining blocks of \eqref{eq:blk.inv.1.1}  proceed by the same argument. By applying \eqref{eq:blk.inv.2.1}, we see that 
	\begin{align}
		[\bar{V}^{-1}]_{1}=&
		\left(
		B_1^\prime B_1 + \lambda_a[1-|\phi|] 
		-  
		B_1^\prime B_2
		\left(B_2^\prime B_2 + \lambda_b[1-|\phi|]\right)^{-1} 
		B_2^\prime B_1
		\right)^{-1} 
		\nonumber
		\\
		=&
		\left(
		\lambda_a[1-|\phi|]  
		+  
		B_1^\prime 
		\left[
		I -  
		B_2
		\left(B_2^\prime B_2 + \lambda_b[1-|\phi|]\right)^{-1} 
		B_2^\prime
		\right]
		B_1
		\right)^{-1}
		\nonumber
		\\
		=&_{(1)}
		\left(
		\lambda_a[1-|\phi|]  
		+  
		B_1^\prime 
		\bar{M}_2
		B_1
		\right)^{-1},
		\label{eq:blk.inv.1.3}
	\end{align}
	where $=_{(1)}$ of \eqref{eq:blk.inv.1.3} follow by the definition of $\bar{M}_2$. 
	
	Second, the equivalence for \eqref{eq:blk.inv.1.2} follows by the first two rows of \eqref{eq:blk.inv.2.2}.
\end{proof}

\begin{lemma}[Norm and eigenvalues]
	\label{lm:norm.eig}
	For any matrix $A$, $||A|| \geq |\rho|$ where $\rho$ is an eigenvalue of $A$.
\end{lemma}
\begin{proof}[Proof of Lemma \ref{lm:norm.eig}]
	This follows because for any matrix $A$ with eigenvalue $\rho$ (and associated eigenvector $v$),
	$||A|| \geq ||Av|| = || \rho v|| = |\rho| ||v|| = |\rho|$, where the inequality follows by the definition of the spectral norm as the supremum of $||Ax||$ over all unit-length vectors $x$.
\end{proof}

\begin{lemma}[Taking derivatives]
	\label{lm:take.deriv}
	\begin{equation}
		\frac{\partial \lambda_a^{1/2}}{\partial\tilde{\lambda}_a}
		=
		\frac{1}{2\epsilon}
		\left[1+\lambda_a^{\frac{1}{k}}\right]^{\frac{k}{2}}
		\lambda_a^{\frac{1}{2}-\epsilon}.
	\end{equation}
\end{lemma}
\begin{proof}[Proof of Lemma \ref{lm:take.deriv}]
	First,
	\begin{equation}
		\frac{\partial\lambda_a^{1/2}}{\partial\lambda_a^{-1}}
		=-\frac{1}{2}\lambda_a^{3/2}.
	\end{equation}
	From \eqref{eq:tilde.lambda},
	\begin{align}
		\frac{\partial \lambda_a^{-1}}{\partial \tilde{\lambda}_a}
		=&
		k[\tilde{\lambda}_a^{-n}-1]^{k-1}[-n]\tilde{\lambda}_a^{-n-1}
		=
		-kn\tilde{\lambda}_a^{-n-1}[\tilde{\lambda}_a^{-n}-1]^{k-1}
		=
		-kn\tilde{\lambda}_a^{-n-1}[\lambda_a^{-1}]^{\tfrac{k-1}{k}}
	\end{align}
	so that
	\begin{align}
		\frac{\partial \lambda_a^{1/2}}{\partial\tilde{\lambda}_a}
		=&
		\frac{kn}{2}\cdot \tilde{\lambda}_a^{-n-1} \cdot
		\lambda_a^{\tfrac{3}{2}-\tfrac{k-1}{k}}
		=
		\frac{kn}{2}\cdot 
		[1+\lambda_a^{\tfrac{1}{k}}]^{\tfrac{n+1}{n}} \cdot
		\lambda_a^{\tfrac{3}{2}-\tfrac{k-1}{k}}
		=
		\frac{kn}{2}\cdot 
		[1+\lambda_a^{\tfrac{1}{k}}]^{\tfrac{n+1}{n}} \cdot
		\lambda_a^{\tfrac{3}{2}-\tfrac{1}{kn}}
		\nonumber\\
		=&_{(1)}
		\frac{kn}{2}\cdot 
		[1+\lambda_a^{\tfrac{1}{k}}]^{\tfrac{k}{2}} \cdot
		\lambda_a^{\tfrac{3}{2}-\epsilon}
	\end{align}
	where $=_{(1)}$ follows from our choice of $(k,n)$ such that $[2(n+1)]^{-1}=\epsilon$, for some $\epsilon>0$ that satisfies Assumption~\ref{ass:graph.connect.2},  $k=2(n+1)/n$, and $\epsilon=(kn)^{-1}$. 
\end{proof}

Below we present some auxiliary Lemmas used in the proofs.

\begin{lemma}[Positive Definiteness of $\Lambda^*$]
	\label{lm:pd.var}
	$\Lambda^*=\Lambda^{1/2}[-\phi\mathcal{A}+I]\Lambda^{1/2}$ is positive definite. 
\end{lemma}
\begin{proof}[Proof of Lemma \ref{lm:pd.var}]
	Keep in mind that $|\phi| <1$. Then to verify that the variance is positive definite, we only need to verify that $\mathcal{A}$ has eigenvalues between -1 and 1. This is true by noting that $\mathcal{A}$ is the normalized adjacency matrix of a graph.
\end{proof}

\begin{lemma}[Interpreting the Precision Matrix]
	\label{lm:norm.precision}
	Suppose $M\in\mathbb{R}^{n\times n}$ is the positive definite precision matrix of a random vector $\bs{X}:=(\bs{X}_{1},\ldots,\bs{X}_n)'\in\mathbb{R}^n$. Then, the partial correlation
	\begin{equation*}
		\text{corr}(\bs{X}_i,\bs{X}_j \mid \bs{X}_{-(i,j)}) = 
		-\frac{M_{ij}}{(M_{ii}M_{jj})^{1/2}}.
	\end{equation*}
	where $\bs{X}_{-(i,j)}$ represents all elements of $\bs{X}$ except for the $i^{th}$ and $j^{th}$ elements. 
\end{lemma}

\begin{lemma}[Results for Computing Conditional Distribution of Normal Vectors]
	\label{lm:norm.cond}
	For a matrix 
	\begin{equation*}
		\Sigma 
		:=
		\begin{bmatrix}
			\Sigma_{1} & \Sigma_{12}\\ 
			\Sigma_{21} & \Sigma_{2}
		\end{bmatrix} = 
		\begin{bmatrix}
			M_{1} & M_{12} \\ 
			M_{21} & M_{2}
		\end{bmatrix}^{-1},
	\end{equation*}
	we have 
	$
	\Sigma_{1} - \Sigma_{12}\Sigma_{2}^{-1}\Sigma_{21}
	= 
	M_{1}^{-1} \text{and} \ \ 
	\Sigma_{12}\Sigma_{2}^{-1}
	=
	-M_{1}^{-1}M_{12}.
	$
\end{lemma}
\begin{proof}[Proof of Lemma \ref{lm:norm.cond}]
	For the first equality, apply Lemma \ref{lm:blk.inv.2}. For the second equality, apply Lemma \ref{lm:blk.inv.2} and 	$D^{-1}C(A-BDC)^{-1}
	=
	(D-CA^{-1}B)^{-1}CA^{-1}$ for conformal matrices.
\end{proof}

\begin{lemma}[Block Matrix Inversion II]
	\label{lm:blk.inv.2}
	Let $M$ denote a generic matrix in this lemma. Then
	for non-singular $M_{1},M_{2}$,
	\begin{align}
		M^{-1} =&\begin{bmatrix}
			[M^{-1}]_{1} & [M^{-1}]_{12}  \\ 
			[M^{-1}]_{21}  & [M^{-1}]_{2}
		\end{bmatrix}
		\nonumber
		\\=
		&\begin{bmatrix}
			(M_{1} - M_{12} M_{2}^{-1} M_{21})^{-1}
			&
			-(M_{1} - M_{12} M_{2}^{-1} M_{21} )^{-1} M_{12} M_{2}^{-1}
			\\
			-(M_{2} - M_{21} M_{1}^{-1} M_{12} )^{-1} M_{21}M_{1}^{-1}
			&
			(M_{2} - M_{21} M_{1}^{-1} M_{12} )^{-1} 
		\end{bmatrix}. 
		\label{eq:blk.inv.2.1}
	\end{align}
	Furthermore, 
	\begin{align}
		(M_{1}-M_{12}M_{2}^{-1}M_{21})^{-1} 
		=& 
		M_{1}^{-1}  
		+
		M_{1}^{-1}  M_{12} (M_{2} - M_{21} M_{1}^{-1} M_{12} )^{-1}  M_{21} M_{1}^{-1}
		\nonumber\\
		(M_{2}-M_{21} M_{1}^{-1}M_{12})^{-1} 
		=& 
		M_{2}^{-1}  
		+
		M_{2}^{-1}  M_{21} (M_{1}-M_{12}M_{2}^{-1}M_{121})^{-1}  M_{12} M_{2}^{-1}
		\nonumber\\
		(M_{1}-M_{12}M_{2}^{-1}M_{21})^{-1}  M_{12} M_{2}^{-1}
		=&
		M_{1}^{-1} M_{12} (M_{2}-M_{21} M_{1}^{-1}M_{12})^{-1} .
		\label{eq:blk.inv.2.2}
	\end{align}
\end{lemma}

\begin{proof}[Proof of Lemma \ref{lm:equalities}]
	Part (a). The intuition for this result follows from standard partitioned regression theory: 
	$[\mathcal{R}L^\dagger \mathcal{R}']_2$ is the variance of the $\bs{\beta}$-subvector of the normalized $\hat{\bs{\theta}}^{\text{ls}}=L^- B'\bs{Y}$, which may also be written as $\hat{\bs{\beta}}^{\text{ls}}= M_c L_{2,\perp}^\dagger B_{2,\perp}'\bs{Y}$, whose variance is $M_c[L_{2,\perp}]^\dagger M_c$
	
	More formally, first note  that 
	$\hat{\bs{\theta}}^{\text{ls}}=(\hat{\bs{\alpha}}^{\text{ls}},\hat{\bs{\beta}}^{\text{ls}})$ 
	is the unique minimizer 
	solving 
	\begin{align}
		&\min_{(\bs{\alpha},\bs{\beta}):\bs{1}'\bs{\beta}=0}
		\lVert Y - B_1\bs{\alpha} - B_2\bs{\beta} \rVert^2 \nonumber 
		=
		\min_{\bs{\beta}:\bs{1}'\bs{\beta}=0} \min_{\bs{\alpha}}
		\lVert Y - B_1\bs{\alpha} - B_2\bs{\beta} \rVert^2 \nonumber \\
		=&_{(1)}
		\min_{\bs{\beta}:\bs{1}'\bs{\beta}=0} \min_{\bs{\alpha}}
		\lVert Y - B_1D_1^{-1}B_1'[\bs{Y}-B_2\bs{\beta}] - B_2\bs{\beta} \rVert^2 \nonumber 
		=
		\min_{\bs{\beta}:\bs{1}'\bs{\beta}=0}
		\lVert M_1 Y - M_1 B_2\bs{\beta} \rVert^2 \nonumber \\
		=&
		\min_{\bs{\beta}:\bs{1}'\bs{\beta}=0}
		\lVert Y_\perp - B_{2,\perp}\bs{\beta} \rVert^2,
		\label{eq:E1.1}
	\end{align}
	where $=_{(1)}$ follows because the solution to the inner minimization is $\tilde{\bs{\alpha}}= D_1^{-1}B_1'[\bs{Y}-B_2\bs{\beta}]$, and $D_1:[B_1'B_1]^{-1}$ and $M_1:=I - B_1D_1^{-1}B_1'$.
	Once we have shown that $M_c L_{2,\perp}^\dagger B_{2,\perp}'\bs{Y}$ is the unique solution to the minimization problem in the final equality of \eqref{eq:E1.1}, it follows that $M_c L_{2,\perp}^\dagger B_{2,\perp}'\bs{Y}$ is equivalent to $\hat{\bs{\beta}}^{\text{ls}}$, and thus they must have the same variance. This implies that $M_c L_{2,\perp}^\dagger M_c = [\mathcal{R}L^\dagger \mathcal{R}']_2$, as desired. 
	
	That $M_c L_{2,\perp}^\dagger B_{2,\perp}'\bs{Y}$ is the unique solution to \eqref{eq:E1.1} follows from a similar argument as regarding $\hat{\theta}^{\text{ls}}$ being the unique LS solution satisfying our chosen normalization $\mathcal{R}$. Firstly, $L_{2,\perp}^- :=M_c L_{2,\perp}^\dagger M_c$ is a generalized inverse\footnote{Verifiable by directly checking $L_{2,\perp}L_{2,\perp}^-L_{2,\perp}=L_{2,\perp}$ and $L_{2,\perp}^-L_{2,\perp}L_{2,\perp}^-=L_{2,\perp}^-$.} of $L_{2,\perp}$ because $M_c L_{2,\perp}=L_{2,\perp}$. Intuitively, $L_{2,\perp}$ is the Laplacian of the one-mode projected graph, and hence the off-diagonal entries of each column/row must sum to the diagonal entry of that column/row, or in other words each column/row of $L_{2,\perp}$ is already mean-zero. Because of this, $M_c L_{2,\perp}^\dagger B_{2,\perp}'\bs{Y}=L_{2,\perp}^- \cdot B_{2,\perp}'\bs{Y}$ is a solution to \eqref{eq:E1.1}. Regarding its uniqueness,  
	any least-squares estimate  $\check{\bs{\beta}}^{\text{ls}}$ can be characterized as 
	\begin{equation}
		\check{\bs{\beta}}^{\text{ls}} = L_{2,\perp}^-\cdot B_{2,\perp}'\bs{Y} +  [I-  L_{2,\perp}^-L_{2,\perp}]\bs{w} 
	\end{equation}
	for some $\bs{w}\in\mathbb{R}^{c}$. Thus if $\check{\bs{\beta}}^{\text{ls}}$ is a LS estimate also satisfying the normalization $\bs{1}^\prime \check{\bs{\beta}}^{\text{ls}}=0$, then
	\begin{equation}
		\begin{aligned}[b]
			\check{\bs{\beta}}^{\text{ls}}
			&=_{(1)}
			M_c\check{\bs{\beta}}^{\text{ls}} 
			=
			L_{2,\perp}^- \cdot B_{2,\perp}'\bs{Y}
			+
			[M_c - M_c L_{2,\perp}^-L_{2,\perp}]\bs{w} \\
			&=_{(2)}
			L_{2,\perp}^- \cdot B_{2,\perp}'\bs{Y}
			+
			[M_c - M_c]\bs{w} 
			=
			L_{2,\perp}^- \cdot B_{2,\perp}'\bs{Y},
		\end{aligned}
	\end{equation}
	where $=_{(1)}$ follows because by definition $\check{\bs{\beta}}^{\text{ls}}$ satisfies $\bs{1}^\prime \check{\bs{\beta}}^{\text{ls}}=0$,  and $=_{(2)}$ follows from a similar argument as \eqref{eq:E6}, requiring only that $M_c\bs{\beta}$ is identified. This in turn is guaranteed by the fact that the pairwise differences of $\{\beta_i\}_{i=1}^c$ are identified.  

	Parts (c) and (d).
	These may be directly verified by expanding $ L^-$ and using the defining properties of the Moore-Penrose inverse:
	$
	XX^\dagger X = X,\ X^\dagger X X^\dagger = X^\dagger.
	$
	
	Part (b) holds because $B_1$ and $B_2$ both have each rows sum to $1$ by construction.
	
	Part (e).
	Lemma~\ref{lm:ctd.est} implies that $\bs{\theta}$ is estimable using the equation
	$
	\bs{Y} = B \bs{\theta} + \bs{U},
	$
	and the equality
	$
	\mathcal{R}  = \mathcal{R}  L^- L
	$
	then follows from Lemma \ref{lm:searle}.
\end{proof}

\begin{proof}[Proof of Lemma~\ref{lm:inequalities}]
	
	Part (a).
	Follows from computing the norms of the individual pieces of $\mathcal{R}$, and using $r$ and$ c$ are of the same asymptotic order by Assumption~\ref{ass:reg.cond}.
	
	Part (b). Let $\bar{d}$ denote the largest diagonal value of $L=B'B$. We have
	\begin{equation}
		\begin{aligned}[b]
			||B|| = \bar{\rho} (L) 
			=
			\bar{\rho}(D^{-1/2}LD^{-1/2}\cdot D)
			\leq
			\bar{\rho}(D^{-1/2}LD^{-1/2})
			\bar{\rho}(D) 
			\leq_{(1)}
			2\bar{d}
		\end{aligned}
	\end{equation}
	where $\leq_{(1)}$ follows because $D^{-1/2}LD^{-1/2}$ is the normalized Laplacian of the bipartite graph, and the eigenvalues of a normalized Laplacian lie in the interval $[0,2]$; see e.g. \cite{Jochmans2019}.
	
	To show part (c), we have
	$
	\lVert L  L^-\rVert  
	=
	\lVert L \mathcal{R}  L^\dagger \mathcal{R}^\prime \rVert 
	=_{(1)} 
	\lVert L  L^\dagger \mathcal{R}^\prime \rVert 
	\leq_{(2)} 
	M \cdot \lVert L  L^\dagger \rVert 
	\leq_{(3)} 
	M 
	$,
	where $=_{(1)}$ follows from $B\mathcal{R}=B$, $=_{(2)}$ follows from $||\mathcal{R}||\le M$, and $=_{(3)}$ holds because $L  L^\dagger$ is a projection matrix.
\end{proof}

\begin{lemma}[\cite{Searle1966} Theorem 5]
	\label{lm:searle}
	Consider an estimating equation
$
		\bs{z}=X\bs{b} + \bs{e}, \ \mathbb{V}_{\bs{b}}[\bs{e}]=\sigma^2I.
$
	Let $\bs{q}^\prime\bs{b}$ be a linear combination of $\bs{b}$, $R$ be a generalized inverse of $X^\prime X$, and $P=RX^\prime X$. Then, $\bs{q}^\prime\bs{b}$ is estimable\footnote{$\bs{q}^\prime\bs{b}$ is said to be estimable if there exists a linear function of $\bs{z}$ whose expectation is $\bs{q}^\prime\bs{b}$.} if and only if 
	$\bs{q}^\prime P = \bs{q}^\prime$. 
\end{lemma}
\begin{proof}[Proof of Lemma \ref{lm:searle}]	
	See Theorem 5 of \cite{Searle1966}.
\end{proof}

\begin{lemma}[Estimability and Connectivity]
	\label{lm:ctd.est}
	Connectedness of $\mathcal{G}$ implies estimability of normalized parameters $\bs{\theta}$.
\end{lemma}
\begin{proof}
	Recall that a student $i$ and teacher $j$ are said to share an edge if $\left\{ t:j(i,t)=j\right\} $ is non-empty, where $j(\cdot,\cdot)$ is our matching function. Recall also the definition of connectedness: a pair of vertices $i$ and $j$ are said to be connected if there exists a path between them -- that is, a sequence of unique edges that begins with one of $(i,j)$ and ends with the other. A graph is said to be connected if there exists a path for every pair of distinct vertices. 
	
	Recall that our normalized parameters, defined using $\mathcal{R}$, are $\beta_{j}-\bar{\beta}$ and $\alpha_{i}+\bar{\beta}$ where $\bar{\beta}:=\frac{1}{c}\textstyle\sum_{j=1}^{c}\beta_{j}$. Furthermore, estimability of these parameters is defined as possessing unbiased estimators of them. 
	Now suppose that the bipartite graph is connected. This implies for any two teachers $j\in \mathcal{T}$ and $\tilde{j}\in \mathcal{T}$, we can find a sequence of vertices $j,i_{1},j_{2},i_{2},\dots,j_{K-1},i_{K},\tilde{j}$ for some $K$ such that every adjacent vertices in this sequence share an edge. By taking the pairwise difference of the corresponding observations $y_{it}$, we have an unbiased estimator for $\beta_{\tilde{j}}-\beta_{j}$. A similar logic leads to the conclusion that we have an unbiased estimator of each and every pairwise difference $\alpha_{\tilde{i}}-\alpha_{i}$. 
	We therefore possess an unbiased estimator of $\beta_{j}-\bar{\beta}$ since
	$
	\beta_{j}-\bar{\beta}=\tfrac{1}{c}\Sigma_{l=1}^{c}[\beta_{j}-\beta_{l}].
	$
	The normalized $\alpha_{i}+\bar{\beta}$ is also estimable since
	$
	\alpha_{i}+\bar{\beta}=\tfrac{1}{c}\Sigma_{j=1}^{c}[\alpha_{i}+\beta_{j}]=\tfrac{1}{c}\Sigma_{j=1}^{c}[(\alpha_{i}+\beta_{j(i,1)})+(\beta_{j}-\beta_{j(i,1)})].
	$
	The observation $y_{it}$ for $t=1$ is an unbiased estimator of the
	first term in the summation, and the second term is simply a pairwise
	difference. 
\end{proof}

\section{Computation of EB Estimators for Large Datasets}
\label{sec:compute}
\subsection{URE}
\label{subsec:compute.ure}
This section explains a computational difficulty of implementing the URE-based shrinkage method for large datasets and proposes a solution. We abstract from the covariates, or alternatively assume that the covariates' coefficients are known. 

Using the relation $S+S_1=\mathcal{R}$, one can show that the URE of \eqref{eq:ure} is 
\begin{equation}
\text{URE}(\bs{\lambda})
\propto
2\sigma^2\tr\left[W^{1/2}\mathcal{R}[L + \Lambda^*]^{-1} \mathcal{R}'W^{1/2} \right] + \left[\bs{v}-\hat{\bs{\theta}}^{\text{ls}}\right]^\prime S^\prime W S \left[\bs{v}-\hat{\bs{\theta}}^{\text{ls}}\right].
\end{equation}
The URE-based approach seeks to minimize $\text{URE}(\bs{\lambda})$ w.r.t. the hyperparameters $\bs{\lambda}$, commonly done through grid search ($\mu$ can be concentrated out using its first-order condition,as we show below). A difficulty that is immediate from the form of the URE is that the evaluation of the trace for each $\bs{\lambda}$ necessitates the inversion of the $(r+c)\times (r+c)$ matrix $[L + \Lambda^*]^{-1}$; this quickly becomes intractable as the dimensionality of the problem gets large. 

To avoid the inversion, we use Hutchinson's method for computing 
$
\tr\left[M'[L + \Lambda^*]^{-1} M \right]
$
for $M:=\mathcal{R}'W^{1/2}$ which proceeds as follows. Let $\bs{z}\sim\mathcal{N}[\bs{0},I]$, and note that 
\begin{align}
\tr[ M' [L + \Lambda^*]^{-1} M ]
=&
\mathbb{E}^{\bs{z}} (M\bs{z})' [L + \Lambda^*]^{-1} (M\bs{z})
\approx
\frac{1}{J} \sum_{j=1}^J (M\bs{z}^{(j)})' [L + \Lambda^*]^{-1} (M\bs{z}^{(j)})  
\end{align}
where we draw $\{ \bs{z}^{(j)} \}_{j=1}^J$ i.i.d. from $\mathcal{N}[\bs{0},I]$. 
In order to concentrate $\mu$ out of $\text{URE}(\bs{\lambda})$ to reduce the dimensionality of the grid, note that the optimal $\mu$ implied by its first-order condition is 
\begin{equation}
\mu^{\text{ure}} = \frac{\tilde{\bs{v}}^\prime S^\prime W S \hat{\bs{\theta}}^{\text{ls}}}{\tilde{\bs{v}}^\prime S^\prime W S \tilde{\bs{v}}}
\label{eq:mu.concen}
\end{equation}
where $\tilde{\bs{v}}:=[\bs{1}_r^\prime\quad\bs{0}_c^\prime]^\prime$.
Concentrating out $\mu$ from the oracle loss 
$[\hat{\bs{\theta}}-\bs{\theta}]'W[\hat{\bs{\theta}}-\bs{\theta}]$ 
can be approached similarly:
\begin{equation}
\mu^{\text{ol}} = \frac{\tilde{\bs{v}}^\prime S^\prime W[\hat{\bs{\theta}}^{\text{ls}}-\bs{\theta}-S\hat{\bs{\theta}}^{\text{ls}}]}{\tilde{\bs{v}}^\prime S^\prime W S \tilde{\bs{v}}}.
\end{equation}

%

\subsection{MLE}
\label{subsec:compute.mle}
The MLE selects $\bs{\lambda}$ to maximize $f(\bs{Y})$, the marginal likelihood of $\bs{Y}$, and $\bs{Y} \mid \bs{\theta} \sim \mathcal{N}[B\bs{\theta},\sigma^2I_n]$ and $\bs{\theta} \sim \mathcal{N}[\bs{v},\Sigma]$ where $\Sigma:=\Gamma^{-1}\sigma^2$ and $\Gamma$ is defined as the variance matrix in \eqref{eq:sorting.prior}. Now 
\begin{align}
&f(\bs{Y})
=\int f(\bs{Y} \mid \bs{\theta}) dP(\bs{\theta}) 
= \mathbb{E}^{\bs{\theta}} [f(\bs{Y} \mid \bs{\theta})]  \nonumber\\
=& \det(\sigma^2I_n)^{-1/2} \mathbb{E}^{\bs{\theta}} 
\exp\left( -\tfrac{1}{2\sigma^2} [\bs{Y}-B\bs{\theta}]'[\bs{Y}-B\bs{\theta}]  \right) \\
=&
\det(\sigma^2I_n)^{-1/2} \mathbb{E}^{\bs{\theta}} 
\exp\left( -\tfrac{1}{2\sigma^2} \left\{
[L^\dagger B'\bs{Y}-\bs{\theta}]' L [L^\dagger B'\bs{Y}-\bs{\theta}]
- (B'\bs{Y})'L^\dagger B'\bs{Y} + \bs{Y}'\bs{Y}
\right\} \right) 
\nonumber \\ 
=&
\det(\sigma^2I_n)^{-1/2} 
\exp\left( -\tfrac{1}{2\sigma^2} \left\{
\bs{Y}'\bs{Y} - (B'\bs{Y})'L^\dagger B'\bs{Y}
\right\} \right) 
\mathbb{E}^{\bs{\theta}} 
\exp\left( -\tfrac{1}{2\sigma^2}
[L^\dagger B'\bs{Y}-\bs{\theta}]' L [L^\dagger B'\bs{Y}-\bs{\theta}]
\right) , \nonumber
\end{align}
where the second-to-last equality uses the identity $LL^\dagger B' = B'$ to verify that 
\begin{equation}
[\bs{Y}-B\bs{\theta}]'[\bs{Y}-B\bs{\theta}] -\bs{Y}'\bs{Y}
=
[L^\dagger B'\bs{Y}-\bs{\theta}]' L [L^\dagger B'\bs{Y}-\bs{\theta}]
- [B'\bs{Y}]'L^\dagger B'\bs{Y}.
\end{equation}
Furthermore, defining $\bs{x}:=L^\dagger B'\bs{Y}-\bs{v}$, and using the prior $\bs{\theta}\sim\mathcal{N}[\bs{v},\Sigma]$ together with the moment generating function of a Gaussian quadratic form, we have
\begin{align}
&\mathbb{E}_{\bs{\theta}} 
\exp\left( 
\left[L^\dagger B'\bs{Y}-\bs{\theta}\right]' 
\left[-\tfrac{1}{2\sigma^2} L\right] 
\left[L^\dagger B'\bs{Y}-\bs{\theta}\right]
\right) \nonumber\\
&	=
\det\left(I-2\left[-\tfrac{1}{2\sigma^2}L\right]\Sigma\right)^{-1/2} 
\cdot
\exp\left(
-\tfrac{1}{2}\bs{x}'
\left[I-(I-2\left[-\tfrac{1}{2\sigma^2}L\right]\Sigma)^{-1}\right]
\Sigma^{-1} \bs{x}\right) \nonumber
\\
&	=
\det\left(I+L\tfrac{1}{\sigma^2}\Sigma\right)^{-1/2} 
\cdot
\exp\left(
-\tfrac{1}{2\sigma^2}\bs{x}'
\left[I-(I+L\tfrac{1}{\sigma^2}\Sigma)^{-1}\right]
\Sigma^{-1} \bs{x}\right), \ \text{where}
\end{align}
\begin{align}
\det\left(I+L\tfrac{1}{\sigma^2}\Sigma\right)^{-1/2} 
=&
\det\left(\Sigma^{-1}\sigma^2+L\right)^{-1/2} 
\det\left(\Sigma\tfrac{1}{\sigma^2}\right)^{-1/2}  \nonumber\\
=& 
\det\left(\Sigma^{-1}\sigma^2+L\right)^{-1/2} 
\det\left(\Sigma^{-1}\sigma^2\right)^{1/2} \\
=&
\det\left(\Gamma+L\right)^{-1/2} 
\det\left(\Gamma\right)^{1/2}, \nonumber\\
\exp\left(
-\tfrac{1}{2}\bs{x}'
\left[
I - (I+L\tfrac{1}{\sigma^2}\Sigma)^{-1}
\right]
\Sigma^{-1} \bs{x}\right)
=&
\exp\left(
-\tfrac{1}{2}\bs{x}'
\left[
I - \sigma^2\Sigma^{-1}(\sigma^2\Sigma^{-1}+L)^{-1}
\right]
\Sigma^{-1} \bs{x}\right) \nonumber\\
=&
\exp\left(
-\tfrac{1}{2\sigma^2}\bs{x}'
\left[
I - \Gamma(\Gamma+L)^{-1}
\right]
\Gamma\bs{x}\right)  \nonumber \\
=&
\exp\left(
-\tfrac{1}{2\sigma^2}
\left[(\Gamma+L)^{-1}L\bs{x}\right]'
\Gamma\bs{x}\right).
\end{align}
In sum, we have 
\begin{align}
&-2\sigma^2 \left[\log f(\bs{Y}) - \log\det(\sigma^2I_n)^{-1/2}\right]
\nonumber \\
=&
\sigma^2\left[\log\det(\Gamma+L) - \log\det(\Gamma) \right]
+ \bs{Y}'\bs{Y} - (B'\bs{Y})'L^\dagger B'\bs{Y}
+ \left[(\Gamma+L)^{-1}L\bs{x}\right]'\Gamma\bs{x}
\nonumber \\
\propto&
\sigma^2\left[\log\det(\Gamma+L) - \log\det(\Gamma) \right]
+ \left[(\Gamma+L)^{-1}L\bs{x}\right]'\Gamma\bs{x},
\end{align}
where we want to minimize the final expression with respect to $\bs{\lambda}$.
The minimizer $\mu^{\text{mle}}$ can be concentrated out by using its first-order condition:
\begin{align}
\mu^{\text{mle}}
=&
\frac{\tilde{\bs{v}}' \left[ L(\Gamma+L)^{-1}\Gamma + \Gamma(\Gamma+L)^{-1}L \right] L^\dagger B'\bs{Y}}{\tilde{\bs{v}}' \left[ L(\Gamma+L)^{-1}\Gamma + \Gamma(\Gamma+L)^{-1}L \right] \tilde{\bs{v}}}
\nonumber \\
=&
\frac{
	\tilde{\bs{v}}' L(\Gamma+L)^{-1}\Gamma L^\dagger B'\bs{Y} + 
	\tilde{\bs{v}}' \Gamma(\Gamma+L)^{-1} B'\bs{Y}
}{
	2\tilde{\bs{v}}' L(\Gamma+L)^{-1}\Gamma \tilde{\bs{v}}
}.
\end{align}

\section{Further Details on the Empirical Analysis}
\label{sec:details.empirics}

\subsection{EB-URE Estimation Results}
\label{sec:details.empirics.ure}

The URE-coefficients on 
\begin{equation}
	constant,\ econ\ disad,\ eng\ learner,\ female,\ black,\ hispanic,\ white, others
\end{equation}
for the two-way fixed effects model in Section~\ref{subsec:appl.ure} are
\begin{equation}
(2.2,-1.28,-0.39,-0.033,-1.77,-1.00,-1.42,-0.83).
\end{equation}
	
\subsection{Computing \cite{Verdier2020}'s estimator}
\label{sec:details.empirics.verdier}
We need an estimate of 
$\mathbb{E}[y_{it} | \left\{y_{is}\right\}_i]$ for $s\le t$, 
when implementing \cite{Verdier2020}'s estimator. Following \cite{Verdier2020}, we estimate this conditional expectation using the predicted values from the following regression:
$
	y_{it} = \phi_1y_{it-1} + \phi_2y_{it-2} + e_{it}.
$

\subsection{LS and Shrinkage in a 1-way Model}
\label{sec:details.empirics.1way}

The standard construction of the 1-way shrinkage estimator $\hat{\beta}_j^{\text{1-way}}$ begins by assuming that $\alpha_i$ of $y_{it} = \alpha_i + \beta_{j(i,t)} + \rho y_{it-1} + u_{it}$ is determined as $\alpha_i=\bs{x}_i^\prime \bs{\gamma}$. 
Because the coefficients $\bs{\gamma}$ and $\rho$ are consistently estimable, we follow the literature in using the within-teacher fixed effects estimator of $\rho$ and $\bs{\gamma}$ as plug-ins\footnote{See \cite{Chetty2014} and \cite{Gilraine2020}, etc.} and construct the one-way fixed effects estimates of $\beta_j$. More precisely, we first construct the LS estimates of the teacher fixed effects as
\begin{equation}
	\tilde{\beta}_j^{\text{ls}} := 
	\frac{1}{d_{b,j}} \sum_{(i,t):j(i,t)=j} 
	( y_{it} - \tilde{\rho} y_{it-1} - \bs{x}_i^\prime \tilde{\bs{\gamma}}^{\text{ls}} )
\end{equation}
where $d_{b,j}:=| \left\{(i,t):j(i,t)=j\right\} |$ is the number of matches for teacher $j$. The one-way shrinkage estimator of $\bs{\beta}$ is then constructed as
\begin{equation}
	\hat{\beta}_j^{\text{1-way}}
	:=
	\frac{\tilde{\sigma}_{\beta}^2}{\tilde{\sigma}_{\beta}^2 + \tilde{\sigma}^2/d_{b,j}}
	\tilde{\beta}_j^{\text{ls}},
\end{equation}
where 
$\tilde{\sigma}_\beta^2 := \text{var}(\tilde{\beta}_j^{\text{ls}})-\tilde{\sigma}^2/\bar{d_b}$ 
and $\bar{d_b}$ is the empirical mean of $\{d_{b,j}\}_{j=1}^c$.
The estimate of $\sigma^2$ from the one-way estimation is $\tilde{\sigma}^2 = 0.216$, and the implied variance of $\beta_j$ is $\tilde{\sigma}_{\beta}^2=0.049$.

\section{Additional Monte Carlo Simulations}
\label{sec:details.MonteCarlo}

This section considers four additional simulation designs. 
The first -- labeled as 1 ($W_{a+b}$) -- 
uses the same DGP as Design 1 of Section~\ref{sec:simul}, but using $W=W_{a+b}$ that includes the  $\bs{\alpha}$ parameters. We exclude the EB-1way from consideration under this design since it is not for the estimation of $\bs{\alpha}$. The next three designs then revert back to the original RMSE \eqref{eq:simul.rmse} that evaluates only estimation of $\bs{\beta}$, but considers alternative distributions beyond Gaussianity.

\begin{table}[h]
	\caption{Simulation Designs}
	\label{tab:app.mc.designs}
	\begin{center}
		\begin{tabular}{lcccc} 
			\toprule
			& \multicolumn{4}{c}{ Simulation Designs } \\
			\cmidrule(lr){2-5}
			& 1 ($W_{a+b}$) & $\beta\sim \text{T}$  & $\beta\sim\text{Exp}$ & $\alpha\sim\text{Exp}$  \\ 
			\midrule
			DGP Parameters \\
			$\pi_{\text{sort}}$ 
			&0.4 & 0.4 & 0.4 & 0.4 \\
			$\pi_{\text{mob}}$ 
			& 0.043 & 0.043 & 0.043 & 0.043 \\
			$\alpha\sim$ 
			& $\mathcal{N}[0,0.6]$ & $\mathcal{N}[0,0.6]$ & $\mathcal{N}[0,0.6]$ & $\text{Exp}(\sqrt{0.6}) $ \\
			$\beta\sim$ 
			& $\mathcal{N}[0,0.06]$ & $\sqrt{0.06\tfrac{1}{3}}\cdot\text{T}(3)$ & $\text{Exp}(\sqrt{0.06}) $ & $\mathcal{N}[0,0.06]$ \\
			\cmidrule(lr){1-5}
			Empirical moments  \\
			$\text{var}(\alpha_i)$  
			& 0.6 & 0.6 & 0.6  & 0.6  \\
			$\text{var}(\beta_{j(i,t)})$ 
			& 0.06 & 0.06  & 0.06  & 0.06 \\
			$\text{cor}(\alpha_i,\beta_{j(i,t)})$  
			& 0.16 & 0.14  &  0.13 & 0.13 \\
			\cmidrule(lr){1-5}
			Empirical moments  \\
			$\text{var}(\hat{\alpha}_i^{\text{ls}})$  
			& 0.77 & 0.77 & 0.77  & 0.78 \\
			$\text{var}(\hat{\beta}_{j(i,t)}^{\text{ls}})$ 
			& 0.18 & 0.18  & 0.18  &  0.18 \\
			$\text{cor}(\hat{\alpha}_i^{\text{ls}},\hat{\beta}_{j(i,t)}^{\text{ls}})$  
			& -0.23 & -0.24  & -0.22  &   -0.23 \\
			\cmidrule(lr){1-5}
			Connectivity \\ 
			$1E5 \cdot \rho_{\text{min}}$ 
			& 50 & 50 &  52 & 49\\
			\bottomrule
		\end{tabular}
	\end{center}
	{\footnotesize {\em Notes:} 
		The symbol $\text{Exp}(s)$ denotes an exponential variable with scale $s$, and $\text{T}(v)$ denotes a student-$t$ variable with degrees of freedom $v$. 
		The moments and connectivity values displayed for the designs are median values across the simulation rounds. 
		The sampling variance $\sigma^2$ is set as $0.12$ for all four designs, consistent with the value estimated from the empirical application. The variable $\rho_{\text{min}}$ denotes the smallest non-zero eigenvalue of the \emph{normalized} Laplacian projected teacher graph; we choose to match these instead of the unnormalized ones because the former always lie within $[0,2]$ and as a result provide a more natural scale to match.  
	}\setlength{\baselineskip}{4mm}
\end{table}

\begin{figure}[h]
	\caption{RMSE of Estimators}
	\label{fig:app.mc.rmse}
	\begin{center}
		\begin{tabular}{cc}
			Design 1 ($W_{a+b}$) & Design $\beta\sim\text{T}$ \\ 
			\includegraphics[width=.35\textwidth]{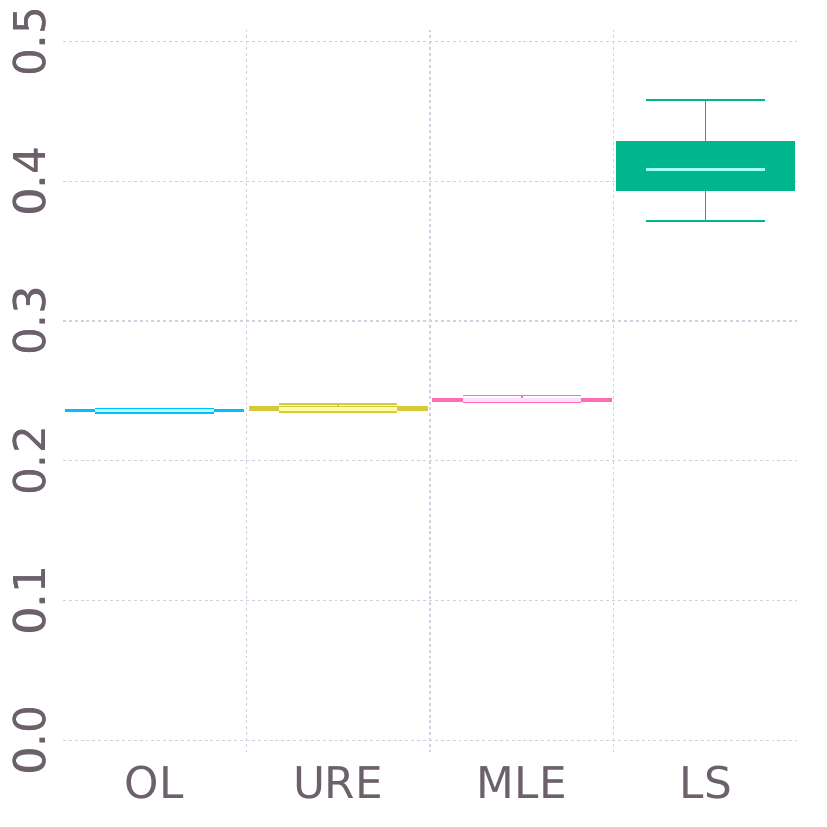} &
			\includegraphics[width=.35\textwidth]{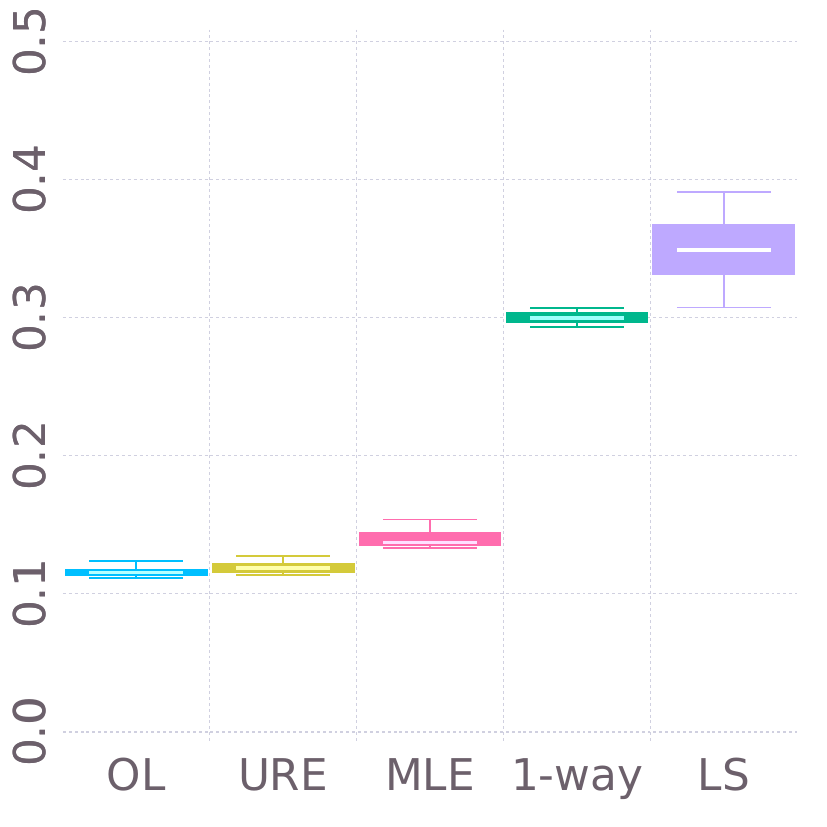} \\
			Design $\beta\sim\text{Exp}$ & Design $\alpha\sim\text{Exp}$ \\ 
			\includegraphics[width=.35\textwidth]{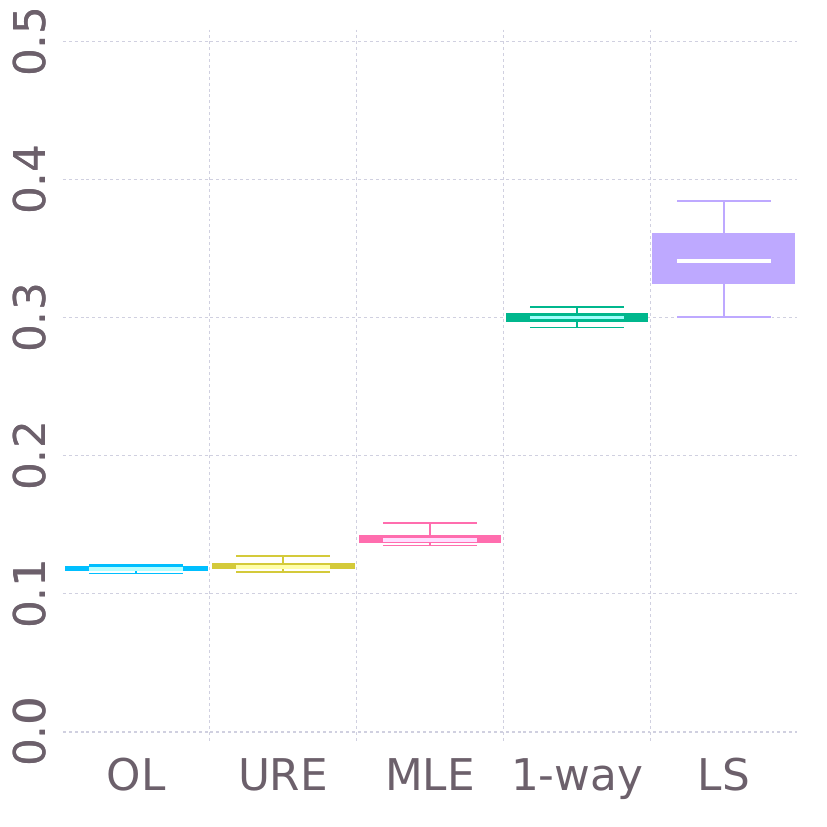} &
			\includegraphics[width=.35\textwidth]{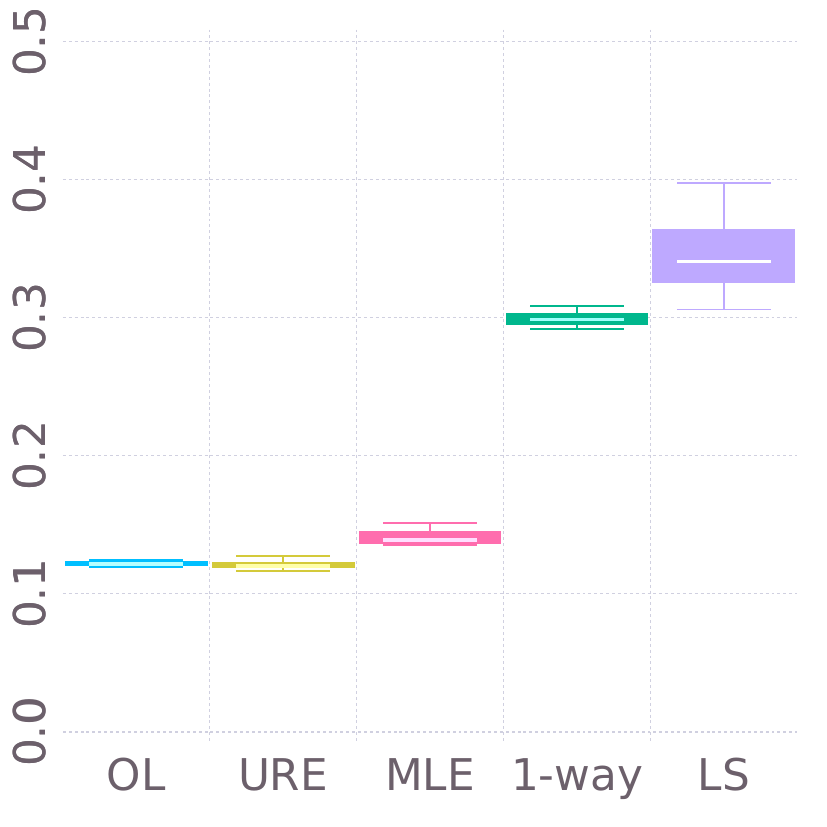} \\
		\end{tabular}   
	\end{center}
	{\footnotesize {\em Notes:} The box plots characterize the distribution of RMSEs across Monte Carlo repetitions.}\setlength{\baselineskip}{4mm}
\end{figure}

\begin{table}[h]
	\caption{Hyperparameter Selection}
	\label{tab:app.mc.hyperpara}
	\begin{center}
		\begin{tabular}{lcccccccccc} 
			\toprule
			& \multicolumn{3}{c}{ EB-OL} & \multicolumn{3}{c}{ EB-URE} &
			\multicolumn{3}{c}{ EB-MLE}
			& \multicolumn{1}{c}{ 1-Way} \\
			\cmidrule(lr){2-4} \cmidrule(lr){5-7} \cmidrule(lr){8-10}
			\cmidrule(lr){11-11} 
			& $\lambda_a$ & $\lambda_b$ & $\phi$ 
			& $\lambda_a$ & $\lambda_b$ & $\phi$ 
			& $\lambda_a$ & $\lambda_b$ & $\phi$ 
			& $\lambda_b$ \\ 
			\midrule
			Design\\
			\;1 ($W_{a+b}$)
			&  0.2 & 5.0 & 0.1 & 0.16 & 4.2 & 0.1 & 0.2 & 3.3 & 0.89 & -  \\ 
			\;$\beta\sim \text{T}$
			&  0.056 & 2.3 & 0.43 & 0.058 & 2.18 & 0.40 & 0.20 & 2.95 & 0.84 & 2.9 \\ 
			\;$\beta\sim \text{Exp}$
			&  0.050 & 2.25 & 0.51 & 0.064 & 2.18 & 0.40 & 0.20 & 3.1 & 0.85 & 2.85 \\ 
			\;$\alpha\sim \text{Exp}$
			& 0.047 & 2.4 & 0.67 & 0.059 & 2.19 & 0.42 & 0.22 & 3.17 & 0.87 & 2.9 \\ 
			\bottomrule
		\end{tabular}
	\end{center}
	{\footnotesize {\em Notes:} The values displayed are the medians across simulation rounds.}\setlength{\baselineskip}{4mm}
\end{table}

\begin{table}[t!]
	\caption{Empirical Moments}
	\label{tab:appmc.mom}
	\begin{center}
		\begin{tabular}{lccccccccc} 
			\toprule
			& \multicolumn{3}{c}{ True} & 
			\multicolumn{3}{c}{ EB-URE} &
			\multicolumn{3}{c}{ EB-MLE} \\
			\cmidrule(lr){2-4} \cmidrule(lr){5-7}
			\cmidrule(lr){8-10}
			& $\text{v}_\alpha$ & $\text{v}_\beta$ & $\rho_{\alpha,\beta}$ 
			& $\text{v}_\alpha$ & $\text{v}_\beta$ & $\rho_{\alpha,\beta}$   
			& $\text{v}_\alpha$ & $\text{v}_\beta$ & $\rho_{\alpha,\beta}$ 
			\\ 
			\midrule
			Design\\
			\; \;1 ($W_{a+b}$)& 0.6 & 0.06 & 0.16 & 0.54 & 0.036 & 0.24 & 0.52 & 0.057 & 0.29 \\ 
			\; $\beta\sim\text{T}$ & 0.6 & 0.058 & 0.14 & 0.62 & 0.047 & 0.18 & 0.51 & 0.056 & 0.28 \\ 
			\; $\beta\sim\text{Exp}$ & 0.6 & 0.06 & 0.13 & 0.62 & 0.047 & 0.18 & 0.51 & 0.056 & 0.28 \\ 
			\; $\alpha\sim\text{Exp}$ & 0.6 & 0.06 & 0.13 & 0.62 & 0.047 & 0.18 & 0.52 & 0.057 & 0.28 \\ 
			\bottomrule
		\end{tabular}
	\end{center}
	{\footnotesize {\em Notes:} $(\text{v}_\alpha,\text{v}_\beta,\rho_{\alpha,\beta})$ represent $\text{var}(\hat{\alpha}_i)$, $\text{var}(\hat{\beta}_{j(i,t)})$ and $\text{cor}(\hat{\alpha}_i,\hat{\beta}_{j(i,t)})$ respectively. The columns for ``True'' use the true values of $(\alpha_j,\beta_{j(i,t)})$ in the computation and are thus the transpose of the second row block of Table~\ref{tab:mc.designs}.}\setlength{\baselineskip}{4mm}
\end{table}

\clearpage

\section{Definitions and Notations}
\noindent Graph-theoretic Objects:
\begin{align}
	L &:= B'B 
	&
	L_{2\perp} &:= B_{2,\perp}'B_{2,\perp} 
	&
	B_{2,\perp} &:= [I-B_1(B_1^\prime B_1)^{-1}B_1^\prime] B_2
	\nonumber\\
	A &:= L-D
	&
	D & := \text{diag}(L) 
	& 
	\mathcal{A} &:= D^{-1/2} A D^{-1/2} 
	\nonumber\\
	\bar{d} &:= \max(D) & L^- &:= \mathcal{R}L^\dagger \mathcal{R}'
	&
	\mathcal{R} &:= \begin{bmatrix}
		I_r & \frac{1}{c}1_{r\times c} \\
		0_{c\times r} & I_c-\frac{1}{c}1_{c\times c}
	\end{bmatrix}
\end{align}
\noindent Hyperparameter Objects:
\begin{align}
	\bs{\lambda}
	&:=
	(\mu,\lambda_a,\lambda_b,\phi)
	& 
	\Lambda &:= \begin{bmatrix}
		\lambda_aI_r & 0_{r\times c} \\ 
		0_{c\times r} & \lambda_bI_c 
	\end{bmatrix}
	&
	\Lambda^* &:= \Lambda^{1/2}[-\phi\mathcal{A}+I]\Lambda^{1/2}
	\nonumber\\
	\bs{v} &:= \begin{bmatrix}
		\mu\bs{1}_r' & \bs{0}_c'
	\end{bmatrix}'
	&
	\tilde{\bs{v}} &:= \begin{bmatrix}
		\bs{1}_r' & \bs{0}_c'
	\end{bmatrix}'
	&
	\mathcal{J} &:= (-\bar{\mu},\bar{\mu})\times \mathbb{R}_{++}^2 \times (-\bar{\phi},\bar{\phi})
	\nonumber\\
	S&:= [L + \Lambda^*]^{-1} \Lambda^* 
	&
	S_1 &:= [L + \Lambda^*]^{-1} L
\end{align}

\noindent Others:
\begin{align}
	W_{a+b} &:= \tfrac{1}{r+c}I_{r+c}
	&
	W_a &:= \begin{bmatrix}
		\tfrac{1}{r}I_r & 0_{r\times c} \\
		0_{c\times r} & 0_{c\times c}
	\end{bmatrix}
	&
	W_b &:= \begin{bmatrix}
		0_{c\times c} & 0_{r\times c} \\
		0_{c\times r} & \tfrac{1}{c}I_c
	\end{bmatrix} 
	\nonumber\\
	W_{\mathcal{R}} &:= \mathcal{R}'W\mathcal{R}
\end{align}
Finally, let $\rho_1(M)\leq\rho_2(M)\leq\dots\leq\rho_n(M)$ denote the ordered eigenvalues of a matrix $M\in\mathbb{R}^{n\times n}$. Let $\bar{\rho}(M)$ denote its \emph{spectral radius}, the largest of the magnitudes of its $n$ eigenvalues. Let $\lVert M \rVert$ denote its \emph{spectral norm}, its largest singular value. When $M$ is positive semi-definite, then $\rho_n(M)=\bar{\rho}(M) = ||M||$.
We use $a\wedge b$ to denote $\min(a,b)$.

\end{appendix}

\end{document}